\documentclass[submission,copyright,creativecommons]{eptcs}
\providecommand{\event}{QAPL 2017} 

\usepackage{amsmath}
\newcommand\numberthis{\addtocounter{equation}{1}\tag{\theequation}}
\usepackage{amssymb}
\usepackage{amsthm}
\usepackage{verbatim} 
\usepackage{paralist}
\usepackage{cite}
\usepackage{appendix}
\usepackage{url}
\usepackage[]{units}
\usepackage{tikz}
\usetikzlibrary{calc,arrows,positioning}
\usepackage{array}
\usepackage{relsize}
\usepackage{wasysym}
\allowdisplaybreaks

\newif\ifdraft\drafttrue

\DeclareFontFamily{U}{mathb}{\hyphenchar\font45}
\DeclareFontShape{U}{mathb}{m}{n}{
      <5> <6> <7> <8> <9> <10> gen * mathb
      <10.95> mathb10 <12> <14.4> <17.28> <20.74> <24.88> mathb12
}{}
\DeclareSymbolFont{mathb}{U}{mathb}{m}{n}
\DeclareMathSymbol{\sqdoublecup} {2}{mathb}{"5F} 
\DeclareMathSymbol{\boxplus} {2}{mathb}{"60} 

\newcommand{\trans}[1][]{\xrightarrow{\, {#1} \, }}
\newcommand{\ntrans}[1][]{\mathrel{{\trans[#1]}\makebox[0em][r]{$\not$\hspace{2ex}}}{\!}}

\DeclareMathOperator{\Kantorovich}{\mathbf{K}}
\DeclareMathOperator{\Hausdorff}{\mathbf{H}}

\newcommand{\support}{\mathsf{supp}}

\newcommand{\rel}{\,{\mathcal R}\,}
\newcommand{\reldist}{\,{\mathcal R}^{\dagger}\,}

\newcommand{\Act}{\mathcal A}

\DeclareMathOperator{\logic}{\mathcal L}

\newcommand{\depth}[1]{\mathrm{dpt}(#1)}
\newcommand{\init}[1]{\mathrm{init}(#1)}
\newcommand{\der}[1]{\mathrm{der}(#1)}

\newcommand{\diam}[1]{\langle #1 \rangle}

\newcommand{\proc}{\mathbf{S}}
\newcommand{\ProbDist}[1]{\Delta(#1)}

\newcommand{\powset}[1]{{\mathcal P}(#1)}

\DeclareMathOperator{\LL}{\mathbb L}
\DeclareMathOperator{\LLd}{\mathbb{L}^{\mathrm{d}}}
\DeclareMathOperator{\LLt}{\mathbb{L}^{\mathrm{t}}}

\DeclareMathOperator{\LLw}{\mathbb {L}_{\mathrm{w}}}
\DeclareMathOperator{\LLwd}{\mathbb{L}^{\mathrm{d}}_{\mathrm{w}}}
\DeclareMathOperator{\LLwt}{\mathbb{L}^{\mathrm{t}}_{\mathrm{w}}}

\newcommand{\Z}{\mathcal Z}
\newcommand{\corr}[1]{\mathrm{corr}_{#1}}
\newcommand{\res}{\mathrm{Res}}

\newcommand{\ctrans}[1][]{\stackrel{#1}{\mathlarger{\twoheadrightarrow}}}

\newcommand{\STr}{\approx_{\mathrm{st}}}
\newcommand{\WTr}{\approx_{\mathrm{wt}}}

\newcommand{\TraceMetric}{\mathbf{d}_{T}}

\newcommand{\dw}{d_T^{\mathrm{w}}}
\newcommand{\Dw}{D_T^{\mathrm{w}}}
\newcommand{\wTraceMetric}{\mathbf{d}_T^{\mathrm{w}}}

\newcommand{\Dtrtd}{\mathcal{D}^{\mathrm{t}}_{\LL}}
\newcommand{\Dtrdd}{\mathcal{D}^{\mathrm{d}}_{\LL}}
\newcommand{\DLLd}{\mathcal{D}_{\LL}}

\newcommand{\Dtrtdw}{\mathcal{D}^{\mathrm{t}}_{\LLw}}
\newcommand{\Dtrddw}{\mathcal{D}^{\mathrm{d}}_{\LLw}}
\newcommand{\DLLdw}{\mathcal{D}_{\LLw}}

\newcommand{\pr}{\mathrm{Pr}}
\newcommand{\tr}{\mathrm{Tr}}
\newcommand{\trw}{\mathrm{Tr}_{\mathrm{w}}}
\newcommand{\add}{\mathrm{Add}}

\newcommand{\argmin}{\mathrm{argmin}}

\newcommand{\nihl}{\mathrm{nil}}

\newcommand{\Acttau}{\Act_{\tau}}

\newcommand{\eqtrace}{\equiv_{\mathrm{w}}}

\newcommand{\Ct}{\mathcal{C}^{\mathrm{t}}}
\newcommand{\Cw}{\mathcal{C}^{\mathrm{w}}}

\newcommand{\Psiw}{\Psi^{\mathrm{w}}}

\newcommand{\TD}{\mathcal{T}}
\newcommand{\TDw}{\mathcal{T}^{\mathrm{w}}}

\newcommand{\val}[2]{[#1](#2)}

\newcommand{\A}{\mathfrak a}

\newcommand{\C}{\mathcal C}

\newcommand{\e}{\mathfrak{e}}

\newcommand{\W}{\mathfrak{W}}
\newcommand{\w}{\mathfrak{w}}

\newcommand{\N}{\mathbb{N}} 

\newenvironment{apx-proof}[1] 
        {\noindent \rm \textbf{Proof of #1.}} 
        {\qed}

\definecolor{lightblue}{RGB}{224,224,255}
\definecolor{lightred}{RGB}{255,224,224}
\definecolor{lightgreen}{RGB}{224,255,224}
\definecolor{lightyellow}{RGB}{255,255,224}
\definecolor{lightpurple}{RGB}{255,224,255}
\definecolor{darkerred}{RGB}{64,0,0}
\definecolor{darkred}{RGB}{128,0,0}
\definecolor{darkblue}{RGB}{0,0,128}
\definecolor{darkgreen}{RGB}{0,128,0}
\definecolor{darkpurple}{RGB}{128,0,128}
\definecolor{grey}{rgb}{0.745098,0.745098,0.745098}
\definecolor{lightgrey}{rgb}{0.9,0.9,0.9}
\definecolor{darkgrey}{rgb}{0.6,0.6,0.6}

\newcommand{\colorpar}[3]{\colorbox{#1}{\parbox{#2}{#3}}}

\newcommand{\marginremark}[3]{\marginpar{\colorpar{#2}{\linewidth}{\color{#1}#3}}}

\makeatletter
\def\THICKhrulefill{\leavevmode \leaders \hrule height 5pt\hfill \kern \z@}
\makeatother

\newcommand{\remarkDG}[1]{\marginremark{darkred}{lightred}{\tiny{[DG]~ #1}}}
\newcommand{\remarkST}[1]{\marginremark{darkgreen}{lightgreen}{\tiny{[ST]~ #1}}}
\newcommand{\remarkVC}[1]{\marginremark{darkblue}{lightyellow}{\tiny{[VC]~ #1}}}

\ifdraft
\else
\renewcommand{\remarkDG}[1]{}
\newcommand{\remarkST}[1]{}
\newcommand{\remarkVC}[1]{}

\fi

\newtheorem{theorem}{Theorem}
\newtheorem{proposition}{Proposition}
\newtheorem{lemma}{Lemma}
\newtheorem{corollary}{Corollary}

\theoremstyle{definition}
\newtheorem{definition}{Definition}
\newtheorem{example}{\emph{Example}}

\theoremstyle{remark}
\newtheorem{remark}{Remark}

\title{Logical Characterization of Trace Metrics}
\author{Valentina Castiglioni
\institute{University of Insubria (IT)}
\email{v.castiglioni2@uninsubria.it}
\and
Simone Tini
\institute{University of Insubria (IT)}
\email{simone.tini@uninsubria.it}
}

\def\titlerunning{Logical Characterization of Trace Metrics}
\def\authorrunning{Valentina Castiglioni \& Simone Tini}

\begin{document}

\maketitle

\begin{abstract}
In this paper we continue our research line on logical characterizations of behavioral metrics obtained from the definition of a metric over the set of logical properties of interest.
This time we provide a characterization of both strong and weak trace metric on nondeterministic probabilistic processes, based on a minimal boolean logic $\LL$ which we prove to be powerful enough to characterize strong and weak probabilistic trace equivalence.
Moreover, we also prove that our characterization approach can be restated in terms of a more classic probabilistic $\LL$-model checking problem.
\end{abstract}


\section{Introduction}
\label{sec:intro}

Behavioral equivalences and modal logics have been successfully employed for the specification and verification of communicating concurrent systems, henceforth processes.
The former ones provide a simple and elegant tool for comparing the observable behavior of processes.
The latter ones allow for an immediate expression of the desired properties of processes.
Since the work of \cite{HM85} on the Hennessy-Milner logic (HML), these two approaches are connected by means of \emph{logical characterizations} of behavioral equivalences: two processes are behaviorally equivalent if and only if they satisfy the same formulae in the logic.
Hence, the characterization of an equivalence subsumes both the fact that the logic is as expressive as the equivalence and the fact that the equivalence preserves the logical properties of processes.

It is common agreement that when also quantitative properties of processes are taken into account a metric semantics is favored over behavioral equivalences, since the latter ones are too sensible to small variations in the probabilistic properties of processes.
Therefore, the interest in logical characterizations of the so called \emph{behavioral metrics} \cite{BW05,DGJP04,DJGP02,GJS90,KN96,DCPP06,SDC07,AFS09,AMRS08,LMP12}, namely the quantitative analogues of equivalences that quantify how far the behavior of two processes is apart, is constantly growing.

In this paper we propose a logical characterization of the \emph{strong} and \emph{weak} variants of the \emph{trace metric} \cite{SDC07} for nondeterministic probabilistic processes (PTSs \cite{S95}).
To this aim we follow the approach of \cite{CGT16a} in which a logical characterization of the bisimilarity metric is provided.
We introduce two boolean logics $\LL$ and $\LLw$, providing a probabilistic choice operator capturing the probability weights that a process assigns to arbitrary traces, which we prove to characterize resp.\ the \emph{strong} and \emph{weak probabilistic trace equivalences} of \cite{S95tr}.
Such a characterization is obtained by introducing the novel notion of \emph{mimicking formulae of resolutions}, i.e.\ formulae capturing, for each possible resolution of nondeterminism for a process, all the executable traces as well as the probability weights assigned to them.
Then we introduce the notions of \emph{distance between formulae} in $\LL$ and $\LLw$ which are $1$-bounded (pseudo)metrics assigning to each pair of formulae a suitable quantitative analogue of their syntactic disparities.
These lift to metrics over processes, called resp.\ $\LL$-\emph{distance} and $\LLw$-\emph{distance}, corresponding to the Hausdorff lifting of the distance between formulae to the sets of formulae satisfied by the two processes.
We prove that our $\LL$-distance and $\LLw$-distance correspond resp.\ to the strong and weak trace metric.

An important feature of our characterization method is that, although it is firmly based on the mimicking formulae of resolutions, it does not actually depend on how these resolutions of nondeterminism are obtained from processes.
For instance, in this paper we consider resolutions obtained via a \emph{deterministic scheduler} \cite{S95tr,BdNL14}, but our approach would not be different when applied to \emph{randomized resolutions} \cite{S95tr,BdNL14}.

Our approach differs from the ones proposed in the literature in that in general logics equipped with a real-valued semantics are used for the characterization, which is then expressed as
\begin{equation}
\label{eq:intro}
d(s,t) = \sup_{\varphi \in L} |\val{\varphi}{s} - \val{\varphi}{t}|
\end{equation}
where $d$ is the behavioral metric of interest, $L$ is the considered logic and $\val{\varphi}{s}$ denotes the value of the formula $\varphi$ at process $s$ accordingly to the real-valued semantics \cite{DGJP04,DJGP02,AFS09,AMRS08,DDG16}.
In \cite{BBLM15} it is proved that the trace metric on Markov Chains (MCs) can be characterized in terms of the probabilistic LTL-model checking problem.
Roughly speaking, a characterization as in~\eqref{eq:intro} is obtained from the boolean logic LTL by assigning a real-valued semantics to it, defined by exploiting the probabilistic properties of the MC: the value of a formula $\varphi \in$ LTL at state $s$ is given by the probability of $s$ to execute a run satisfying $\varphi$.
In this paper we show that we can obtain a similar result by means of our distance between formulae.
More precisely, we provide an alternative characterization of the trace metric on PTSs $\TraceMetric$ in terms of the probabilistic $\LL$-model checking problem.
In detail, we define a real-valued semantics for $\LL$ by assigning to each formula $\Psi \in \LL$ at process $s$ the value $\val{\Psi}{s}$ corresponding to the minimal distance between $\Psi$ and any formula satisfied by $s$.
Thus we could use this real-valued semantics to verify whether process $s$ behaves within an allowed tolerance wrt.\ to the specification given by the formula $\Psi$.
Then, by exploiting some properties of the Hausdorff metric, we will be able to conclude that 
$ \displaystyle
\TraceMetric(s,t) = \sup_{\Psi \in \LL} \mid \val{\Psi}{s} - \val{\Psi}{t} \mid
$
thus giving that the verification of any $\LL$-formula in $s$ cannot differ from its verification in $t$ for more than $\TraceMetric(s,t)$ which, in turn, constitutes the maximal observable error in the approximation of $s$ with $t$.

We can summarize our contributions as follows: 
\begin{enumerate}
\item Logical characterization of both strong and weak trace metric: we define a distance on the class of formulae $\LL$ (resp.\ $\LLw$) and we prove that the strong (resp.\ weak) trace metric between two processes equals the syntactic distance between the sets of formulae satisfied by them.
\item Logical characterization of strong trace metric in terms of a probabilistic $\LL$-model checking problem: by means of the distance between formulae we equip $\LL$ with a real-valued semantics and we use it to establish a characterization of the trace metric as in \eqref{eq:intro}.
\item Logical characterization of both strong and weak probabilistic trace equivalence: by exploiting the notion of mimicking formula, we prove that two processes are strong (resp.\ weak) trace equivalent if and only if they satisfy the same (resp.\ syntactically equivalent) formulae in $\LL$ (resp.\ $\LLw$).
\end{enumerate} 


\section{Background}
\label{sec:background_chap7}

\subsection{Nondeterministic probabilistic transition systems}

Nondeterministic probabilistic transition systems \cite{S95} combine LTSs \cite{K76} and discrete time Markov chains \cite{HJ94,Ste94}, allowing us to model reactive behavior, nondeterminism and probability.

As state space we take a set $\proc$, whose elements are called $\emph{processes}$.
We let $s,t,\ldots$ range over $\proc$.  
Probability distributions over $\proc$ are mappings $\pi \colon \proc \to [0,1]$ with $\sum_{s \in \proc} \pi(s) = 1$ that assign to each $s \in \proc$ its probability $\pi(s)$. 
By $\ProbDist{\proc}$ we denote the set of all distributions over $\proc$.
We let $\pi, \pi',\dots$ range over $\ProbDist{\proc}$.
For $\pi \in \ProbDist{\proc}$, we denote by $\support(\pi)$ the support of $\pi$, namely $\support(\pi) = \{ s \in \proc \mid \pi(s) >0\}$. 
We consider only distributions with \emph{finite} support.
For $s \in \proc$ we denote by $\delta_s$ the \emph{Dirac distribution} defined by $\delta_s(s)= 1$ and $\delta_s(t)=0$ for $s \neq t$.
The convex combination $\sum_{i \in I} p_i \pi_i$ of a family $\{\pi_i\}_{i \in I}$ of distributions $\pi_i \in \ProbDist{\proc}$ with $p_i \in (0,1]$ and $\sum_{i \in I} p_i = 1$ is defined by $(\sum_{i \in I} p_i \pi_i)(s) = \sum_{i \in I} (p_i \pi_i(s))$ for all $s \in \proc$. 

\begin{definition}[PTS, \cite{S95}]
A \emph{nondeterministic probabilistic labeled transition system (PTS)} is a triple $(\proc,\Act,\trans[])$, where: 
\begin{inparaenum}[(i)]
\item $\proc$ is a countable set of processes, 
\item $\Act$ is a countable set of \emph{actions}, and 
\item $\trans[] \subseteq {\proc \times \Act \times \ProbDist{\proc}}$ is a \emph{transition relation}. 
\end{inparaenum}
\end{definition}

We call $(s,a,\pi)\in\trans[]$ a \emph{transition}, and we write $s\trans[a]\pi$ for $(s,a,\pi) \in\trans[]$.
We write $s \trans[a] $ if there is a distribution $\pi \in \ProbDist{\proc}$ with $s \trans[a] \pi$, and $s \ntrans[a]$ otherwise. 
Let $\init{s} =\{ a \in \Act \mid  s\trans[a]\}$ denote the set of the actions that can be performed by $s$.
Let $\der{s,a} =\{\pi\in\ProbDist{\proc} \mid s\trans[a]\pi\}$ denote the set of the distributions reachable from $s$ through action $a$.
We say that a process $s \in \proc$ is \emph{image-finite} if for all actions $a \in\init{s}$ the set $\der{s,a}$ is finite \cite{HPSWZ11}.
In this paper we consider only processes that are image-finite.

Throughout the paper we will introduce some equivalence relations on traces and on modal formulae.
To deal with the equivalence of probability distributions over these elements, we need to introduce the notion of \emph{lifting} of a relation.

\begin{definition}
\label{def:lifting_relation}
Let $X$ be any set.
Consider a relation $\rel \subseteq X \times X$. 
Then the \emph{lifting} of $\rel$ is the relation $\reldist \subseteq \ProbDist{X} \times \ProbDist{X}$ with $\pi \reldist \pi'$ if whenever $\pi = \sum_{i \in I} p_i \delta_{x_i}$ then $\pi' = \sum_{i \in I, j_i \in J_i} p_{j_i} \delta_{y_{j_i}}$ with $\sum_{j_i \in J_i} p_{j_i} = p_i$ and $x_i \rel y_{j_i}$ for all $j_i \in J_i$.
\end{definition}

Moreover, we can lift relations to relations over sets.
Given a relation $\rel \subseteq X \times Y$, we say that two subsets $X' \subseteq X, Y' \subseteq Y$ are in relation $\rel$, notation $X' \rel Y'$, if{f} 
\begin{inparaenum}[(i)]
\item for each $x \in X'$ there is an $y \in Y'$ with $x \rel y$, and
\item for each $y \in Y'$ there is an $x \in X'$ with $x \rel y$. 
\end{inparaenum}


\subsection{Strong probabilistic trace equivalence}

A probabilistic trace equivalence is a relation over $\proc$ that equates processes $s,t \in \proc$ if for all resolutions of nondeterminism they can mimic each other's sequences of transitions with the same probability.

\begin{definition}
[Computation, \cite{BdNL14}]
\label{def:computation}
Let $P = (\proc, \Act, \trans[])$ be a PTS and $s,s' \in \proc$.
We say that
$c := s_0 \ctrans[a_1] s_1 \ctrans[a_2] s_2 \dots s_{n-1} \ctrans[a_n] s_n$
is a \emph{computation} of $P$ of length $n$ from $s = s_0$ to $s'= s_n$ if{f} for all $i = 1,\dots,n$ there exists a transition $s_{i-1} \trans[a_i] \pi_i$ in $P$ such that $s_i \in \support(\pi_i)$, with $\pi_i(s_i)$ being the \emph{execution probability} of step $s_{i-1} \ctrans[a_i] s_i$ conditioned on the selection of transition $s_{i-1} \trans[a_i] \pi_i$ of $P$ at $s_{i-1}$.
We denote by $\pr(c) = \prod_{i = 1}^{n} \pi_i(s_i)$ the product of the execution probabilities of the steps in $c$.
\end{definition}

Let $s,s',s'' \in \proc$.
Given any computation $c' = s' \ctrans[a_1] s_1 \ctrans[a_2] \dots \ctrans[a_n] s''$ from $s'$ to $s''$, we write $c = s \ctrans[a] c'$ if $c = s \ctrans[a] s' \ctrans[a_1] \dots \ctrans[a_n] s''$ is a computation from $s$ to $s''$.
We say that $c$ is a \emph{computation from $s$} if $c$ is a computation from $s$ to some process $s'$.
Then, $c$ is \emph{maximal} if it is not a proper prefix of any other computation from $s$. 
We denote by $\C(s)$ (resp.\ $\C_{\max}(s)$) the set of computations (resp.\ maximal computations) from $s$.
Given any $\C \subseteq \C(s)$, we define $\pr(\C) = \sum_{c \in \C} \pr(c)$ whenever none of the computations in $\C$ is a proper prefix of any of the others.

We denote by $\Act^{\star}$ the set of \emph{finite sequences} of actions in $\Act$ and we call \emph{trace} any element $\alpha \in \Act^{\star}$. 
The special symbol $\e \not \in \Act$ denotes the empty trace.
We say that a computation is \emph{compatible} with the trace $\alpha \in \Act^{\star}$ if{f} the sequence of actions labeling the computation steps is equal to $\alpha$.
We denote by $\C(s,\alpha) \subseteq \C(s)$ the set of computations of $s$ which are compatible with $\alpha$, and by $\C_{\max}(s,\alpha)$ the set $\C_{\max}(s,\alpha) = \C_{\max}(s) \cap \C(s,\alpha)$.

\begin{definition}
\label{def:trc}
Let $s \in \proc$ and consider any $c \in \C(s)$.
We denote by $\tr(c) \in \Act^{\star}$ the trace to which $c$ is compatible.
We extend this notion to sets by letting $\tr(\C') = \{\tr(c) \mid c \in \C'\}$ for any $\C' \subseteq \C(s)$.
We say that $\tr(\C(s))$ is the \emph{set of traces} of $s$ and $\tr(\C_{\max}(s))$ is the \emph{set of maximal traces} of $s$.
\end{definition}

To establish trace equivalence we need first to deal with nondeterministic choices of processes.
To this aim, we consider all possible resolutions of nondeterminism one by one.
Using the notation of \cite{BdNL14}, our resolutions correspond to the resolutions obtained via a \emph{deterministic scheduler} (see Fig.~\ref{fig:ex_resolutions} for an example).

\begin{definition}
[Resolution, \cite{BdNL14}]
\label{def:det_res}
Let $P = (\proc, \Act,\trans[])$ be a PTS and $s \in \proc$.
We say that a PTS $\Z = (Z,\Act,\trans[]_{\Z})$ is a \emph{resolution} for $s$ if{f} there exists a state correspondence function $\corr{\Z} \colon Z \to \proc$ such that $s = \corr{\Z}(z_s)$ for some $z_s \in Z$, called the \emph{initial state} of $\Z$, and moreover it holds that:
\begin{itemize}
\item $z_s \not \in \support(\pi)$ for any $\pi \in \bigcup_{z \in Z, a \in \Act} \der{z,a}$.
\item Each $z \in Z\setminus\{z_s\}$ is such that $z \in \support(\pi)$ for some $\pi \in \bigcup_{z' \in Z\setminus\{z\}, a \in \Act} \der{z',a}$.
\item Whenever $z \trans[a]_{\Z} \pi$, then $\corr{\Z}(z) \trans[a] \pi'$ with $\pi(z') = \pi'(\corr{\Z}(z'))$ for all $z' \in Z$.
\item Whenever $z \trans[a_1]_{\Z} \pi_1$ and $z \trans[a_2]_{\Z} \pi_2$ then $a_1 = a_2$ and $\pi_1 = \pi_2$. 
\end{itemize}
We let $\res(s)$ be the set of resolutions for $s$ and $\res(\proc) = \bigcup_{s \in \proc} \res(s)$ be the set of all resolutions on $\proc$.
\end{definition}

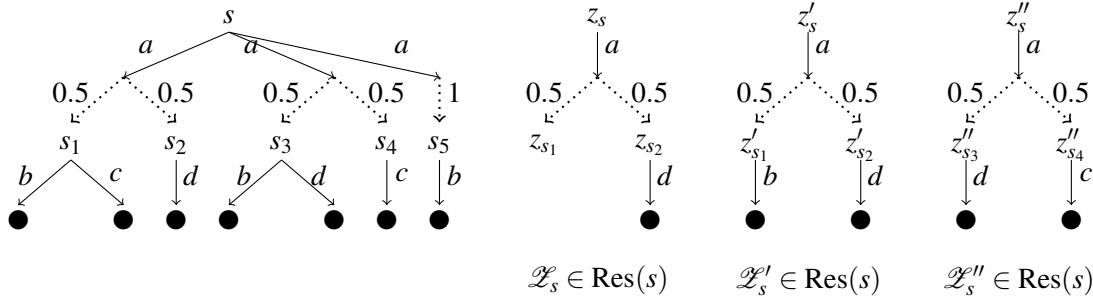
\begin{figure}[t!]
\begin{center}
\begin{tikzpicture}
\node at (2.8,4.7){$\boldsymbol{s}$};
\draw[->](2.8,4.5)--(1.4,3.9);
\node at (1.7,4.3){$\boldsymbol{a}$};
\draw[->](2.8,4.5)--(4.2,3.9);
\node at (3.1,4.3){$\boldsymbol{a}$};
\draw[->](2.8,4.5)--(5.6,3.9);
\node at (5.1,4.3){$\boldsymbol{a}$};
\draw[dotted,thick,->](1.4,3.9)--(0.7,3.3);
\node at (0.7,3.7){$\boldsymbol{0.5}$};
\draw[dotted,thick,->](1.4,3.9)--(2.1,3.3);
\node at (2.1,3.7){$\boldsymbol{0.5}$};
\draw[dotted,thick,->](4.2,3.9)--(3.5,3.3);
\node at (3.5,3.7){$\boldsymbol{0.5}$};
\draw[dotted,thick,->](4.2,3.9)--(4.9,3.3);
\node at (4.9,3.7){$\boldsymbol{0.5}$};
\draw[dotted,thick,->](5.6,3.9)--(5.6,3.3);
\node at (5.8,3.7){$\boldsymbol{1}$};
\node at (0.7,3){$\boldsymbol{s_1}$};
\node at (2.1,3){$\boldsymbol{s_2}$};
\node at (3.5,3){$\boldsymbol{s_3}$};
\node at (4.9,3){$\boldsymbol{s_4}$};
\node at (5.6,3){$\boldsymbol{s_5}$};
\draw[->](0.7,2.8)--(0,2.2);
\node at (0.1,2.6){$\boldsymbol{b}$};
\draw[->](0.7,2.8)--(1.4,2.2);
\node at (1.3,2.6){$\boldsymbol{c}$};
\draw[->](2.1,2.8)--(2.1,2.2);
\node at (2.3,2.6){$\boldsymbol{d}$};
\draw[->](3.5,2.8)--(2.8,2.2);
\node at (3,2.6){$\boldsymbol{b}$};
\draw[->](3.5,2.8)--(4.2,2.2);
\node at (4,2.6){$\boldsymbol{d}$};
\draw[->](4.9,2.8)--(4.9,2.2);
\node at (5.1,2.6){$\boldsymbol{c}$};
\draw[->](5.6,2.8)--(5.6,2.2);
\node at (5.8,2.6){$\boldsymbol{b}$};
\node at (0,2){$\CIRCLE$};
\node at (1.4,2){$\CIRCLE$};
\node at (2.1,2){$\CIRCLE$};
\node at (2.8,2){$\CIRCLE$};
\node at (4.2,2){$\CIRCLE$};
\node at (4.9,2){$\CIRCLE$};
\node at (5.6,2){$\CIRCLE$};
\node at (7.7,4.7){$\boldsymbol{z_s}$};
\draw[->](7.7,4.5)--(7.7,3.9);
\node at (7.9,4.3){$\boldsymbol{a}$};
\draw[dotted,thick,->](7.7,3.9)--(7,3.3);
\node at (7,3.7){$\boldsymbol{0.5}$};
\draw[dotted,thick,->](7.7,3.9)--(8.4,3.3);
\node at (8.4,3.7){$\boldsymbol{0.5}$};
\node at (7,3){$\boldsymbol{z_{s_1}}$};
\node at (8.4,3){$\boldsymbol{z_{s_2}}$};
\draw[->](8.4,2.8)--(8.4,2.2);
\node at (8.6,2.6){$\boldsymbol{d}$};
\node at (8.4,2){$\CIRCLE$};
\node at (7.7,1.2){$\Z_s \in \res(s)$};
\node at (10.5,4.7){$\boldsymbol{z'_s}$};
\draw[->](10.5,4.5)--(10.5,3.9);
\node at (10.7,4.3){$\boldsymbol{a}$};
\draw[dotted,thick,->](10.5,3.9)--(9.8,3.3);
\node at (9.8,3.7){$\boldsymbol{0.5}$};
\draw[dotted,thick,->](10.5,3.9)--(11.2,3.3);
\node at (11.2,3.7){$\boldsymbol{0.5}$};
\node at (9.8,3){$\boldsymbol{z'_{s_1}}$};
\node at (11.2,3){$\boldsymbol{z'_{s_2}}$};
\draw[->](9.8,2.8)--(9.8,2.2);
\node at (10,2.6){$\boldsymbol{b}$};
\draw[->](11.2,2.8)--(11.2,2.2);
\node at (11.4,2.6){$\boldsymbol{d}$};
\node at (9.8,2){$\CIRCLE$};
\node at (11.2,2){$\CIRCLE$};
\node at (10.5,1.2){$\Z_s' \in \res(s)$};
\node at (13.3,4.7){$\boldsymbol{z''_s}$};
\draw[->](13.3,4.5)--(13.3,3.9);
\node at (13.5,4.3){$\boldsymbol{a}$};
\draw[dotted,thick,->](13.3,3.9)--(12.6,3.3);
\node at (12.6,3.7){$\boldsymbol{0.5}$};
\draw[dotted,thick,->](13.3,3.9)--(14,3.3);
\node at (14,3.7){$\boldsymbol{0.5}$};
\node at (12.6,3){$\boldsymbol{z''_{s_3}}$};
\node at (14,3){$\boldsymbol{z''_{s_4}}$};
\draw[->](12.6,2.8)--(12.6,2.2);
\node at (12.8,2.6){$\boldsymbol{d}$};
\draw[->](14,2.8)--(14,2.2);
\node at (14.2,2.6){$\boldsymbol{c}$};
\node at (12.6,2){$\CIRCLE$};
\node at (14,02){$\CIRCLE$};
\node at (13.3,1.2){$\Z_s'' \in \res(s)$};
\end{tikzpicture}
\end{center}
\caption{\label{fig:ex_resolutions} An example of three distinct resolutions for process $s$. Black circles stand for the probability distribution $\delta_{\mathrm{nil}}$, with $\mathrm{nil}$ process that cannot execute any action.}
\end{figure}

\emph{Strong probabilistic trace equivalence} equates two processes if their resolutions can be matched so that they assign the same probability to all traces.

\begin{definition}
[Strong probabilistic trace equivalence, \cite{S95tr,BdNL14}]
\label{def:prob_trace_dist}
Let $P=(\proc,\Act,\trans[])$ be a PTS.
We say that $s,t \in \proc$ are \emph{strong probabilistic trace equivalent}, notation $s \STr t$, if{f} it holds that:
\begin{itemize}
\item For each resolution $\Z_s \in \res(s)$ of $s$ there is a resolution $\Z_t \in \res(t)$ of $t$ such that for all traces 
$\alpha \in \Act^{\star}$ we have
$
\pr(\C(z_s,\alpha)) = \pr(\C(z_t,\alpha)).
$ 
\item For each resolution $\Z_t \in \res(t)$ of $t$ there is a resolution $\Z_s \in \res(s)$ of $s$ such that for all traces 
$\alpha \in \Act^{\star}$ we have
$
\pr(\C(z_t,\alpha)) = \pr(\C(z_s,\alpha)).
$ 
\end{itemize}
\end{definition}

\begin{figure}
\begin{center}
\begin{tikzpicture}
\node at (1.75,4.7){$\boldsymbol{t}$};
\draw[->](1.75,4.5)--(1.75,3.9);
\node at (1.95,4.3){$\boldsymbol{a}$};
\draw[dotted,thick,->](1.75,3.9)--(0.7,3.3);
\node at (0.7,3.7){$\boldsymbol{0.5}$};
\draw[dotted,thick,->](1.75,3.9)--(2.8,3.3);
\node at (2.8,3.7){$\boldsymbol{0.5}$};
\node at (0.7,3){$\boldsymbol{t_1}$};
\node at (2.8,3){$\boldsymbol{t_2}$};
\draw[->](0.7,2.8)--(0,2.2);
\node at (0.2,2.6){$\boldsymbol{b}$};
\draw[->](0.7,2.8)--(1.4,2.2);
\node at (1.2,2.6){$\boldsymbol{c}$};
\draw[->](2.8,2.8)--(2.1,2.2);
\node at (2.3,2.6){$\boldsymbol{b}$};
\draw[->](2.8,2.8)--(3.5,2.2);
\node at (3.3,2.6){$\boldsymbol{d}$};
\node at (0,2){$\CIRCLE$};
\node at (1.4,2){$\CIRCLE$};
\node at (2.1,2){$\CIRCLE$};
\node at (3.5,2){$\CIRCLE$};
\node at (4.9,4.7){$\boldsymbol{z_t}$};
\draw[->](4.9,4.5)--(4.9,3.9);
\node at (5.1,4.3){$\boldsymbol{a}$};
\draw[dotted,thick,->](4.9,3.9)--(4.2,3.3);
\node at (4.2,3.7){$\boldsymbol{0.5}$};
\draw[dotted,thick,->](4.9,3.9)--(5.6,3.3);
\node at (5.6,3.7){$\boldsymbol{0.5}$};
\node at (4.2,3){$\boldsymbol{z_{t_1}}$};
\node at (5.6,3){$\boldsymbol{z_{t_2}}$};
\draw[->](5.6,2.8)--(5.6,2.2);
\node at (5.8,2.6){$\boldsymbol{d}$};
\node at (5.6,2){$\CIRCLE$};
\node at (4.9,1.2){$\Z_t \in \res(t)$};
\node at (7.7,4.7){$\boldsymbol{z'_t}$};
\draw[->](7.7,4.5)--(7.7,3.9);
\node at (7.9,4.3){$\boldsymbol{a}$};
\draw[dotted,thick,->](7.7,3.9)--(7,3.3);
\node at (7,3.7){$\boldsymbol{0.5}$};
\draw[dotted,thick,->](7.7,3.9)--(8.4,3.3);
\node at (8.4,3.7){$\boldsymbol{0.5}$};
\node at (7,3){$\boldsymbol{z'_{t_1}}$};
\node at (8.4,3){$\boldsymbol{z'_{t_2}}$};
\draw[->](7,2.8)--(7,2.2);
\node at (7.2,2.6){$\boldsymbol{b}$};
\draw[->](8.4,2.8)--(8.4,2.2);
\node at (8.6,2.6){$\boldsymbol{d}$};
\node at (7,2){$\CIRCLE$};
\node at (8.4,2){$\CIRCLE$};
\node at (7.7,1.2){$\Z_t' \in \res(t)$};
\node at (10.5,4.7){$\boldsymbol{z''_t}$};
\draw[->](10.5,4.5)--(10.5,3.9);
\node at (10.7,4.3){$\boldsymbol{a}$};
\draw[dotted,thick,->](10.5,3.9)--(9.8,3.3);
\node at (9.8,3.7){$\boldsymbol{0.5}$};
\draw[dotted,thick,->](10.5,3.9)--(11.2,3.3);
\node at (11.2,3.7){$\boldsymbol{0.5}$};
\node at (9.8,3){$\boldsymbol{z''_{t_1}}$};
\node at (11.2,3){$\boldsymbol{z''_{t_2}}$};
\draw[->](9.8,2.8)--(9.8,2.2);
\node at (10,2.6){$\boldsymbol{c}$};
\draw[->](11.2,2.8)--(11.2,2.2);
\node at (11.4,2.6){$\boldsymbol{d}$};
\node at (9.8,2){$\CIRCLE$};
\node at (11.2,2){$\CIRCLE$};
\node at (10.5,1.2){$\Z_t'' \in \res(t)$};
\node at (13.3,4.7){$\boldsymbol{z'''_t}$};
\draw[->](13.3,4.5)--(13.3,3.9);
\node at (13.5,4.3){$\boldsymbol{a}$};
\draw[dotted,thick,->](13.3,3.9)--(12.6,3.3);
\node at (12.6,3.7){$\boldsymbol{0.5}$};
\draw[dotted,thick,->](13.3,3.9)--(14,3.3);
\node at (14,3.7){$\boldsymbol{0.5}$};
\node at (12.6,3){$\boldsymbol{z'''_{t_1}}$};
\node at (14,3){$\boldsymbol{z'''_{t_2}}$};
\draw[->](12.6,2.8)--(12.6,2.2);
\node at (12.8,2.6){$\boldsymbol{b}$};
\draw[->](14,2.8)--(14,2.2);
\node at (14.2,2.6){$\boldsymbol{b}$};
\node at (12.6,2){$\CIRCLE$};
\node at (14,2){$\CIRCLE$};
\node at (13.3,1.2){$\Z_t''' \in \res(t)$};
\end{tikzpicture}
\end{center}
\caption{\label{fig:strong_trace} Process $t$ is strong trace equivalent to process $s$ in Fig.~\ref{fig:ex_resolutions}}
\end{figure}
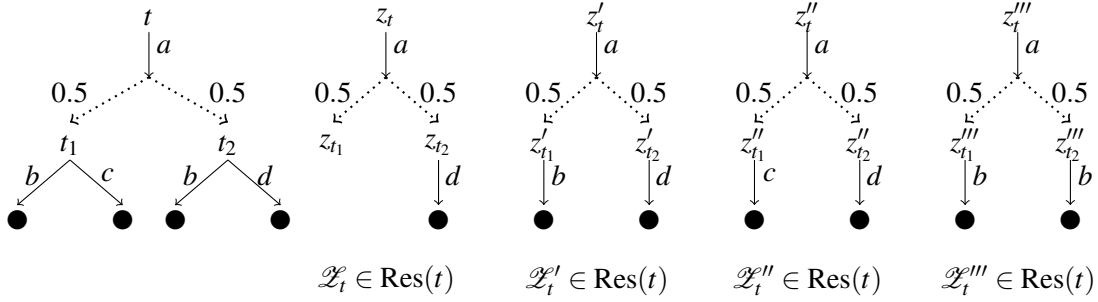

\begin{example}
\label{ex:strong_trace}
Consider process $s$ in Fig.~\ref{fig:ex_resolutions} and process $t$ in Fig.~\ref{fig:strong_trace}.
We have that $s \STr t$.
Briefly, it is immediate to check that the three resolutions $\Z_s ,\Z'_s,\Z''_s \in \res(s)$ in Fig.~\ref{fig:ex_resolutions} are matched resp.\ by the three resolutions $\Z_t,\Z_t',\Z_t'' \in \res(t)$ in Fig.~\ref{fig:strong_trace}.
Moreover, for all other resolutions, we notice that accordingly to the chosen resolutions for processes $t_1$ and $t_2$, process $s$ can always match their traces and related probabilities by selecting the proper $a$-branch.
In particular, resolution $\Z_t''' \in \res(t)$ in Fig.~\ref{fig:strong_trace} is matched by the resolution for $s$ corresponding to the rightmost $a$-branch. 
\end{example}


\subsection{Weak probabilistic trace equivalence}

We extend the set of actions $\Act$ to the set $\Acttau$ containing also the silent action $\tau$.
We let $\A$ range over $\Acttau$.

Usually, traces are not distinguished by any occurrence of $\tau$ in them \cite{SDC07}.
Hence, we introduce the notion of \emph{equivalence of traces}.

\begin{definition}
[Equivalence of traces]
\label{def:eqtrace}
The relation of \emph{equivalence of traces} $\eqtrace \subseteq \Acttau^{\star} \times \Acttau^{\star}$ is the smallest equivalence relation satisfying
\begin{inparaenum}
\item $\varepsilon \eqtrace \varepsilon$ and
\item given $\alpha = \A_1 \alpha'$, $\beta = \A_2 \beta'$ we have $\alpha \eqtrace \beta$ if{f} 
\end{inparaenum}
\begin{itemize}
\item either $\A_1 = \tau$ and $\alpha' \eqtrace \beta$,
\item or $\A_2 = \tau$ and $\alpha \eqtrace \beta'$
\item or $\A_1 = \A_2$ and $\alpha' \eqtrace \beta'$.
\end{itemize}
\end{definition}

For each trace $\alpha \in \Acttau^{\star}$, we denote by $[\alpha]_{\mathrm{w}}$ the equivalence class of $\alpha$ with respect to $\eqtrace$, namely $[\alpha]_{\mathrm{w}} = \{\beta \in \Acttau^{\star} \mid \beta \eqtrace \alpha\}$.
Moreover, for each computation $c$, we let $\trw(c) = [\tr(c)]_{\mathrm{w}}$.

Given any process $s \in \proc$ and any trace $\alpha \in \Acttau^{\star}$, we say that a computation $c \in \C(s)$ is in $\Cw(s,\alpha)$ if{f} $\tr(c) \eqtrace \alpha$ and $c$ is not a proper prefix of any other computation in $\Cw(s,\alpha)$. 
This is to avoid to count multiple times the same execution probabilities in the evaluation of $\pr(\Cw(s,\alpha))$.

\begin{definition}
[Weak probabilistic trace equivalence]
\label{def:weak_trace_equivalence}
Let $P=(\proc,\Act,\trans[])$ be a PTS.
We say that $s,t \in \proc$ are \emph{weak probabilistic trace equivalent}, notation $s \WTr t$, if{f} it holds that:
\begin{itemize}
\item For each resolution $\Z_s \in \res(s)$ of $s$ there is a resolution $\Z_t \in \res(t)$ of $t$ such that for all traces 
$\alpha \in \Act^{\star}$ we have
$
\pr(\Cw(z_s,\alpha)) = \pr(\Cw(z_t,\alpha)).
$ 
\item For each resolution $\Z_t \in \res(t)$ of $t$ there is a resolution $\Z_s \in \res(s)$ of $s$ such that for all traces 
$\alpha \in \Act^{\star}$ we have
$
\pr(\Cw(z_t,\alpha)) = \pr(\Cw(z_s,\alpha)).
$ 
\end{itemize}
\end{definition}


\section{Trace metrics}
\label{sec:trace_metrics}

In this section we introduce the quantitative analogues of strong and weak probabilistic trace equivalence, namely the \emph{strong} and \emph{weak trace metric}, resp., which are $1$-bounded pseudometrics that quantify how much the behavior of two processes is apart wrt.\ the strong (resp.\ weak) probabilistic trace semantics. 
Our metrics are a revised version of the trace metric proposed in \cite{SDC07}.
Briefly, in \cite{SDC07} there is a distinction between the notions of \emph{path} and \emph{trace}: any $\alpha \in \Acttau^{\star}$ is called path and the trace related to a path is obtained by deleting any occurrence of $\tau$ from it.
The metric in \cite{SDC07} is then defined only on traces and it has inspired our strong trace metric.
In the present paper we distinguish between the strong and the weak case and we regain the results in \cite{SDC07} by our equivalence of traces: the weak trace metric coincides with the strong one on the quotient space wrt.\ $\eqtrace$.

\subsection{The Kantorovich and Hausdorff lifting functionals}

In the literature we can find several examples of behavioral metrics on systems with probability and nondeterminism (see among others \cite{AFS09,BW01a,B05,DCPP06,DGJP04,SDC07}).
In this paper we follow the approach of \cite{BW01a,DCPP06,SDC07} in which two kind of metrics are combined to obtain a metric on the system.
The \emph{Kantorovich metric} \cite{K42} quantifies the disparity between the probabilistic properties of processes and it is defined by means of the notion of \emph{matching}.
For any set $X$, a matching for distributions $\pi,\pi' \in \ProbDist{X}$ is a distribution over the product space $\w \in \Delta(X \times X)$ with $\pi$ and $\pi'$ as left and right marginal resp., namely $\sum_{y\in X} \w(x,y)=\pi(x)$ and $\sum_{x\in X} \w(x,y)=\pi'(y)$ for all $x,y \in X$. 
Let $\W(\pi,\pi')$ denote the set of all matchings for $\pi,\pi'$.

\begin{definition}
[Kantorovich metric, \cite{K42}]
\label{def:Kantorovich}
Let $d\colon X \times X \to [0,1]$ be a 1-bounded metric. 
The \emph{Kantorovich lifting} of $d$ is the 1-bounded metric $\Kantorovich(d)\colon \ProbDist{X} \times \ProbDist{X} \to [0,1]$ defined for all $\pi,\pi' \in \ProbDist{X}$ by
\[
\Kantorovich(d)(\pi,\pi') = \min_{\w \in \W(\pi,\pi')} \sum_{x,y \in X}\w(x,y) \cdot d(x,y).
\]
\end{definition}

We remark that since we are considering only probability distributions with finite support, the minimum over $\W(\pi,\pi')$ is well defined for all $\pi,\pi'\in \ProbDist{X}$. 

The \emph{Hausdorff metric} allows us to lift any distance over probability distributions to a distance over sets of probability distributions.

\begin{definition}
[Hausdorff metric]
\label{def:Hausdorff}
Let $\hat{d} \colon \ProbDist{X} \times \ProbDist{X} \to [0,1]$ be a 1-bounded metric.
The \emph{Hausdorff lifting} of $\hat{d}$ is the $1$-bounded metric $\Hausdorff(\hat{d})\colon \powset{\ProbDist{X}} \times \powset{\ProbDist{X}} \to [0,1]$ defined by
\[
\Hausdorff(\hat{d})(\Pi_1,\Pi_2) = \max \Big\{ \sup_{\pi_1 \in \Pi_1}\inf_{\pi_2 \in \Pi_2} \hat{d}(\pi_1,\pi_2), \sup_{\pi_2\in \Pi_2}\inf_{\pi_1\in \Pi_1} \hat{d}(\pi_2,\pi_1) \Big\}
\]
for all $\Pi_1,\Pi_2 \subseteq \ProbDist{X}$, where $\inf \emptyset = 1$, $\sup \emptyset = 0$.
\end{definition}

Hence, given two processes $s,t\in \proc$, the idea is to quantify the distance between each pair of their resolutions by exploiting the Kantorovich metric, which quantifies the disparities in the probabilities of the two processes to execute the same traces.
Then, we lift this distance on resolutions to a distance between $s$ and $t$ by means of the Hausdorff metric.
Intuitively, as each resolution captures a different set of nondeterministic choices of a process, we use the Hausdorff metric to compare the possible choices of the two processes and to match them in order to obtain the minimal distance. 


\subsection{Strong trace metric}
\label{sec:strong_trace_metric}

To define the strong trace metric we start from a distance between traces, defined as the discrete metric over traces: two traces are at distance $1$ if they are distinct, otherwise the distance is set to $0$.
Differently from \cite{SDC07} we do not consider any discount on the distance between traces.
Trace equivalences, and thus metrics, are usually employed when the observations on the system cannot be done in a step-by-step fashion, but only the total behavior of the system can be observed.
Hence, a step-wise discount does not fit in this setting.
However, the discount would not introduce any technical issue.

\begin{definition}
[Distance between traces]
\label{def:trace_metric_traces}
The \emph{distance between traces} $d_T \colon \Act^{\star} \times \Act^{\star} \to [0,1]$ is defined for any pair of traces $\alpha, \beta \in \Act^{\star}$ by
\[
d_T(\alpha, \beta) = 
\begin{cases}
0 & \text{ if  } \alpha = \beta \\
1 & \text{ otherwise.}
\end{cases}
\] 
\end{definition}

Following \cite{SDC07} we aim to lift the distance $d_T$ to a distance between resolutions by means of the Kantorovich lifting functional
which, we recall, is defined on probability distributions.
As shown in the following example, we are not guaranteed that the function $\pr(\C(\_,\_))$ defines a probability distribution on the set of traces of a resolution.

\begin{example}
\label{ex:Pr_not_prob_dist}
Consider process $t$ and the resolution $\Z_r \in \res(t)$ for it, represented in Fig.~\ref{fig:strong_trace}.
We can distinguish three computations for $z_t$:
\[
\begin{array}{l}
 c_1 = z_t \ctrans[a] z_{t_1} \\
 c_2 = z_t \ctrans[a] z_{t_2} \\
 c_3 = z_t \ctrans[a] z_{t_2} \ctrans[d] \mathrm{nil}. 
\end{array}
\]
Clearly, $\tr(\C(z_t)) = \{a, ad\}$.
Then we have
\[
\begin{array}{l}
 \pr(\C(z_t,a)) = \sum_{c \in \C(z_t,a)} \pr(c) = \pr(c_1) + \pr(c_2) = 1 \\[0.5 ex]
 \pr(\C(z_t,ad)) = \sum_{c \in \C(z_t,ad)} \pr(c) = \pr(c_3) = 0.5
\end{array}
\]
from which we gather 
\[
\sum_{\alpha \in \tr(\C(z_t))} \pr(\C(z_t,\alpha)) = \pr(\C(z_t,a)) + \pr(\C(z_t,ad)) = 1 + 0.5 > 1.
\]
\end{example}

However, as shown in the following lemma, if we consider only maximal computations we obtain a probability distribution over traces.

\begin{lemma}
\label{lem:prob_computation_sum_to_1}
Consider any resolution $\Z \in \res(\proc)$ with initial state $z$.
We have that $\sum_{c \in \C_{\max}(z)} \pr(c) = 1$.
\end{lemma}


\begin{proof}
We proceed by induction over the depth of $z$.

The base case $\depth{z}=0$ is immediate since we have that $\C(z) = \{\varepsilon\}$ and $\pr(\varepsilon) = 1$.

Consider now the inductive step $\depth{z}>0$.
Assume, wlog., that $z \trans[a]_{\Z} \pi$.
Therefore, each trace $c \in \C_{\max}(z)$ will be of the form $c = z \ctrans[a] c'$ for some $c' \in \C_{\max}(z')$ for any $z' \in \support(\pi)$ and moreover for such a trace $c$ it holds that $\pr(c) = \pi(z') \pr(c')$.
Thus we have
\[
\begin{array}{llr}
\sum_{c \in \C_{\max}(z)} \pr(c)  &={} \sum_{z' \in \support(\pi) \atop c' \in \C_{\max}(z')} \pi(z') \pr(c') \\
 & ={} \sum_{z' \in \support(\pi)} \pi(z') \Big( \sum_{c' \in \C_{\max}(z')}\pr(c') \Big) \\
 & ={}\sum_{z' \in \support(\pi)} \pi(z') \cdot 1 & \text{(by induction over $\depth{z'} < \depth{z}$)} \\
& ={}  1.
\end{array}
\]
\end{proof}


\begin{definition}
[Trace distribution]
\label{def:trace_distribution}
Consider any resolution $\Z \in \res(\proc)$, with initial state $z$.
We define the \emph{trace distribution} of $\Z$ as the function $\TD_{\Z} \colon \Act^{\star} \to [0,1]$ defined for each $\alpha \in \Act^{\star}$ by 
\[
\TD_{\Z}(\alpha) = \pr(\C_{\max}(z,\alpha)).
\]
\end{definition}

Notice that only maximal computations are in the support of $\TD_{\Z}$.
This guarantees that $\TD_{\Z}$ is a distribution. 

\begin{lemma}
\label{lem:trace_distribution}
Consider any resolution $\Z \in \res(\proc)$, with initial state $z$.
Then the trace distribution $\TD_{\Z}$ of $\Z$ is a probability distribution over $\Act^{\star}$.
\end{lemma}


\begin{proof}
By definition and by Lemma~\ref{lem:prob_computation_sum_to_1} we have that for each $\alpha \in \Act^{\star}$
\[
0 \le \pr(\C_{\max}(z,\alpha)) = \sum_{c \in \C_{\max}(z,\alpha)} \pr(c) \le \sum_{c \in \C_{\max}(z)} \pr(c) = 1 
\]
Hence, we are guaranteed that $\TD_{\Z}(\alpha) \in [0,1]$ for each $\alpha \in \Act^{\star}$.
Thus, to prove the thesis we simply need to show that $\sum_{\alpha \in \Act^{\star}} \TD_{\Z}(\alpha) = 1$.
We have that
\[
\begin{array}{rlr}
\sum_{\alpha \in \Act^{\star}} \TD_{\Z}(\alpha) ={} & \sum_{\alpha \in \Act^{\star}} \pr(\C_{\max}(z,\alpha)) \\
={} & \sum_{\alpha \in \tr(\C_{\max}(z))} \pr(\C_{\max}(z,\alpha)) \\
={} & \sum_{\alpha \in \tr(\C_{\max}(z)), c \in \C_{\max}(z,\alpha)} \pr(c) \\
={} & \sum_{c \in \bigcup_{\alpha \in \tr(\C_{\max}(z))} \C_{\max}(z,\alpha)} \pr(c) \\
={} & \sum_{c \in \C_{\max}(z)} \pr(c) \\
={} & 1
\end{array}
\]
where
\begin{itemize}
\item the second equality follows from the fact that by definition $\pr(\C_{\max}(z,\alpha)) = 0$ for each $\alpha \not \in \tr(\C_{\max}(z))$;
\item the fourth equality follows from the fact that each maximal computation of $z$ belongs to a set $\C_{\max}(z,\alpha)$ for at most one trace $\alpha$, namely $\bigcup_{\alpha \in \tr(\C_{\max}(z))} \C_{\max}(z,\alpha)$ is a disjoint union (and therefore no probability weight is counted more than once);
\item the fifth equality follows by the fact that the disjoint union $\bigcup_{\alpha \in \tr(\C_{\max}(z))} \C_{\max}(z,\alpha)$ is a partition of $\C_{\max}(z)$;
\item the sixth equality follows by Lemma~\ref{lem:prob_computation_sum_to_1}.
\end{itemize}
\end{proof}


We remark that function $\TD_{\_}$ plays the role of the \emph{trace distribution} introduced in \cite{S95tr}.
Formally, in \cite{S95tr} the trace distribution for a resolution is defined as the probability space built over its set of traces.
Here, we simply identify it with the probability distribution defined on the probability space.
In this setting, two resolutions are said to be \emph{trace distribution equivalent} if they have the same trace distribution and thus two processes are trace equivalent if their resolutions are pairwise equivalent.

\begin{lemma}
\label{lem:computations_same_sum}
Consider any resolution $\Z \in \res(\proc)$ with initial state $z$.
Consider any trace $\alpha \in \Act^{\star}$.
Then $\pr(\C(z,\alpha)) = \sum_{c \in P_{\max}(z,\alpha)} \pr(c)$, where $P_{\max}(z,\alpha)$ is the set of maximal computations from $z$ having a prefix which is compatible with $\alpha$.
\end{lemma}


\begin{proof}
For simplicity let us distinguish two cases.
\begin{enumerate}
\item $\pr(\C(z,\alpha)) = 0$.
This implies that there is no computation from $z$ which is compatible with $\alpha$.
Clearly, this gives that there can not be any maximal computation from $z$ having a prefix compatible with $\alpha$, namely $P_{\max}(z,\alpha) = \emptyset$.
Thus we have $\sum_{c \in P_{\max}(z,\alpha)} \pr(c) = 0$ from which the thesis follows.

\item $\pr(\C(z,\alpha)) > 0$.
In this case, we proceed by induction over $|\alpha|$.
\begin{itemize}
\item Base case $|\alpha|=0$, namely $\alpha = \varepsilon$.
The only computation compatible with $\alpha$ is the empty computation for which it holds that $\pr(\C(z,\alpha)) = 1$.
Since the empty computation is a prefix for all computations from $z$ we have that $P_{\max}(z,\alpha) = \C_{\max}(z)$.
By Lemma~\ref{lem:prob_computation_sum_to_1} we have that $\sum_{c \in \C_{\max}(z)}\pr(c) = 1$ and thus the thesis follows.
\item Inductive step $|\alpha|>0$.
Assume wlog that the only transition inferable for $z$ in $\Z$ is $z \trans[a]_{\Z}\pi$.
Hence $\alpha = a\alpha'$ for some $\alpha' \in \Act^{\star}$, with $|\alpha'|< |\alpha|$.
Then we have
\[
\begin{array}{rlr}
\pr(\C(z,\alpha)) ={} & \sum_{z' \in \support(\pi)} \pi(z') \pr(\C(z',\alpha')) \\
={} & \sum_{z' \in \support(\pi)} \Big( \pi(z') \cdot \sum_{c' \in P_{\max}(z',\alpha')} \pr(c') \Big) & \text{(by induction over $|\alpha'|$)}\\
={} & \sum_{z' \in \support(\pi),\, c' \in P_{\max}(z',\alpha')} \pi(z') \pr(c') \\
={} & \sum_{c \in P_{\max}(z,a\alpha')} \pr(c)
\end{array}
\]
where the last equality follows by considering that 
\[
P_{\max}(z,a\alpha') = \Big\{ c \mid c = z \trans[a]_{\Z} c' \text{ and } c' \in \bigcup_{z'\in \support(\pi)}P_{\max}(z',\alpha') \Big\}.
\]
\end{itemize}
\end{enumerate}
\end{proof}


\begin{proposition}
\label{prop:equivalence_S95tr_BdNL14_strong}
For any pair of resolutions $\Z_1,\Z_2 \in \res(\proc)$, with initial states $z_1,z_2$ resp., we have that $\TD_{\Z_1} = \TD_{\Z_2}$ if{f} $\Pr(\C(z_1,\alpha)) = \Pr(\C(z_2,\alpha))$ for all traces $\alpha \in \Act^{\star}$.
\end{proposition}


\begin{proof}
The thesis follows by applying the same arguments used it the proof of Theorem~\ref{thm:equivalent_resolutions} below.
\end{proof}


Hence, we can now follow \cite{SDC07} to define the \emph{trace metric}.

\begin{definition}
[Trace distance on resolutions]
\label{def:trace_metric_det_res}
The \emph{trace distance on resolutions} $D_T \colon \res(\proc) \times \res(\proc) \to [0,1]$ is defined for any $\Z_1,\Z_2 \in \res(\proc)$ by
\[
D_T(\Z_1,\Z_2) = \Kantorovich(d_T)(\TD_{\Z_1}, \TD_{\Z_2}).
\]
\end{definition}

\begin{proposition}
[\!\protect{\cite[Proposition~2]{SDC07}}]
\label{prop:kernel_trace_distribution}
The kernel of $D_T$ is strong trace distribution equivalence of resolutions.
\end{proposition}

To deal with nondeterministic choices, we lift the distance over deterministic resolutions to a pseudometric over processes by means of the Hausdorff lifting functional.

\begin{definition}
[Strong trace metric]
\label{def:trace_metric}
\emph{Strong trace metric} $\TraceMetric \colon \proc \times \proc \to [0,1]$ is defined for all $s,t \in \proc$ as 
\[
\TraceMetric(s,t) = \Hausdorff(D_T)(\res(s), \res(t)).
\]
\end{definition}

\begin{proposition}
[\!\protect{\cite[Proposition~3]{SDC07}}]
\label{prop:kernel_trace_equivalent}
The kernel of $\TraceMetric$ is probabilistic strong trace equivalence.
\end{proposition}

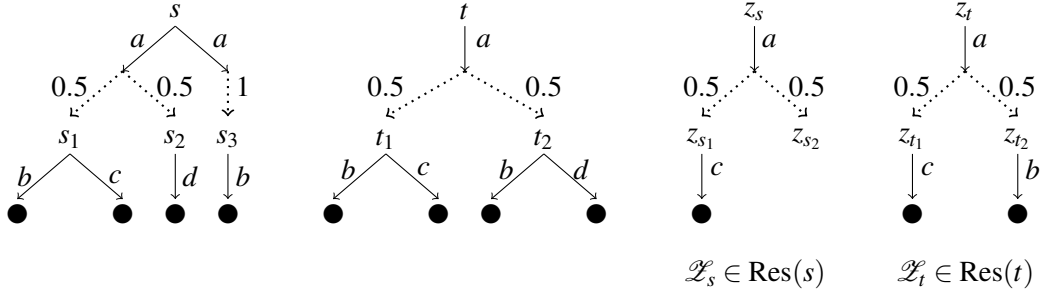
\begin{figure}
\begin{center}
\begin{tikzpicture}
\node at (2.1,4.7){$\boldsymbol{s}$};
\draw[->](2.1,4.5)--(1.4,3.9);
\node at (1.6,4.3){$\boldsymbol{a}$};
\draw[->](2.1,4.5)--(2.8,3.9);
\node at (2.7,4.3){$\boldsymbol{a}$};
\draw[dotted,thick,->](1.4,3.9)--(0.7,3.3);
\node at (0.7,3.7){$\boldsymbol{0.5}$};
\draw[dotted,thick,->](1.4,3.9)--(2.1,3.3);
\node at (2.1,3.7){$\boldsymbol{0.5}$};
\draw[dotted,thick,->](2.8,3.9)--(2.8,3.3);
\node at (3,3.7){$\boldsymbol{1}$};
\node at (0.7,3){$\boldsymbol{s_1}$};
\node at (2.1,3){$\boldsymbol{s_2}$};
\node at (2.8,3){$\boldsymbol{s_3}$};
\draw[->](0.7,2.8)--(0,2.2);
\node at (0.1,2.5){$\boldsymbol{b}$};
\draw[->](0.7,2.8)--(1.4,2.2);
\node at (1.3,2.5){$\boldsymbol{c}$};
\draw[->](2.1,2.8)--(2.1,2.2);
\node at (2.3,2.5){$\boldsymbol{d}$};
\draw[->](2.8,2.8)--(2.8,2.2);
\node at (3,2.5){$\boldsymbol{b}$};
\node at (0,2){$\CIRCLE$};
\node at (1.4,2){$\CIRCLE$};
\node at (2.1,2){$\CIRCLE$};
\node at (2.8,2){$\CIRCLE$};
\node at (5.95,4.7){$\boldsymbol{t}$};
\draw[->](5.95,4.5)--(5.95,3.9);
\node at (6.2,4.3){$\boldsymbol{a}$};
\draw[dotted,thick,->](5.95,3.9)--(4.9,3.3);
\node at (4.9,3.7){$\boldsymbol{0.5}$};
\draw[dotted,thick,->](5.95,3.9)--(7,3.3);
\node at (7,3.7){$\boldsymbol{0.5}$};
\node at (4.9,3){$\boldsymbol{t_1}$};
\node at (7,3){$\boldsymbol{t_2}$};
\draw[->](4.9,2.8)--(4.2,2.2);
\node at (4.4,2.6){$\boldsymbol{b}$};
\draw[->](4.9,2.8)--(5.6,2.2);
\node at (5.4,2.6){$\boldsymbol{c}$};
\draw[->](7,2.8)--(6.3,2.2);
\node at (6.5,2.6){$\boldsymbol{b}$};
\draw[->](7,2.8)--(7.7,2.2);
\node at (7.5,2.6){$\boldsymbol{d}$};
\node at (4.2,2){$\CIRCLE$};
\node at (5.6,2){$\CIRCLE$};
\node at (6.3,2){$\CIRCLE$};
\node at (7.7,2){$\CIRCLE$};
\node at (9.8,4.7){$\boldsymbol{z_s}$};
\draw[->](9.8,4.5)--(9.8,3.9);
\node at (10,4.3){$\boldsymbol{a}$};
\draw[dotted,thick,->](9.8,3.9)--(9.1,3.3);
\node at (9.1,3.7){$\boldsymbol{0.5}$};
\draw[dotted,thick,->](9.8,3.9)--(10.5,3.3);
\node at (10.5,3.7){$\boldsymbol{0.5}$};
\node at (9.1,3){$\boldsymbol{z_{s_1}}$};
\node at (10.5,3){$\boldsymbol{z_{s_2}}$};
\draw[->](9.1,2.8)--(9.1,2.2);
\node at (9.3,2.6){$\boldsymbol{c}$};
\node at (9.1,2){$\CIRCLE$};
\node at (9.8,1.2){$\Z_s \in \res(s)$};
\node at (12.6,4.7){$\boldsymbol{z_t}$};
\draw[->](12.6,4.5)--(12.6,3.9);
\node at (12.8,4.3){$\boldsymbol{a}$};
\draw[dotted,thick,->](12.6,3.9)--(11.9,3.3);
\node at (11.9,3.7){$\boldsymbol{0.5}$};
\draw[dotted,thick,->](12.6,3.9)--(13.3,3.3);
\node at (13.3,3.7){$\boldsymbol{0.5}$};
\node at (11.9,3){$\boldsymbol{z_{t_1}}$};
\node at (13.3,3){$\boldsymbol{z_{t_2}}$};
\draw[->](11.9,2.8)--(11.9,2.2);
\node at (12.1,2.6){$\boldsymbol{c}$};
\draw[->](13.3,2.8)--(13.3,2.2);
\node at (13.5,2.6){$\boldsymbol{b}$};
\node at (11.9,2){$\CIRCLE$};
\node at (13.3,2){$\CIRCLE$};
\node at (12.6,1.2){$\Z_t \in \res(t)$};
\end{tikzpicture}
\end{center}
\caption{\label{fig:strong_metric} Processes $s,t$ are such that $s \not\STr t$ and $\TraceMetric(s,t) = 0.5$.}
\end{figure}

\begin{example}
\label{ex:strong_metric}
Consider processes $s,t$ in Fig.~\ref{fig:strong_metric}.
We have that $s \not \STr t$.
Notice that none of the resolutions for $s$ can exhibit both traces $ab$ and $ac$.
Thus, whenever we chose resolution $\Z_t \in \res(t)$ in Fig.~\ref{fig:strong_metric} for $t$, then there is no resolution for $s$ that can match $\Z_t$ on all traces.

Let us evaluate the trace distance between $s$ and $t$.
Since resolution $\Z_t$ for $t$ distinguishes the two processes, we start by evaluating its distance from the resolutions for $s$.
Consider the resolution $\Z_s \in \res(s)$ in Fig.~\ref{fig:strong_metric}.
By Def.~\ref{def:trace_distribution}, we have 
\[
\TD_{\Z_s} = 0.5 \delta_{ac} + 0.5 \delta_{a}
\qquad 
\TD_{\Z_t} = 0.5 \delta_{ac} + 0.5 \delta_{ab}.
\]
Clearly, $d_T(ac,ac) = 0$ and $d_T(ac,a) = d_T(ac,ab) = d_T(a,ab) = 1$.
Thus, by Def~\ref{def:trace_metric_det_res} we have
\[
\begin{array}{rlr}
D_T(\Z_s,\Z_t) ={} & \Kantorovich(d_T)(\TD_{\Z_s}, \TD_{\Z_t}) \\
={} & \min_{\w \in \W(\TD_{\Z_s}, \TD_{\Z_t})} \sum_{\alpha \in \support(\TD_{\Z_s}), \beta \in \support(\TD_{\Z_t})} \w(\alpha,\beta) \cdot d_T(\alpha,\beta) \\
={} & 0.5 \cdot d_T(ac,ac) + 0.5 \cdot d_T(a,ab) \\
={} & 0.5
\end{array}
\]
where to minimize the distance we have matched the two occurrences of the trace $ac$.
By similar calculations, one can easily obtain that 
\[
0.5 = D_T(\Z_t,\Z_s) = \sup_{\Z_2 \in \res(t)}\, \inf_{\Z_1 \in \res(s)}\, D_T(\Z_2,\Z_1).
\]
Moreover, it is immediate to check that whichever resolution for $s$ we choose, there is always a resolution for $t$ which is at trace distance $0$ from it, namely
\[
0 = \sup_{\Z_1 \in \res(s)}\, \inf_{\Z_2 \in \res(t)} D_T(\Z_1,\Z_2). 
\]
Therefore, we can conclude that 
\[
\TraceMetric(s,t) = \Hausdorff(D_T)(\res(s), \res(t)) = \max\{0,\, 0.5\} = 0.5
\]
\end{example}


\subsection{Weak trace metric}
\label{sec:weak_trace_metric}

To obtain the quantitative analogue of the weak trace equivalence, it is enough to adapt the notion of distance between traces (Definition~\ref{def:trace_metric_traces}) to the weak context.
The idea is that since silent steps cannot be observed, then they should not count on the trace distance.
Thus we introduce the notion of \emph{weak distance between traces} which is a $1$-bounded pseudometric over $\Acttau^{\star}$ having $\eqtrace$ as kernel.

\begin{definition}
[Weak distance between traces]
\label{def:weak_trace_metric_traces}
The \emph{weak distance between traces} $\dw \colon \Acttau^{\star} \times \Acttau^{\star} \to [0,1]$ is defined for any pair of traces $\alpha, \beta \in \Acttau^{\star}$ by
\[
\dw (\alpha, \beta) = 
\begin{cases}
0 & \text{ if } \alpha \eqtrace \beta \\
1 & \text{ otherwise.}
\end{cases}
\]
\end{definition}

It is clear that $\dw$ is a $1$-bounded pseudometric whose kernel is the equivalence of traces.

By substituting $d_T$ with $\dw$ in Definition~\ref{def:trace_metric_det_res} we obtain the notion of \emph{weak trace distance between resolutions}, denoted by the $1$-bounded pseudometric $\Dw$.
By lifting the relation of equivalence of traces $\eqtrace$ to an equivalence on probability distributions over traces $\eqtrace^{\dagger}$, we obtain that the kernel of $\Dw$ is given by the lifted equivalence on trace distributions, namely by the weak trace distribution equivalence of resolutions.
We can prove that our characterization of weak trace equivalence is equivalent to the one proposed in \cite{S95tr} in terms of trace distributions.

To simplify the reasoning in the upcoming proofs, let us define the weak version of the trace distribution given in Definition~\ref{def:trace_distribution}.
The idea is that we want to define a probability distribution on the traces executable by a resolution up-to trace equivalence.

\begin{definition}
\label{def:weak_trace_distribution}
Let $s \in \proc$ and consider any resolution $\Z \in \res(\proc)$, with $z = \corr{\Z}^{-1}(s)$.
We define the \emph{weak trace distribution} for $\Z$ as the function $\TDw_{\Z} \colon \Acttau^{\star} \to [0,1]$ defined by
$
\TDw_{\Z}(\alpha) = \Pr(\Cw_{\max}(z,\alpha)).
$
\end{definition}

\begin{lemma}
\label{lem:weak_trace_distribution}
For each $\Z \in \res(\proc)$, the weak trace distribution $\TDw_{\Z}$ is a probability distribution over $\Act^{\star}$.
\end{lemma}


\begin{proof}
The thesis follows by applying the same arguments used in the proof of Lemma~\ref{lem:trace_distribution} above.
\end{proof}


\begin{remark}
\label{rmk:tdw_probability}
Notice that $\TDw_{\_}$ is not a probability distribution over $\Acttau^{\star}$.
In fact it is enough to consider the simple resolution $\Z$ having $z$ as initial state for which the only transition in $\Z$ is $c = z \trans[a]_{\Z} \delta_{\nihl}$, namely $z$ executes $a$ and then with probability $1$ it ends its execution.
Clearly we have that $a \eqtrace \tau^{n} a \tau^{m}$ for all $n,m \ge 0$.
Let $\alpha_{n,m} = \tau^{n} a \tau^{m}$.
Then by definition of weak trace distribution (Definition~\ref{def:weak_trace_distribution}) we would have that $\TDw_{\Z}(\alpha_{n,m}) = \pr(\Cw_{\max}(z,\alpha_{n,m})) = \pr(c) = 1$, for all $n,m \ge 0$.
Clearly this would imply that $\sum_{\alpha \in \Acttau^{\star}} \TDw_{\Z}(\alpha) = \sum_{n,m \ge 0} \TDw_{\Z}(\alpha_{n,m}) > 1$.

However we remark hat $\TD_{\Z}$ is a probability distribution over $\Acttau^{\star}$ and thus $\Dw$ is well defined.
\end{remark}

We aim to show now that there is a strong relation between the trace distribution for a resolution and its weak version: they are equivalent distributions.

\begin{lemma}
\label{lem:equivalent_TD}
For each $\Z \in \res(\proc)$ we have that $\TD_{\Z} \eqtrace^{\dagger} \TDw_{\Z}$.
\end{lemma}


\begin{proof}
The thesis follows by applying the same arguments used in the proof of Lemma~\ref{lem:mimicking_equiv_weak_mimicking} below.
\end{proof}


\begin{proposition}
\label{prop:equivalent_definitions_of_weak}
For any pair of resolutions $\Z_1, \Z_2 \in \res(\proc)$, with initial states $z_1$ and $z_2$ resp., we have that $\TD_{\Z_1} \eqtrace^{\dagger} \TD_{\Z_2}$ if{f} $\pr(\Cw(z_1, \alpha)) = \pr(\Cw(z_2,\alpha))$ for all $\alpha \in \Act^{\star}$.
\end{proposition}


\begin{proof}
The thesis follows by the same arguments used in the proof of Theorem~\ref{thm:equivalent_resolutions_weak} below.
\end{proof}


\begin{proposition}
\label{prop:kernel_weak_trace_metric_det_res}
The kernel of $\Dw$ is weak trace distribution equivalence of resolutions.
\end{proposition}


\begin{proof}
The thesis follows by the same arguments used in the proof of Theorem~\ref{thm:kernel_of_Dtrddw} below.
\end{proof}


By substituting $D_T$ with $\Dw$ in Definition~\ref{def:trace_metric} we obtain the notion of \emph{weak trace metric}, denoted by the $1$-bounded pseudometric $\wTraceMetric$.

\begin{definition}
[Weak trace metric]
\label{def:weak_trace_metric}
The \emph{weak trace metric} $\wTraceMetric \colon \proc \times \proc \to [0,1]$ is defined for all $s,t \in \proc$ as 
\[
\wTraceMetric(s,t) = \Hausdorff(\Dw)(\res(s), \res(t)).
\]
\end{definition}

The kernel of the weak trace metric is weak trace equivalence.

\begin{proposition}
\label{prop:kernel_weak_trace_metric}
The kernel of $\wTraceMetric$ is probabilistic weak trace equivalence.
\end{proposition}


\begin{proof}
($\Rightarrow$)
Assume first that $\wTraceMetric(s,t) = 0$.
We aim to show that $s \WTr t$.
Since 
\begin{itemize}
\item by definition $\wTraceMetric(s,t) = \Hausdorff(\Dw)(\res(s),\res(t))$ and
\item the kernel of $\Dw$ is $\eqtrace^{\dagger}$ by Proposition~\ref{prop:kernel_weak_trace_metric_det_res}
\end{itemize} 
from $\wTraceMetric(s,t) = 0$ we can infer that $\res(s) \eqtrace^{\dagger} \res(t)$.
Then, by Proposition~\ref{prop:equivalent_definitions_of_weak} we can conclude that $s \WTr t$. 

($\Leftarrow$)
Assume now that $s \WTr t$.
We aim to show that this implies that $\wTraceMetric(s,t) = 0$.
By Proposition~\ref{prop:equivalent_definitions_of_weak} we have that $s \WTr t$ implies that $\res(s) \eqtrace^{\dagger} \res(t)$.
Since the kernel of $\Dw$ is given by $\eqtrace^{\dagger}$ (Proposition~\ref{prop:kernel_weak_trace_metric_det_res}), we can infer
\[
\wTraceMetric(s,t) = \Hausdorff(\Dw)(\res(s), \res(t)) = 0.
\]
\end{proof}


\section{Modal logics for traces}
\label{sec:logic_for_traces}

In this section we introduce two minimal modal logics $\LL$ and $\LLw$ that will allow us to characterize resp.\ the strong trace metric and its weak version, as well as the equivalences constituting their kernels.
The logic $\LL$ (and consequently $\LLw$) can be seen either as a simplified version of the modal logic $\logic$ from \cite{DD11}, which has been successfully employed in \cite{CGT16a} to characterize the bisimilarity metric \cite{DGJP04,BW01a,DCPP06}, or more simply as a probabilistic version of the logic characterizing the trace semantics in the fully nondeterministic case \cite{BFvG04}.


More precisely, $\LL$ consists of two classes of formulae.
The class $\LLt$ of \emph{trace formulae}, which are constituted by (finite) sequences of diamond operators and that will be used to represent traces, exactly as in the fully nondeterministic case.
Then, since we are treating traces as \emph{distributions over traces}, to capture the considered trace semantics we introduce the class $\LLd$ of \emph{trace distribution formulae}, which are defined by a probabilistic choice operator $\bigoplus$ as \emph{probability distributions over trace formulae}.

\begin{definition}
[Modal logic $\LL$]
\label{def:logic_LL}
The classes of \emph{trace distribution formulae} $\LLd$ and \emph{trace formulae} $\LLt$ over $\Act$ are defined by the following BNF-like grammar: 
\[
\LLd\colon\;  \Psi ::= \; 	\displaystyle \bigoplus_{i \in I} r_i \Phi_i 
\qquad\qquad
\LLt \colon\;  \Phi ::=  \;	\top \ | \ 	\diam{a}\Phi  
\]
where: 
\begin{inparaenum}[(i)]
\item 
$\Psi$ ranges over $\LLd$, 
\item 
$\Phi$ ranges over $\LLt$,
\item 
$a\in\Act$,
\item $I \neq \emptyset$ is a finite set of indexes,
\item the formulae $\Phi_i$ for $i \in I$ are pairwise distinct, namely $\Phi_i \neq \Phi_j$ for each $i,j \in I$ with $i \neq j$ and
\item for all $i\in I$ we have $r_i\in (0,1]$ and $\sum_{i\in I} r_i = 1$.
\end{inparaenum}
\end{definition}

To improve readability, we shall write $r_1\Phi_1 \oplus r_2 \Phi$ for $\bigoplus_{i \in I} r_i \Phi_i$ with $I=\{1,2\}$ and $\Phi$ for $\bigoplus_{i \in I} r_i \Phi_i$ with $I = \{i\}$, $r_i = 1$ and $\Phi_i = \Phi$.

\begin{definition}
[Depth]
\label{def:depth_LL}
The \emph{depth of trace distribution formulae} in $\LLd$ is defined as $\depth{\bigoplus_{i \in I} r_i \Phi_i} = \max_{i \in I} \depth{\Phi_i}$ where the \emph{depth of trace formulae} in $\LLt$ is defined by induction on their structure as 
\begin{inparaenum}[(i)]
\item $\depth{\top} = 0$ and
\item $\depth{\diam{a}\Phi} = 1+ \depth{\Phi}$.
\end{inparaenum}
\end{definition}

\begin{definition}
[Semantics of $\LLt$]
\label{def:satisfiability_LLt}
The \emph{satisfaction relation} $\models \, \subseteq \C \times \LLt$ is defined by structural induction over trace formulae in $\LLt$ by
\begin{itemize}
\item $c \models \top$ always;
\item $c \models \diam{a}\Phi$ if{f} $c = s \ctrans[a] c'$ for some computation $c'$ such that $c' \models \Phi$.
\end{itemize} 
\end{definition}

We say that a computation $c$ from a process $s$ is \emph{compatible} with the trace formula $\Phi \in \LLt$, notation $c \in \Ct(s,\Phi)$, if $c \models \Phi$ and $|c| = \depth{\Phi}$.

\begin{definition}
[Semantics of $\LLd$]
\label{def:satisfiability_LLd}
The \emph{satisfaction relation} $\models \, \subseteq \proc \times \LLd$ is defined by 
\begin{itemize}
\item $s \models \bigoplus_{i \in I} r_i \Phi_i$ if{f} there is a resolution $\Z \in \res(s)$ with $z = \corr{\Z}^{-1}(s)$ such that for each $i \in I$ we have
$
\pr(\Ct_{\max}(z,\Phi_i)) = r_i.
$
\end{itemize}
\end{definition}

We let $\LL(s)$ denote the set of formulae satisfied by process $s \in \proc$, namely $\LL(s) = \{\Psi \in \LLd \mid s \models \Psi\}$.

\begin{example}
\label{ex:logic}
Consider process $t$ in Fig.~\ref{fig:strong_metric}.
It is easy to verify that $t \models 0.5 \diam{a}\diam{c}\top \oplus 0.5 \diam{a}\diam{b}\top$.
In fact, if we consider the resolution $\Z_t \in \res(t)$ in the same figure, we have that the computation $c_1 = z_t \ctrans[a] z_{t_1} \ctrans[c] \mathrm{nil}$ is compatible with the trace formula $\diam{a}\diam{c}\top$ and that the computation $c_2 = z_t \ctrans[a] z_{t_2} \ctrans[b] \mathrm{nil}$ is compatible with the trace formula $\diam{a}\diam{b}\top$.
Moreover, we have $\pr(\Ct_{\max}(z_t,\diam{a}\diam{c}\top)) = 0.5$ and $\pr(\Ct_{\max}(z_t,\diam{a}\diam{b}\top)) = 0.5$.
\end{example}


The modal logic $\LLw$ differs from $\LL$ solely in the labels of the diamonds in $\LLwt$ which range over $\Acttau$ in place of $\Act$.
Hence, syntax and semantics of $\LLw$ directly follow from Definition~\ref{def:logic_LL} and Defs.~\ref{def:satisfiability_LLt}-\ref{def:satisfiability_LLd}, resp.

We let $\LLw(s)$ denote the set of formlae satisfied by process $s \in \proc$, namely $\LLw(s) = \{\Psi \in \LLwd \mid s \models \Psi\}$.

We introduce the $\LLw$-\emph{equivalence} which extends the equivalence of traces $\eqtrace$ to trace formulae.
\begin{definition}
[$\LLw$-equivalence of formulae]
\label{def:equivalence_of_formulae}
The relation of $\LLw$-\emph{equivalence of trace formulae} $\eqtrace \subseteq \LLwt \times \LLwt$ is the smallest equivalence relation satisfying
\begin{inparaenum}[(i)]
\item $\top \eqtrace \top$ and
\item $\diam{\A_1}\Phi_1 \eqtrace \diam{\A_2}\Phi_2$ if{f} 
\end{inparaenum}
\begin{itemize}
\item either $\A_1 = \tau$ and $\Phi_1 \eqtrace \diam{\A_2}\Phi_2$,
\item or $\A_2 = \tau$ and $\diam{\A_1}\Phi_1 \eqtrace \Phi_2$
\item or $\A_1 = \A_2$ and $\Phi_1 \eqtrace \Phi_2$.
\end{itemize}
Then, the relation of $\LLw$-\emph{equivalence of trace distribution formulae} $\eqtrace^{\dagger} \subseteq \LLwd \times \LLwd$ is obtained by lifting $\eqtrace$ to a relation on probability distributions over trace formulae.
\end{definition}

\begin{remark}
\label{rmk:equivalence_is_equality}
Clearly we have $\LLw_{/\eqtrace} = \LL$, namely the notion of $\eqtrace$ coincides with the equality of formulae when restricted to $(\LLd \times \LLd) \cup (\LLt \times \LLt)$.
Given any $\Psi_1,\Psi_2 \in \LLd$, we say that $\Psi_1 = \Psi_2$ if they express the same probability distribution over trace formulae.
\end{remark}

Notice that we are using the same symbol $\eqtrace$ to denote both the equivalence of traces and $\LLw$-equivalence.
The meaning will always be clear from the context.


\section{Logical characterization of relations}
\label{sec:char_of_trace_equivalence}

In this section we present the characterization of strong (resp.\ weak) trace equivalence by means of $\LL$ (resp.\ $\LLw$) (Theorem~\ref{thm:det_char} and Theorem~\ref{thm:weak_char}).
Following \cite{CGT16a}, we introduce the notion of \emph{mimicking formula} of a resolution as a formula expressing the trace distribution for that resolution.
Mimicking formulae characterize the (weak) trace distribution equivalence of resolutions: two resolutions are (weak) trace distribution equivalent if{f} their mimicking formulae are equal (resp.\ $\LLw$-equivalent) (Theorem~\ref{thm:equivalent_resolutions} and Theorem~\ref{thm:equivalent_resolutions_weak}).


The \emph{mimicking formula} of a resolution $\Z \in \res(\proc)$ is defined as a trace distribution formula assigning a positive weight only to the maximal traces of $\Z$.
Hence, we need to identify each maximal trace of $\Z$ with a proper trace formula.
This is achieved through the notion of \emph{tracing formula} of a trace.

\begin{definition}
[Tracing formula]
\label{def:tracing_formula}
Given any trace $\alpha \in \Act^{\star}$ we define the \emph{tracing formula} of $\alpha$, notation $\Phi_{\alpha} \in \LLt$, inductively on the structure of $\alpha$ as follows:
\[
\Phi_{\alpha} = 
\begin{cases}
\top & \text{ if } \alpha = \varepsilon \\
\diam{a}\Phi_{\alpha'} & \text{ if } \alpha = a \alpha', \alpha' \in \Act^{\star}.
\end{cases}
\]
\end{definition}

\begin{lemma}
\label{lem:computation_and_tracing_formula}
Let $s \in \proc$ and $\alpha \in \Act^{\star}$.
For each $c \in \C(s)$ we have $\tr(c) = \alpha$ if{f} $c \models \Phi_{\alpha}$ and $|c| = \depth{\Phi_{\alpha}}$.
\end{lemma}


\begin{proof}
($\Rightarrow$)
Assume first that $\tr(c)= \alpha$.
We aim to show that this implies that $|c| = \depth{\Phi_{\alpha}}$ and $c \models \Phi_{\alpha}$.
To this aim we proceed by induction over $|c|$.
\begin{itemize}
\item Base case $|c| = 0$, namely $c$ is the empty computation.
Since $\alpha = \tr(c)$, this gives that $\alpha = \varepsilon$ and therefore, by Def.~\ref{def:tracing_formula}, $\Phi_{\varepsilon} = \top$.
Then from Def.~\ref{def:depth_LL} we gather $\depth{\Phi_{\alpha}} = 0 = |c|$ and by Def.~\ref{def:satisfiability_LLt} we are guaranteed that $c \models \Phi_{\varepsilon}$.

\item Inductive step $|c| > 0$.
Assume wlog that $c = s \ctrans[a] c'$.
In particular this implies that $|c'| < |c|$.
Therefore, from $\alpha = \tr(c)$ we get that $\alpha$ must be of the form $\alpha = a\alpha'$ for $\alpha' = \tr(c')$.
By Def.~\ref{def:tracing_formula}, $\alpha = a\alpha'$ implies $\Phi_{\alpha} = \diam{a}\Phi_{\alpha'}$.
From $\alpha'= \tr(c')$ and the inductive hypothesis over $|c'|$ we get that $\depth{\Phi_{\alpha'}} = |c'|$ and $c' \models \Phi_{\alpha'}$.
This, taken together with $c = s \ctrans[a] c'$ gives $c \models \Phi_{\alpha}$.
Moreover, we have
\[
\depth{\Phi_{\alpha}} = \depth{\Phi_{\alpha'}} + 1 = |c'| + 1 = |c|
\]
thus concluding the proof.
\end{itemize}

($\Leftarrow$)
Assume now that $|c| = \depth{\Phi_{\alpha}}$ and $c \models  \Phi_{\alpha}$.
We aim to show that this implies that $\tr(c) = \alpha$, namely that $c$ is compatible with $\alpha$.
From $c \models \Phi_{\alpha}$ and the definition of tracing formula (Definition~\ref{def:tracing_formula}) we gather that the sequence of the labels of the first $\depth{\Phi_{\alpha}}$ execution steps of $c$ matches $\alpha$.
Moreover, $|c| = \depth{\Phi_{\alpha}}$ implies that those steps are actually the only execution steps for $c$.
Therefore we can immediately conclude that $\tr(c) = \alpha$.
\end{proof}


We remark that a computation $c$ is compatible with $\Phi_{\alpha}$ if{f} $c$ and $\alpha$ satisfy previous Lemma~\ref{lem:computation_and_tracing_formula}. 

\begin{definition}
[Mimicking formula]
\label{def:mimicking_formula_for_traces}
Consider any resolution $\Z \in \res(\proc)$ with initial state $z$.
We define the \emph{mimicking formula} of $\Z$, notation $\Psi_{\Z}$, as 
\[
\Psi_{\Z} = \bigoplus_{\alpha \in \tr(\C_{\max}(z))} \pr(\C_{\max}(z,\alpha)) \Phi_{\alpha}
\]  
where, for each $\alpha \in \tr(\C_{\max}(z))$, the formula $\Phi_{\alpha}$ is the tracing formula of $\alpha$.
\end{definition}

\begin{lemma}
\label{lem:mimic_well_defined}
For any resolution $\Z \in\res(\proc)$, the mimicking formula of $\Z$ is a well defined trace distribution formula. 
\end{lemma}


\begin{proof}
By definition of mimicking formula (Definition~\ref{def:mimicking_formula_for_traces}) we have
\[
\Psi_{\Z} = \bigoplus_{\alpha \in \tr(\C_{\max}(z))} \pr(\C_{\max}(z,\alpha)) \Phi_{\alpha}
\] 
where for each $\alpha \in \tr(\C_{\max}(z))$ the formula $\Phi_{\alpha}$ is the tracing formula of trace $\alpha$.

Hence, to prove that $\Psi_{\Z}$ is a well defined trace distribution formula we simply need to show that 
\[
\sum_{\alpha \in \tr(\C_{\max}(z))} \pr(\C_{\max}(z,\alpha)) = 1
\]
which follows by Lemma~\ref{lem:trace_distribution} by noticing that $\sum_{\alpha \in \tr(\C_{\max}(z))} \pr(\C_{\max}(z,\alpha)) = \sum_{\alpha \in \Act^{\star}} \pr(\C_{\max}(z,\alpha))$.
\end{proof}


\begin{example}
\label{ex:mimicking_res}
Consider the resolutions $\Z_s \in \res(s)$ and $\Z_t \in \res(t)$ for processes $s$ and $t$, resp., in Fig.~\ref{fig:strong_metric}.
The mimicking formulae for them are, resp.
\begin{align*}
& \Psi_{\Z_s} = 0.5 \diam{a}\diam{c}\top \oplus 0.5 \diam{a}\top \\
& \Psi_{\Z_t} = 0.5 \diam{a}\diam{c}\top \oplus 0.5 \diam{a}\diam{b}\top.
\end{align*}
\end{example}

The following results give us a first insight on the characterizing power of mimicking formulae: given $s \in \proc$, the set of the mimicking formulae of its resolutions constitutes the set of formulae satisfied by $s$.

\begin{proposition}
\label{prop:s_models_res_mimicking}
Let $s \in \proc$.
For each $\Z \in \res(s)$ it holds that $s \models \Psi_{\Z}$.
\end{proposition}


\begin{proof}
Let $\Z \in \res(s)$, with $z = \corr{\Z}^{-1}(s)$.
Hence, by definition of mimicking formula (Definition~\ref{def:mimicking_formula_for_traces}) we have that 
\[
\Psi_{\Z} = \bigoplus_{\alpha \in \tr(\C_{\max}(z))} \pr(\C_{\max}(z,\alpha)) \Phi_{\alpha}
\]
where, for each $\alpha \in \tr(\C_{\max}(z))$ we have that $\Phi_{\alpha}$ is the tracing formula of $\alpha$.
We need to show that $s \models \Psi_{\Z}$, namely we need to exhibit a resolution $\bar{\Z} \in \res(s)$, with $\bar{z} = \corr{\bar{\Z}}^{-1}(s)$, s.t.\ for each $\alpha \in \tr(\C_{\max}(z))$ we have that $\pr(\Ct(\bar{z}, \Phi_{\alpha})) = \pr(\C_{\max}(z,\alpha))$.
We aim to show that $\Z$ is such a resolution, namely that for each $\alpha \in \tr(\C_{\max}(z))$ we have
\[
\pr(\Ct_{\max}(z,\Phi_{\alpha})) = \pr(\C_{\max}(z,\alpha)).
\]
Let $\alpha \in \tr(\C_{\max}(z))$.
By definition we have 
\begin{align*}
\Ct_{\max}(z,\Phi_{\alpha}) ={} & \{c \in \C_{\max}(z) \mid c \models \Phi_{\alpha} \wedge |c| = \depth{\Phi_{\alpha}}\} \\
={} & \{c \in \C_{\max}(z) \mid \tr(c) = \alpha\} & \text{(by Lemma~\ref{lem:computation_and_tracing_formula})} \\
={} & \C_{\max}(z,\alpha) & \text{($\alpha \in \tr(\C_{\max}(z))$)}.
\end{align*}
Thus, we can conclude that
\[
\pr(\Ct_{\max}(z,\Phi_{\alpha})) =
\sum_{c \in \Ct_{\max}(z,\Phi_{\alpha})} \pr(c) =
\sum_{c \in \C_{\max}(z,\alpha)} \pr(c) =
\pr(\C_{\max}(z,\alpha)).
\]
\end{proof}


\begin{theorem}
\label{thm:LLs_is_cup_tracing_formula_resolution}
Let $s \in \proc$.
We have that 
$
\LL(s) = \{1\top\} \cup \{\Psi_{\Z} \mid \Z \in \res(s)\}.
$
\end{theorem}


\begin{proof}
From Proposition~\ref{prop:s_models_res_mimicking} and the definition of the relation $\models$ (Definition~\ref{def:satisfiability_LLd}) we can immediately infer that $\{\Psi_{\Z} \mid \Z \in \res(s)\} \subseteq \LL(s)$.
Moreover $1\top \in \LL(s)$ is immediate.
To conclude the proof we need to show that also the opposite inclusion holds, namely that $\LL(s) \setminus \{1\top\} \subseteq \{ \Psi_{\Z} \mid \Z \in \res(s)\}$.
To this aim, consider any $\Psi = \bigoplus_{i \in I} r_i \Phi_i$ and assume that $\Psi \in \LL(s)$.
We have to show that $\Psi$ is the mimicking formula of some resolution for $s$.
Since $s \models \Psi$, from Definition~\ref{def:satisfiability_LLd} we can infer that there is at least one resolution $\Z \in \res(s)$ with $z = \corr{\Z}^{-1}(s)$ s.t.\ for each $i \in I$ we have
$
\pr(\Ct_{\max}(z,\Phi_i)) = r_i.
$
We aim to prove that among the resolutions ensuring that $s \models \Psi$, there is a particular resolution $\Z \in \res(s)$ s.t.
\begin{equation}
\label{eq:thm_LL_res_proof_ob}
\Psi = \Psi_{\Z}.
\end{equation}
First of all we recall that by definition of trace distribution formula (Definition~\ref{def:logic_LL}), for each $i \in I$ we have $r_i > 0$ and moreover $\sum_{i \in I} r_i = 1$.
By definition of $\Ct$, we have that $c \in \Ct_{\max}(z,\Phi_i)$ if{f} $c \models \Phi_i$ and $|c| = \depth{\Phi_i}$, which by Lemma~\ref{lem:computation_and_tracing_formula} implies that $\Phi_i = \Phi_{\tr(c)}$.
Hence, let us consider the resolution $\Z \in \res(s)$ s.t. for each $i \in I$ we have $\Ct_{\max}(z,\Phi_i) \subseteq \C_{\max}(z)$, namely the resolution s.t. the computations compatible with the trace formulae $\Phi_i$ are all maximal.
Notice that the existence of such a resolution is guaranteed by $s \models \Psi$.
Since for each $c \in \Ct_{\max}(z,\Phi_i)$ we have $c \in \C_{\max}(z)$, we can infer that $\tr(c) \in \tr(\C_{\max}(z))$, namely $\Phi_i = \Phi_{\alpha}$ for some $\alpha \in \tr(\C_{\max}(z))$.
This gives that whenever $\Phi_i = \Phi_{\alpha}$, for some $\alpha \in \tr(\C_{\max}(z))$, then we can prove (as done in the proof of Proposition~\ref{prop:s_models_res_mimicking}) that
\begin{equation}
\label{eq:thm_LL_res_weights}
\pr(\Ct_{\max}(z, \Phi_i)) = \pr(\C_{\max}(z,\alpha)).
\end{equation}
Furthermore, we have obtained that 
$
\{\Phi_i \mid i \in I\} \subseteq \{ \Phi_{\alpha} \mid \alpha \in \tr(\C_{\max}(z))\}.
$

To prove Equation~\eqref{eq:thm_LL_res_proof_ob} we need to show that also the opposite inclusion holds.
Assume by contradiction that there is at least one $\beta \in \tr(\C_{\max}(z))$ s.t.\ there is no $i \in I$ with $\Phi_i = \Phi_{\beta}$.
Then we would have
\[
\begin{array}{rlr}
1 ={} & \sum_{i \in I} r_i \\
={} & \sum_{i \in I} \pr(\Ct_{\max}(z,\Phi_i)) \\
\le & \sum_{\alpha \in \tr(\C_{\max}(z)) \setminus \{\beta\}} \pr(\C_{\max}(z,\alpha)) & \text{(by Equation~\eqref{eq:thm_LL_res_weights})}\\
< & \sum_{\alpha \in \tr(\C_{\max}(z))} \pr(\C_{\max}(z,\alpha)) & \text{($\beta \in \tr(\C_{\max}(z))$ implies $\pr(\C_{\max}(z,\beta))>0$)} \\
={} & 1 & \text{(by Lemma~\ref{lem:trace_distribution})}
\end{array}
\]
which is a contradiction.
Hence we can conclude that 
$
\{\Phi_i \mid i \in I\} = \{ \Phi_{\alpha} \mid \alpha \in \tr(\C_{\max}(z))\}
$
and thus, due to Equation~\eqref{eq:thm_LL_res_weights}, that Equation~\eqref{eq:thm_LL_res_proof_ob} holds.
\end{proof}


\begin{remark}
In Theorem~\ref{thm:LLs_is_cup_tracing_formula_resolution}, $1\top$ is not included in the set of mimicking formulae of resolutions merely for sake of presentation, as $1\top$ is the mimicking formula of the resolution for $s$ in which no action is executed.
\end{remark}

The following theorem states that two resolutions are trace distribution equivalent if{f} their mimicking formulae are the same.

\begin{theorem}
\label{thm:equivalent_resolutions}
Let $s,t \in \proc$ and consider $\Z_s \in \res(s)$, with $z_s = \corr{\Z_s}^{-1}(s)$, and $\Z_t \in \res(t)$, with $z_t = \corr{\Z_t}^{-1}(t)$.
Then $\Psi_{\Z_s} = \Psi_{\Z_t}$ if{f} for all $\alpha \in \Act^{\star}$ it holds that $\pr(\C(z_s,\alpha)) = \pr(\C(z_t,\alpha))$.
\end{theorem}


\begin{proof}
($\Rightarrow$)
Assume first that $\Psi_{\Z_s} = \Psi_{\Z_t}$.
We aim to show that this implies $\pr(\C(z_s,\alpha)) = \pr(\C(z_t,\alpha))$ for all $\alpha \in \Act^{\star}$.
By definition of mimicking formula (Definition~\ref{def:mimicking_formula_for_traces}) we have
\[
\Psi_{\Z_s} = \bigoplus_{\alpha \in \tr(\C_{\max}(z_s))} \pr(\C_{\max}(z_s,\alpha)) \Phi_{\alpha}
\]  
where for each $\alpha \in \tr(\C_{\max}(z_s))$ the formula $\Phi_{\alpha}$ is the tracing formula of $\alpha$.
Analogously
\[
\Psi_{\Z_t} = \bigoplus_{\beta \in \tr(\C_{\max}(z_t))} \pr(\C_{\max}(z_t,\beta)) \Phi_{\beta}
\]  
where for each $\beta \in \tr(\C_{\max}(z_t))$ the formula $\Phi_{\beta}$ is the tracing formula of $\beta$.

Then from the assumption $\Psi_{\Z_s} = \Psi_{\Z_t}$ we gather
\begin{enumerate}
\item \label{item:same_max_computations}
$\tr(\C_{\max}(z_s)) = \tr(\C_{\max}(z_t))$;
\item \label{item:same_r}
from previous item~\ref{item:same_max_computations} we have that 
$
\pr(\C_{\max}(z_s,\alpha)) 
= 
\pr(\C_{\max}(z_t,\alpha))
$
for each $\alpha \in \tr(\C_{\max}(z_s))$.
\end{enumerate}
We notice that item~\ref{item:same_max_computations} above implies the stronger relation
\begin{equation}
\label{eq:same_computations}
\tr(\C(z_s)) = \tr(\C(z_t)).
\end{equation}
In fact each $\alpha \in \tr(\C(z_s))$ is either a trace in $\tr(\C_{\max}(z_s))$ or a proper prefix of a trace in that set.
In both cases item~\ref{item:same_max_computations} guarantees that each trace in $\tr(\C(z_s))$ has a matching trace in $\tr(\C(z_t))$ and viceversa.

Now, consider any $\alpha \in \Act^{\star}$.
We aim to show that $\pr(\C(z_s,\alpha)) = \pr(\C(z_t,\alpha))$.
For simplicity of presentation, we can distinguish two cases.
\begin{itemize}
\item $\pr(\C(z_s,\alpha)) = 0$.
In this case we have that no computation from $z_s$ is compatible with $\alpha$, namely there is no computation from $z_s$ for which the sequence of the labels of the execution steps matches $\alpha$.
More precisely, we have that $\alpha \not \in \tr(\C(z_s))$.
Since by Equation~\eqref{eq:same_computations} we have that $\tr(\C(z_s)) = \tr(\C(z_t))$ we can directly conclude that $\alpha \not \in \tr(\C(z_t))$, namely $\pr(\C(z_t,\alpha)) = 0$.

\item $\pr(\C(z_s,\alpha)) > 0$.
In this case we have that $\alpha \in \tr(\C(z_s))$ and by Equation~\eqref{eq:same_computations} we have that this implies that $\alpha \in \tr(\C(z_t))$.
Hence we are guaranteed that $\pr(\C(z_t,\alpha)) > 0$.
It remains to show that $\pr(\C(z_s,\alpha)) = \pr(\C(z_t,\alpha))$.
We have
\[
\begin{array}{rlr}
\pr(\C(z_s,\alpha)) ={} & \sum_{c \in P_{\max}(z_s,\alpha)} \pr(c) & \text{(by Lemma~\ref{lem:computations_same_sum})} \\
={} & \sum_{\beta \in \tr(P_{\max}(z_s,\alpha))} \pr(\C_{\max}(z_s, \beta)) & \text{(by def.\ of $P_{\max}$)} \\
={} & \sum_{\beta \in \tr(P_{\max}(z_s,\alpha))} \pr(\C_{\max}(z_t, \beta)) & \text{($P_{\max}(z_s,\alpha) \subseteq \C_{\max}(z_s)$ and item~\ref{item:same_r})} \\
={} & \sum_{\beta' \in \tr(P_{\max}(z_t,\alpha))} \pr(\C_{\max}(z_t,\beta')) & \text{(by Equation~\eqref{eq:same_computations})} \\
={} & \sum_{c' \in P_{\max}(z_t,\alpha)} \pr(c') & \text{(by def.\ of $P_{\max}$)} \\
={} & \pr(\C(z_t,\alpha)) & \text{(by Lemma~\ref{lem:computations_same_sum}).}
\end{array}
\]
\end{itemize}

($\Leftarrow$)
Assume now that for all $\alpha \in \Act^{\star}$ it holds that $\pr(\C(z_s,\alpha)) = \pr(\C(z_t,\alpha))$.
We aim to show that this implies that $\Psi_{\Z_s} = \Psi_{\Z_t}$.
By definition of mimicking formula (Definition~\ref{def:mimicking_formula_for_traces}) we have
\begin{align*}
\Psi_{\Z_s} = \bigoplus_{\alpha \in \tr(\C_{\max}(z_s))} \pr(\C_{\max}(z_s,\alpha)) \Phi_{\alpha} \\
\Psi_{\Z_t} = \bigoplus_{\beta \in \tr(\C_{\max}(z_t))} \pr(\C_{\max}(z_t,\beta)) \Phi_{\beta}.
\end{align*}
Therefore, to prove $\Psi_{\Z_s} = \Psi_{\Z_t}$ we need to show that
\begin{flalign}
& \label{eq:proof_obligation_same_traces}
\tr(\C_{\max}(z_s)) = \tr(\C_{\max}(z_t)) \\
& \label{eq:proof_obligation_same_weights}
\pr(\C_{\max}(z_s,\alpha)) = \pr(\C_{\max}(z_t,\alpha))
\text{ for each } \alpha \in \tr(\C_{\max}(z_s)).
\end{flalign}
First of all we notice that $\pr(\C(z_s,\alpha)) = \pr(\C(z_t,\alpha))$ for each $\alpha \in \Act^{\star}$ implies that $\tr(\C(z_s)) = \tr(\C(z_t))$.
This is due to the fact that by definition, given any $\alpha \in \Act^{\star}$, $\pr(\C(z_s,\alpha)) > 0$ if{f} there is at least one computation $c \in \C(z_s)$ s.t.\ $\alpha = \tr(c)$.
Since $\pr(\C(z_s,\alpha )) > 0$ implies $\pr(\C(z_t,\alpha)) > 0$ we can infer that for each $\alpha \in \tr(\C(z_s))$ there is at least one computation $c' \in \tr(\C(z_t))$ s.t.\ $\alpha = \tr(c')$, namely $\tr(\C(z_s)) \subseteq \tr(\C(z_t))$.
As the same reasoning can be applied symmetrically to each $\alpha \in \tr(\C(z_t))$, we can conclude that
\begin{equation}
\label{eq:same_set_traces}
\tr(\C(z_s)) = \tr(\C(z_t)).
\end{equation}
Next we aim to show that a similar result holds even if we restrict our attention to maximal computations, that is we aim to prove Equation~\eqref{eq:proof_obligation_same_traces}.

Let $\alpha \in \tr(\C_{\max}(z_s))$.
Notice that for this $\alpha$ we have $\C_{\max}(z_s,\alpha) \subseteq P_{\max}(z_s,\alpha)$.
Then we have
\[
\begin{array}{rlr}
\pr(\C(z_s,\alpha)) ={} & \sum_{c \in P_{\max}(z_s,\alpha)} \pr(c) & \text{(by Lemma~\ref{lem:computations_same_sum})} \\
={} & \sum_{c \in \C_{\max}(z_s,\alpha)} \pr(c) + \sum_{c' \in P_{\max}(z_s,\alpha) \setminus \C_{\max}(z_s,\alpha)} \pr(c').
\numberthis\label{eq:division}
\end{array}
\]
Moreover, by Lemma~\ref{lem:computations_same_sum} it holds that 
$
\pr(\C(z_t,\alpha)) = \sum_{c'' \in P_{\max}(z_t,\alpha)} \pr(c'').
$

Therefore, from $\pr(\C(z_s,\alpha)) = \pr(\C(z_t,\alpha))$ we gather that 
\begin{equation}
\label{eq:for_contradiction}
\sum_{c \in \C_{\max}(z_s,\alpha)} \pr(c) + \sum_{c' \in P_{\max}(z_s,\alpha) \setminus \C_{\max}(z_s,\alpha)} \pr(c') 
=
\sum_{c'' \in P_{\max}(z_t,\alpha)} \pr(c'').
\end{equation}
Assume by contradiction that $P_{\max}(z_t,\alpha) \cap \C_{\max}(z_t,\alpha) = \emptyset$, namely there is no maximal computation from $z_t$ which is compatible with $\alpha$.
Then for each action $a \in \Act$ consider the trace $\alpha a$ and define $\add_{z_s}(\alpha) = \{ a \in \Act \mid \alpha a \in \tr(\C(z_s))\}$.
From Equation~\eqref{eq:same_set_traces} we can directly infer that $\add_{z_s}(\alpha) = \add_{z_t}(\alpha)$.
Moreover, since we are assuming that no maximal computation from $z_t$ is compatible with $\alpha$, we get
\begin{flalign}
& \label{eq:for_contradiction_1}
\bigcup_{a \in \add_{z_s}(\alpha)} P_{\max}(z_s,\alpha a) = P_{\max}(z_s,\alpha) \setminus \C_{\max}(z_s,\alpha) \\
& \label{eq:for_contradiction_2}
\bigcup_{a \in \add_{z_t}(\alpha)} P_{\max}(z_t,\alpha a) = P_{\max}(z_t,\alpha)
\end{flalign}
where the unions are guaranteed to be disjoint (a single computation cannot be compatible with more than one trace $\alpha a$).
Furthermore, by Lemma~\ref{lem:computations_same_sum} we have that for each $a \in \add_{z_s}(\alpha)$
\[
\begin{array}{rlr}
\pr(\C(z_s, \alpha a)) = \sum_{c_1 \in P_{\max}(z_s, \alpha a)} \pr(c_1) \\
\pr(\C(z_t, \alpha a)) = \sum_{c_2 \in P_{\max}(z_t, \alpha a)} \pr(c_2)
\end{array}
\]
from which we get that for each $a \in \add_{z_s}(\alpha)$ it holds that
\begin{equation}
\label{eq:for_contradiction_3}
\sum_{c_1 \in P_{\max}(z_s, \alpha a)} \pr(c_1) = \sum_{c_2 \in P_{\max}(z_t, \alpha a)} \pr(c_2).
\end{equation}
Therefore we have that
\[
\begin{array}{rlr}
& \sum_{c \in \C_{\max}(z_s,\alpha)} \pr(c) + \sum_{c' \in P_{\max}(z_s, \alpha) \setminus \C_{\max}(z_s, \alpha)} \pr(c') \\
={} & \sum_{c \in \C_{\max}(z_s, \alpha)} \pr(c) + \sum_{c' \in \bigcup_{a \in \add_{z_s}(\alpha)} P_{\max}(z_s, \alpha a)} \pr(c') & \text{(by Equation~\eqref{eq:for_contradiction_1})} \\
={} & \sum_{c \in \C_{\max}(z_s,\alpha)} \pr(c) + \sum_{a \in \add_{z_s}(\alpha)} \Big( \sum_{c' \in P_{\max}(z_s, \alpha a)} \pr(c') \Big) & \text{(disjoint union)} \\
={} & \sum_{c \in \C_{\max}(z_s, \alpha)} \pr(c) + \sum_{a \in \add_{z_s}(\alpha)} \Big( \sum_{c'' \in P_{\max}(z_t,\alpha a)} \pr(c'') \Big) & \text{(by Equation~\eqref{eq:for_contradiction_3})} \\
={} & \sum_{c \in \C_{\max}(z_s, \alpha)} \pr(c) + \sum_{c'' \in \bigcup_{a \in \add_{z_t}(\alpha)}P_{\max}(z_t, \alpha a)} \pr(c'') & \text{($\add_{z_s}(\alpha) = \add_{z_t}(\alpha)$)} \\
={} & \sum_{c \in \C_{\max}(z_s, \alpha)} \pr(c)  + \sum_{c'' \in P_{\max}(z_t, \alpha)} \pr(c'') & \text{(by Equation~\eqref{eq:for_contradiction_2})}.
\end{array}
\]
Thus we have obtained that 
\[
\sum_{c \in \C_{\max}(z_s,\alpha)} \pr(c) + \sum_{c' \in P_{\max}(z_s, \alpha) \setminus \C_{\max}(z_s, \alpha)} \pr(c')
=
\sum_{c \in \C_{\max}(z_s, \alpha)} \pr(c)  + \sum_{c'' \in P_{\max}(z_t, \alpha)} \pr(c'')
\]
which, since by the choice of $\alpha$ we have that $\sum_{c \in \C_{\max}(z_s, \alpha)} \pr(c) > 0$, is in contradiction with Equation~\eqref{eq:for_contradiction}.
Therefore, we have obtained that whenever $\alpha \in \tr(\C_{\max}(z_s))$ then there is at least one maximal computation $c$ from $z_t$ s.t.\ $\alpha = \tr(c)$, that is $\tr(\C_{\max}(z_s)) \subseteq \tr(\C_{\max}(z_t))$.
Since the same reasoning can be applied symmetrically to each $\alpha \in \tr(\C_{\max}(z_t))$ we gather that also $\tr(\C_{\max}(z_t)) \subseteq \tr(\C_{\max}(z_s))$ holds.
The two inclusions give us Equation~\eqref{eq:proof_obligation_same_traces}. 

Finally, we aim to prove Equation~\eqref{eq:proof_obligation_same_weights}.
Let $\alpha \in \tr(\C_{\max}(z_s))$.
We can distinguish two cases.
\begin{itemize}
\item $|\alpha| = \depth{z_s}$.
First of all we notice that from Equation~\eqref{eq:same_set_traces} and the assumption $\pr(\C(z_s,\beta)) = \pr(\C(z_t,\beta))$ for each $\beta \in \Act^{\star}$, we can infer that $|\alpha| = \depth{z_t}$.
Hence, we have 
\[
\pr(\C_{\max}(z_s,\alpha)) = \pr(\C(z_s,\alpha)) = \pr(\C(z_t,\alpha)) = \pr(\C_{\max}(z_t,\alpha)).
\]
\item $|\alpha| < \depth{z_s}$.
Then we have 
\[
\begin{array}{rlr}
& \pr(\C_{\max}(z_s, \alpha)) \\
={} & \sum_{c \in \C_{\max}(z_s, \alpha)} \pr(c) \\
={} & \sum_{c' \in P_{\max}(z_s, \alpha)} \pr(c') - \sum_{c'' \in P_{\max}(z_s, \alpha) \setminus \C_{\max}(z_s, \alpha)} \pr(c'') \\
={} & \pr(\C(z_s, \alpha)) - \sum_{c'' \in P_{\max}(z_s, \alpha) \setminus \C_{\max}(z_s, \alpha)} \pr(c'') \\
={} & \pr(\C(z_t, \alpha)) - \sum_{c'' \in P_{\max}(z_s, \alpha) \setminus \C_{\max}(z_s, \alpha)} \pr(c'') \\
={} & \sum_{c''' \in P_{\max}(z_t, \alpha)} \pr(c''') - \sum_{c'' \in P_{\max}(z_s, \alpha) \setminus \C_{\max}(z_s, \alpha)} \pr(c'') \\
={} & \sum_{c_1 \in \C_{\max}(z_t, \alpha)} \pr(c_1) + 
\sum_{c_2 \in P_{\max}(z_t, \alpha) \setminus \C_{\max}(z_t, \alpha)} \pr(c_2) - \sum_{c'' \in P_{\max}(z_s, \alpha) \setminus \C_{\max}(z_s, \alpha)} \pr(c'')\\
={} & \sum_{c_1 \in \C_{\max}(z_t, \alpha)} \pr(c_1) + 
\sum_{c_2 \in \bigcup_{b \in \add_{z_t}(\alpha)} P_{\max}(z_t, \alpha b)} \pr(c_2) - \sum_{c'' \in \bigcup_{b \in \add_{z_s}(\alpha)}P_{\max}(z_s, \alpha b)} \pr(c'') \\
={} & \sum_{c_1 \in \C_{\max}(z_t, \alpha)} \pr(c_1) + 
\sum_{b \in \add_{z_t}(\alpha)} \Big( \sum_{c_2 \in P_{\max}(z_t, \alpha b)} \pr(c_2) \Big) + \\
& - \sum_{b \in \add_{z_s}(\alpha)}( \sum_{c'' \in P_{\max}(z_s, \alpha b)} \pr(c'') ) \\
={} & \sum_{c_1 \in \C_{\max}(z_t, \alpha)} \pr(c_1) + 
\sum_{b \in \add_{z_t}(\alpha)} \pr(\C(z_t, \alpha b)) - \sum_{b \in \add_{z_s}(\alpha)} \pr(\C(z_s, \alpha b)) \\
={} & \sum_{c_1 \in \C_{\max}(z_t, \alpha)} \pr(c_1) \\
={} & \pr(\C_{\max}(z_t, \alpha))
\end{array}
\]
where
\begin{itemize}
\item the second and the sixth steps follow by Equation~\eqref{eq:division};
\item the third, fifth and ninth steps follow by Lemma~\ref{lem:computations_same_sum};
\item the fourth step follows by $\pr(\C(z_s, \alpha)) = \pr(\C(z_t, \alpha))$;
\item the seventh step follows by Equation~\eqref{eq:for_contradiction_1};
\item the tenth step follows by $\add_{z_s}(\alpha) = \add_{z_t}(\alpha)$ (given by Equation~\eqref{eq:same_set_traces}) and the initial assumption which guarantees that for each $b \in \add_{z_s}(\alpha)$, $\pr(\C(z_s, \alpha b)) = \pr(\C(z_t, \alpha b))$. 
\end{itemize}
\end{itemize}
\end{proof}


Then we can derive the characterization result for the strong case: two processes $s,t$ are strong trace equivalent if{f} they satisfy the same formulae in $\LL$.

\begin{theorem}
\label{thm:det_char}
For all $s,t \in \proc$ we have that $s \STr t$ if{f} $\LL(s) = \LL(t)$.
\end{theorem}


\begin{proof}
($\Rightarrow$) 
Assume first that $s \STr t$.
We aim to sow that this implies that $\LL(s) = \LL(t)$.
By Definition~\ref{def:prob_trace_dist} $s \STr t$ implies that 
\begin{enumerate}[(i)]
\item \label{item:thm_det_char_1}
for each resolution $\Z_s \in \res(s)$, with $z_s = \corr{\Z_s}^{-1}(s)$, there is a resolution $\Z_t \in \res(t)$, with $z_t = \corr{\Z_t}^{-1}(t)$, s.t.\ for each $\alpha \in \Act^{\star}$ we have $\pr(\C(z_s,\alpha)) = \pr(\C(z_t,\alpha))$;
\item \label{item:thm_det_char_2}
for each resolution $\Z_t \in \res(t)$, with $z_t = \corr{\Z_t}^{-1}(t)$, there is a resolution $\Z_s \in \res(s)$, with $z_s = \corr{\Z_s}^{-1}(s)$, s.t.\ for each $\alpha \in \Act^{\star}$ we have $\pr(\C(z_s,\alpha)) = \pr(\C(z_t,\alpha))$.
\end{enumerate}
Consider any $\Z_s \in \res(s)$, with $z_s = \corr{\Z_s}^{-1}(s)$, and let $\Z_t \in \res(t)$, with $z_t = \corr{\Z_t}^{-1}(t)$, be any resolution of $t$ satisfying item~\eqref{item:thm_det_char_1} above.
By Theorem~\ref{thm:equivalent_resolutions}, $\pr(\C(z_s,\alpha)) = \pr(\C(z_t,\alpha))$ for all $\alpha \in \Act^{\star}$ implies that $\Psi_{\Z_s} = \Psi_{\Z_t}$.
More precisely, we have that
\begin{equation}
\label{eq:thm_det_char_1}
\text{for each } \Z_s \in \res(s) \text{ there is } \Z_t \in \res(t) \text{ s.t.\ } \Psi_{\Z_s} = \Psi_{\Z_t}.
\end{equation}
Symmetrically, item~\eqref{item:thm_det_char_2} above taken together with Theorem~\ref{thm:equivalent_resolutions} gives that 
\begin{equation}
\label{eq:thm_det_char_2}
\text{for each } \Z_t \in \res(t) \text{ there is a } \Z_s \in \res(s) \text{ s.t.\ } \Psi_{\Z_t} = \Psi_{\Z_s}. 
\end{equation}
Therefore, from Equations~\eqref{eq:thm_det_char_1} and~\eqref{eq:thm_det_char_2} we gather 
\begin{equation}
\label{eq:thm_det_char_3}
\{\Psi_{\Z_s} \mid \Z_s \in \res(s)\} = \{\Psi_{\Z_t} \mid \Z_t \in \res(t)\}.
\end{equation}
By Theorem~\ref{thm:LLs_is_cup_tracing_formula_resolution} we have that $\LL(s) = \{1\top\} \cup \{\Psi_{\Z_s} \mid \Z_s \in \res(s)\}$ and similarly $\LL(t) = \{1\top\} \cup \{\Psi_{\Z_t} \mid \Z_t \in \res(t)\}$.
Therefore, from Equation~\eqref{eq:thm_det_char_3} we can conclude that
$
\LL(s) = \LL(t).
$

($\Leftarrow$)
Assume now that $\LL(s) = \LL(t)$.
We aim to show that this implies that $s \STr t$.
By Theorem~\ref{thm:LLs_is_cup_tracing_formula_resolution} we have that $\LL(s) = \{1\top\} \cup \{ \Psi_{\Z_s} \mid \Z_s \in \res(s)\}$ and analogously $\LL(t) = \{1\top\} \cup \{\Psi_{\Z_t} \mid \Z_t \in \res(t)\}$.
Hence, from the assumption we can infer that
$
\{\Psi_{\Z_s} \mid \Z_s \in \res(s)\} = \{\Psi_{\Z_t} \mid \Z_t \in \res(t)\}.
$

Clearly the equality between the two sets implies that 
\begin{itemize}
\item for each $\Z_s \in \res(s)$ there is a $\Z_t \in \res(t)$ s.t.\ $\Psi_{\Z_s} = \Psi_{\Z_t}$ and 
\item for each $\Z_t \in \res(t)$ there is a $\Z_s \in \res(s)$ s.t.\ $\Psi_{\Z_t} = \Psi_{\Z_s}$.
\end{itemize}
By applying Theorem~\ref{thm:equivalent_resolutions} to the two items above we obtain that
\begin{itemize}
\item for each resolution $\Z_s \in \res(s)$, with $z_s = \corr{\Z_s}^{-1}(s)$, there is a resolution $\Z_t \in \res(t)$, with $z_t = \corr{\Z_t}^{-1}(t)$, s.t.\ for each $\alpha \in \Act^{\star}$ we have $\pr(\C(z_s,\alpha)) = \pr(\C(z_t,\alpha))$;
\item for each resolution $\Z_t \in \res(t)$, with $z_t = \corr{\Z_t}^{-1}(t)$, there is a resolution $\Z_s \in \res(s)$, with $z_s = \corr{\Z_s}^{-1}(s)$, s.t.\ for each $\alpha \in \Act^{\star}$ we have $\pr(\C(z_s,\alpha)) = \pr(\C(z_t,\alpha))$;
\end{itemize}
from which we can conclude that $s \STr t$.
\end{proof}


The notions of \emph{tracing formula} and \emph{mimicking formula} and the related results Lemma~\ref{lem:computation_and_tracing_formula}, Lemma~\ref{lem:mimic_well_defined}, Proposition~\ref{prop:s_models_res_mimicking} and Theorem~\ref{thm:LLs_is_cup_tracing_formula_resolution} can be easily extended to the weak case by extending the set of traces $\Act^{\star}$ to the set $\Acttau^{\star}$. 

The following theorem gives the characterization of weak trace distribution equivalence: two resolutions are weak trace distribution equivalent if{f} their mimicking formulae are $\LLw$-equivalent.

To simplify the upcoming proofs, we introduce an alternative version of the \emph{weak mimicking formula}, which captures the weak trace distribution (see Definition~\ref{def:weak_trace_distribution}) of resolutions.

\begin{definition}
\label{def:weak_mimicking}
Consider any resolution $\Z \in \res(\proc)$ with initial state $z$.
We define the \emph{weak mimicking formula} of $\Z$ as the trace distribution formula $\Psiw_{\Z}$ given by
\[
\Psiw_{\Z} = \bigoplus_{\alpha \in \trw(\C_{\max}(z))} \pr(\Cw_{\max}(z,\alpha)) \Phi_{\alpha}
\]  
where, for each $\alpha \in \trw(\C_{\max}(z))$, the formula $\Phi_{\alpha}$ is the tracing formula of $\alpha$.
\end{definition}

Notice that from the definitions of $\Cw_{\max}(\_,\_)$ and $\trw(\_)$ we can infer that $\Psiw_{\_}$ represents a trace distribution formula over the quotient space of $\LLw$ wrt.\ $\eqtrace$, that is $\Psiw_{\_} \in \LLd$.

\begin{lemma}
\label{lem:mimicking_equiv_weak_mimicking}
For each $\Z \in \res(\proc)$ it holds that $\Psi_{\Z} \eqtrace^{\dagger} \Psiw_{\Z}$.
\end{lemma}


\begin{proof}
Consider $\Z \in \res(\proc)$ with initial state $z$.
First of all we recall that by definition of mimicking formula (Definition~\ref{def:mimicking_formula_for_traces}) we have
\[
\Psi_{\Z} = \bigoplus_{\alpha \in \tr(\C_{\max}(z))} \pr(\C_{\max}(z,\alpha)) \Phi_{\alpha}
\]
where for each $\alpha \in \tr(\C_{\max}(z))$, the formula $\Phi_{\alpha}$ is the tracing formula of $\alpha$.
By definition of weak mimicking formula (Definition~\ref{def:weak_mimicking}) we have
\[
\Psiw_{\Z} = \bigoplus_{\beta \in \trw(\C_{\max}(z))} \pr(\Cw_{\max}(z,\beta)) \Phi_{\beta}
\]
where for each $\beta \in \trw(\C_{\max}(z))$, the formula $\Phi_{\beta}$ is the tracing formula of $\beta$.
Moreover, we have that for each $\beta \in \trw(\C_{\max}(z))$
\[
\begin{array}{rlr}
\pr(\Cw_{\max}(z,\beta)) ={} & \sum_{c \in \Cw_{\max}(z,\beta)} \pr(c) \\
={} & \sum_{c \in \C_{\max}(z) \text{ s.t. } \tr(c) \eqtrace \beta} \pr(c) \\
={} & \sum_{\alpha \in \tr(\C_{\max}(z)) \text{ s.t. } \alpha \eqtrace \beta} \pr(\C_{\max}(z,\alpha)). 
\end{array}
\]
Furthermore, by definition of tracing formula (Definition~\ref{def:tracing_formula}) and of $\eqtrace$ (Definition~\ref{def:eqtrace}), it is immediate that $\alpha \eqtrace \beta$ if{f} $\Phi_{\alpha} \eqtrace \Phi_{\beta}$, for each $\alpha, \beta \in \Act^{\star}$.
For simplicity, we denote by $\alpha_{\beta}$ each $\alpha \in \tr(\C_{\max}(z))$ s.t.\ $\alpha \eqtrace \beta$ for some $\beta \in \trw(\C_{\max}(z))$.
Notice that by construction of $\trw(\_)$, no trace $\alpha \in \tr(\C_{\max}(z))$ can be equivalent to more than one $\beta \in \trw(\C_{\max}(z))$.
Therefore, we have obtained that
\[
\begin{array}{rlr}
\Psiw_{\Z} ={} & \bigoplus_{\beta \in \trw(\C_{\max}(z))} \pr(\Cw_{\max}(z,\beta)) \Phi_{\beta} \\
\eqtrace^{\dagger} & \bigoplus_{\beta \in \trw(\C_{\max}(z)) \atop \alpha_{\beta} \in \tr(\C_{\max}(z))} \pr(\C_{\max}(z, \alpha_{\beta})) \Phi_{\alpha_{\beta}} \\
\eqtrace^{\dagger} & \bigoplus_{\alpha \in \tr(\C_{\max}(z))} \pr(\C_{\max}(z,\alpha)) \Phi_{\alpha} \\
={} & \Psi_{\Z}.
\end{array}
\]
\end{proof}


\begin{theorem}
\label{thm:equivalent_resolutions_weak}
Let $s,t \in \proc$ and consider $\Z_s \in \res(s)$, with $z_s = \corr{\Z_s}^{-1}(s)$, and $\Z_t \in \res(t)$, with $z_t = \corr{\Z_t}^{-1}(t)$.
Then $\Psi_{\Z_s} \eqtrace^{\dagger} \Psi_{\Z_t}$ if{f} for all $\alpha \in \Act^{\star}$ it holds that $\pr(\Cw(z_s,\alpha)) = \pr(\Cw(z_t,\alpha))$.
\end{theorem}


\begin{proof}
($\Rightarrow$)
Assume first that $\Psi_{\Z_s} \eqtrace^{\dagger} \Psi_{\Z_t}$.
We aim to show that $\pr(\Cw(z_s,\alpha)) = \pr(\Cw(z_t,\alpha))$ for all $\alpha \in \Act^{\star}$.
By Lemma~\ref{lem:mimicking_equiv_weak_mimicking} we have that
\[
\Psi_{\Z_s} \eqtrace^{\dagger} \Psiw_{\Z_s} 
\quad\text{ and }\quad
\Psi_{\Z_t} \eqtrace^{\dagger} \Psiw_{\Z_t}.
\]  
Thus, $\Psi_{\Z_s} \eqtrace^{\dagger} \Psi_{\Z_t}$ implies $\Psiw_{\Z_s} \eqtrace^{\dagger} \Psiw_{\Z_t}$.
Hence the prove the proof obligation, it is enough to prove that
\begin{equation}
\label{eq:thm_weak_equivalent_proof_obligation}
\Psiw_{\Z_s} \eqtrace^{\dagger} \Psiw_{\Z_t} \;\text{ implies }\; \pr(\Cw(z_s,\alpha)) = \pr(\Cw(z_t,\alpha)) \text{ for each } \alpha \in \Act^{\star}.
\end{equation}
From $\Psiw_{\Z_s} \eqtrace^{\dagger} \Psiw_{\Z_t}$ we get that 
\[
\Psiw_{\Z_t} = \bigoplus_{\alpha \in \trw(\C_{\max}(z_s)) \atop \beta_{\alpha} \in \trw(\C_{\max}(z_t)) \cap [\alpha]_{\mathrm{w}}}
\pr(\Cw_{\max}(z_t,\beta_{\alpha})) \Phi_{\beta_{\alpha}}
\]
where, for each $\alpha \in \trw(\C_{\max}(z_s))$, $\sum_{\beta_{\alpha} \in \trw(\C_{\max}(z_t)) \cap [\alpha]_{\mathrm{w}}} \pr(\Cw_{\max}(z_t,\beta_{\alpha})) = \pr(\Cw_{\max}(z_s,\alpha))$ and $\Phi_{\beta_{\alpha}} \eqtrace \Phi_{\alpha}$ for each $\beta_{\alpha} \in \trw(\C_{\max}(z_t)) \cap [\alpha]_{\mathrm{w}}$.

We notice that by definition the elements of $\trw(\C_{\max}(z_t))$ represent distinct equivalence classes with respect to $\eqtrace$.
Thus we are guaranteed that for each $\alpha \in \trw(\C_{\max}(z_t)) \cap [\alpha]_{\mathrm{w}}$ contains a single trace $\beta_{\alpha}$.
Therefore, in this particular case, $\Psiw_{\Z_s} \eqtrace^{\dagger} \Psiw_{\Z_t}$ is equivalent to say that $\Psiw_{\Z_s} = \Psiw_{\Z_t}$.
Moreover, since the representative of the equivalence classes wrt $\eqtrace$ can always be chosen in $\Act^{\star}$, we can always construct the sets $\trw(\C_{\max}(z_s))$ and $\trw(\C_{\max}(z_t))$ in such a way that $\trw(\C_{\max}(z_s)) \cap \Act^{\star} = \trw(\C_{\max}(z_s))$ and $\trw(\C_{\max}(z_t)) \cap \Act^{\star} = \trw(\C_{\max}(z_t))$. 
Hence, the same argumentations presented in the first part of the proof of Theorem~\ref{thm:equivalent_resolutions} allow us to prove the proof obligation Equation~\eqref{eq:thm_weak_equivalent_proof_obligation}.

($\Leftarrow$)
Assume now that for all $\alpha \in \Act^{\star}$ it holds that $\pr(\Cw(z_s,\alpha)) = \pr(\Cw(z_t,\alpha))$.
We aim to show that this implies that $\Psi_{\Z_s} \eqtrace^{\dagger} \Psi_{\Z_t}$.
To this aim we show that the assumption $\pr(\Cw(z_s,\alpha)) = \pr(\Cw(z_t,\alpha))$ for all $\alpha \in \Act^{\star}$ implies $\Psiw_{\Z_s} = \Psiw_{\Z_t}$.
This follows from the same argumentations presented in the second part of the proof of Theorem~\ref{thm:equivalent_resolutions}.
Then, since from Lemma~\ref{lem:mimicking_equiv_weak_mimicking} we have $\Psi_{\Z_s} \eqtrace^{\dagger} \Psiw_{\Z_s}$ and $\Psi_{\Z_t} \eqtrace^{\dagger} \Psiw_{\Z_t}$, we can conclude that $\Psi_{\Z_s} \eqtrace^{\dagger} \Psi_{\Z_t}$ as required.
\end{proof}


Then we can derive the characterization result for the weak case: two processes $s,t$ are weak trace equivalent if{f} they satisfy equivalent formulae in $\LLw$.

\begin{theorem}
\label{thm:weak_char}
For all $s,t \in \proc$ we have that $s \WTr t$ if{f} $\LLw(s) \eqtrace^{\dagger} \LLw(t)$.
\end{theorem}


\begin{proof}
($\Rightarrow$) 
Assume first that $s \WTr t$.
We aim to sow that this implies that $\LLw(s) \eqtrace^{\dagger} \LLw(t)$.
By Definition~\ref{def:weak_trace_equivalence} $s \WTr t$ implies that 
\begin{enumerate}[(i)]
\item \label{item:thm_weak_char_1}
for each resolution $\Z_s \in \res(s)$, with $z_s = \corr{\Z_s}^{-1}(s)$, there is a resolution $\Z_t \in \res(t)$, with $z_t = \corr{\Z_t}^{-1}(t)$, s.t.\ for each $\alpha \in \Act^{\star}$ we have $\pr(\Cw(z_s,\alpha)) = \pr(\Cw(z_t,\alpha))$;
\item \label{item:thm_weak_char_2}
for each resolution $\Z_t \in \res(t)$, with $z_t = \corr{\Z_t}^{-1}(t)$, there is a resolution $\Z_s \in \res(s)$, with $z_s = \corr{\Z_s}^{-1}(s)$, s.t.\ for each $\alpha \in \Act^{\star}$ we have $\pr(\Cw(z_s,\alpha)) = \pr(\Cw(z_t,\alpha))$.
\end{enumerate}
Consider any $\Z_s \in \res(s)$, with $z_s = \corr{\Z_s}^{-1}(s)$, and let $\Z_t \in \res(t)$, with $z_t = \corr{\Z_t}^{-1}(t)$, be any resolution of $t$ satisfying item~\eqref{item:thm_weak_char_1} above.
By Theorem~\ref{thm:equivalent_resolutions_weak}, $\pr(\Cw(z_s,\alpha)) = \pr(\Cw(z_t,\alpha))$ for all $\alpha \in \Act^{\star}$ implies that $\Psi_{\Z_s} \eqtrace^{\dagger} \Psi_{\Z_t}$.
More precisely, we have that
\begin{equation}
\label{eq:thm_weak_char_1}
\text{for each } \Z_s \in \res(s) \text{ there is } \Z_t \in \res(t) \text{ s.t.\ } \Psi_{\Z_s} \eqtrace^{\dagger} \Psi_{\Z_t}.
\end{equation}
Symmetrically, item~\eqref{item:thm_weak_char_2} above taken together with Theorem~\ref{thm:equivalent_resolutions_weak} gives that 
\begin{equation}
\label{eq:thm_weak_char_2}
\text{for each } \Z_t \in \res(t) \text{ there is a } \Z_s \in \res(s) \text{ s.t.\ } \Psi_{\Z_t} \eqtrace^{\dagger} \Psi_{\Z_s}. 
\end{equation}
Therefore, from Equations~\eqref{eq:thm_weak_char_1} and~\eqref{eq:thm_weak_char_2} we gather 
\begin{equation}
\label{eq:thm_weak_char_3}
\{\Psi_{\Z_s} \mid \Z_s \in \res(s)\} \eqtrace^{\dagger} \{\Psi_{\Z_t} \mid \Z_t \in \res(t)\}.
\end{equation}
By Theorem~\ref{thm:LLs_is_cup_tracing_formula_resolution} we have that $\LLw(s) = \{1\top\} \cup \{\Psi_{\Z_s} \mid \Z_s \in \res(s)\}$ and similarly $\LLw(t) = \{1\top\} \cup \{\Psi_{\Z_t} \mid \Z_t \in \res(t)\}$.
Therefore, from Equation~\eqref{eq:thm_weak_char_3} we can conclude that
$
\LLw(s) \eqtrace^{\dagger} \LLw(t).
$

($\Leftarrow$)
Assume now that $\LLw(s) \eqtrace^{\dagger} \LLw(t)$.
We aim to show that this implies that $s \WTr t$.
By Theorem~\ref{thm:LLs_is_cup_tracing_formula_resolution} we have that $\LLw(s) = \{1\top\} \cup \{ \Psi_{\Z_s} \mid \Z_s \in \res(s)\}$ and analogously $\LLw(t) = \{1\top\} \cup \{\Psi_{\Z_t} \mid \Z_t \in \res(t)\}$.
Hence, from the assumption we can infer that
$
\{\Psi_{\Z_s} \mid \Z_s \in \res(s)\} \eqtrace^{\dagger} \{\Psi_{\Z_t} \mid \Z_t \in \res(t)\}.
$

Clearly the equivalence between the two sets implies that 
\begin{itemize}
\item for each $\Z_s \in \res(s)$ there is a $\Z_t \in \res(t)$ s.t.\ $\Psi_{\Z_s} \eqtrace^{\dagger} \Psi_{\Z_t}$ and 
\item for each $\Z_t \in \res(t)$ there is a $\Z_s \in \res(s)$ s.t.\ $\Psi_{\Z_t} \eqtrace^{\dagger} \Psi_{\Z_s}$.
\end{itemize}
By applying Theorem~\ref{thm:equivalent_resolutions_weak} to the two items above we obtain that
\begin{itemize}
\item for each resolution $\Z_s \in \res(s)$, with $z_s = \corr{\Z_s}^{-1}(s)$, there is a resolution $\Z_t \in \res(t)$, with $z_t = \corr{\Z_t}^{-1}(t)$, s.t.\ for each $\alpha \in \Act^{\star}$ we have $\pr(\Cw(z_s,\alpha)) = \pr(\Cw(z_t,\alpha))$;
\item for each resolution $\Z_t \in \res(t)$, with $z_t = \corr{\Z_t}^{-1}(t)$, there is a resolution $\Z_s \in \res(s)$, with $z_s = \corr{\Z_s}^{-1}(s)$, s.t.\ for each $\alpha \in \Act^{\star}$ we have $\pr(\Cw(z_s,\alpha)) = \pr(\Cw(z_t,\alpha))$;
\end{itemize}
from which we can conclude that $s \WTr t$.
\end{proof}


\section{Logical characterization of trace metrics}
\label{sec:char_of_trace_metric}

In this section we present the logical characterization of strong and weak trace metric (resp.\ Theorem~\ref{thm:LL_metric_char} and Theorem~\ref{thm:LLw_metric_char}).
We define a suitable distance on formulae in $\LL$ (resp.\ $\LLw$) and we characterize the strong (resp.\ weak) trace metric between processes as the distance between the sets of formulae satisfied by them.


\subsection{$\LL$-characterization of strong trace metric}
\label{sec:char_metric_strong}

Firstly, we need to define a distance on trace formulae.

\begin{definition}
[Distance on $\LLt$]
\label{def:metric_on_LLt}
The function $\Dtrtd \colon \LLt \times \LLt \to [0,1]$ is defined over $\LLt$ as follows:
\[
\Dtrtd( \Phi_1, \Phi_2 ) = 
\begin{cases}
0 & \text{ if } \Phi_1 = \Phi_2 \\
1 & \text{ otherwise.} 
\end{cases}
\]
\end{definition}

\begin{proposition}
\label{prop:distance_on_LLt_is_metric}
The function $\Dtrtd$ is a $1$-bounded metric over $\LLt$.
\end{proposition}


\begin{proof}
The thesis follows by noticing that $\Dtrtd$ is the discrete metric over $\LLt$.
\end{proof}


To define a distance over trace distribution formulae we see them as probability distribution over trace formulae and we define the distance over $\LLd$ as the Kantorovich lifting of the metric $\Dtrtd$.

\begin{definition}
[Distance on $\LLd$]
\label{def:metric_on_LLd}
The function $\Dtrdd \colon \LLd \times \LLd \to [0,1]$ is defined over $\LLd$ as follows:
\[
\Dtrdd(\Psi_1, \Psi_2) = \Kantorovich(\Dtrtd)(\Psi_1, \Psi_2).
\]
\end{definition}

\begin{proposition}
\label{prop:distance_on_LLd_is_metric}
The function $\Dtrdd$ is a $1$-bounded metric over $\LLd$.
\end{proposition}


\begin{proof}
First we prove that $\Dtrdd$ is a metric over $\LLd$, namely that
\begin{enumerate}
\item \label{item:identity}
$\Dtrdd(\Psi_1,\Psi_2) = 0$ if{f} $\Psi_1 = \Psi_2$;
\item \label{item:symmetry}
$\Dtrdd(\Psi_1,\Psi_2) = \Dtrdd(\Psi_2,\Psi_1)$;
\item \label{item:triang_ineq}
$\Dtrdd(\Psi_1,\Psi_2) \le \Dtrdd(\Psi_1,\Psi_3) + \Dtrdd(\Psi_3,\Psi_2)$.
\end{enumerate}

\underline{\emph{Proof of item~\ref{item:identity}}}

($\Leftarrow$) Assume first that $\Psi_1 = \Psi_2$.
Then $\Dtrdd(\Psi_1,\Psi_2) = 0$ immediately follows from Definition~\ref{def:metric_on_LLd}, since the Kantorovich metric is a pseudometric.

($\Rightarrow$) Assume now that $\Dtrdd(\Psi_1,\Psi_2) = 0$. 
We aim to show that this implies that $\Psi_1 = \Psi_2$.
Assume wlog. that $\Psi_1 = \bigoplus_{i \in I} r_i \Phi_i$ and that $\Psi_2 = \bigoplus_{j \in J} r_j \Phi_j$.
Then we have
\begin{equation}
\label{eq:Kantorovich_0}
\Dtrdd(\Psi_1, \Psi_2) = \min_{\w \in \W(\Psi_1,\Psi_2)} \sum_{i \in I, j \in J} \w(\Phi_i,\Phi_j) \Dtrtd(\Phi_i,\Phi_j)
\end{equation}
and the distance in Equation~\eqref{eq:Kantorovich_0} is $0$ if, given the optimal matching $\bar{\w}$ 
\[
\bar{\w}(\Phi_i,\Phi_j) > 0 \text{ if{f} } \Dtrtd(\Phi_i,\Phi_j) = 0.
\]
By Proposition~\ref{prop:distance_on_LLt_is_metric} we have that $\Dtrtd(\Phi_i,\Phi_j) = 0$ if{f} $\Phi_i = \Phi_j$.
In particular, let $\Phi_{j_i}$ be any formula in $\{\Phi_j \mid j \in J\}$ s.t.\ $\Phi_i = \Phi_{j_i}$.
Since by Definition~\ref{def:logic_LL} the trace formulae $\Phi_i$ occurring in $\Psi_1$ are pairwise distinct and, analogously, the trace formulae $\Phi_j$ occurring in $\Psi_2$ are pairwise distinct, we gather that
\[
\begin{array}{rlr}
r_i = \sum_{j \in J} \bar{\w}(\Phi_i,\Phi_j) = \sum_{j_i \in J} \bar{\w}(\Phi_i,\Phi_{j_i}) = \bar{\w}(\Phi_i,\Phi_{j_i}) \\
r_j = \sum_{i \in I} \bar{\w}(\Phi_i,\Phi_j) = \sum_{i_j \in I} \bar{\w}(\Phi_{i_j},\Phi_j) = \bar{\w}(\Phi_{i_j},\Phi_j).
\end{array}
\]
Therefore we can infer that $\Psi_1 = \Psi_2$ as probability distributions over $\LLt$.


\underline{\emph{Proof of item~\ref{item:symmetry}}}
Immediate from the discrete metric and the matching being both symmetric.

\underline{\emph{Proof of item~\ref{item:triang_ineq}}}
Assume wlog. that $\Psi_1 = \bigoplus_{i \in I} r_i \Phi_i$, $\Psi_2 = \bigoplus_{j \in J} r_j \Phi_j$ and $\Psi_3 = \bigoplus_{h \in H} r_h \Phi_h$.

Let $\w_{1,3} \in \W(\Psi_1, \Psi_3)$ be an optimal matching for $\Psi_1, \Psi_3$, namely
\[
\Dtrdd(\Psi_1, \Psi_3) = \min_{\w \in \W(\Psi_1,\Psi_3)} \sum_{i \in I \atop h \in H} \w(\Phi_i,\Phi_h) \Dtrtd(\Phi_i,\Phi_h) = \sum_{i \in I \atop h \in H} \w_{1,3}(\Phi_i,\Phi_h) \Dtrtd(\Phi_i,\Phi_h)
\]
and let $\w_{2,3}\in \W(\Psi_2, \Psi_3)$ be an optimal matching for $\Psi_2, \Psi_3$, that is
\[
\Dtrdd(\Psi_2,\Psi_3) = \min_{\w \in \W(\Psi_2,\Psi_3)} \sum_{j \in J \atop h \in H} \w(\Phi_j,\Phi_h) \Dtrtd(\Phi_j,\Phi_h) = \sum_{j \in J \atop h \in H} \w_{2,3}(\Phi_j, \Phi_h) \Dtrtd(\Phi_j,\Phi_h).
\]
Consider now the function $f \colon I \times J \times H \to [0,1]$ defined by 
\[
f(i,j,h) = \w_{1,3}(\Phi_i, \Phi_h) \cdot \w_{2,3}(\Phi_j, \Phi_h) \cdot \frac{1}{r_h}.
\]
Then, we have $\sum_{j \in J} f(i,j,h) = \w_{1,3}(\Phi_i, \Phi_h)$ namely the projection of $f$ over the first and third components coincides with the optimal matching for $\Psi_1, \Psi_3$.
Similarly, $\sum_{i \in I} f(i,j,h) = \w_{2,3}(\Phi_j, \Phi_h)$ namely the projection of $f$ over the second and third components coincides with the optimal matching for $\Psi_2, \Psi_3$.
Moreover, it holds that $\sum_{j \in J,\, h \in H} f(i,j,h) = r_i$ and $\sum_{i \in I,\, h \in H} f(i,j,h) = r_j$, that is $f(i,j,h)$ is a matching in $\W(\Psi_1,\Psi_2)$.
Therefore,
\[
\begin{array}{rlr}
\Dtrdd(\Psi_1, \Psi_2) ={} & \min_{\w \in \W(\Psi_1,\Psi_2)} \sum_{i \in I,\, j \in J} \w(\Phi_i,\Phi_j) \Dtrtd(\Phi_i,\Phi_j)  & \text{(by definition)}\\
\le{} & \sum_{i \in I,\, j \in J,\, h \in H} f(i,j,h) \Dtrtd(\Phi_i,\Phi_j)  & \text{(by construction of $f$)}\\
\le{} & \sum_{i \in I,\, j \in J,\, h \in H} f(i,j,h) \big( \Dtrtd(\Phi_i, \Phi_h)\, 
+ \,\Dtrtd(\Phi_j,\Phi_h) \big) & \text{(since $\Dtrtd$ is a metric)}\\
={} & \sum_{i \in I,\, j \in J,\, h \in H} f(i,j,h) \Dtrtd(\Phi_i, \Phi_h) + \\
 & \sum_{i \in I,\, j \in J,\, h \in H} f(i,j,h) \Dtrtd(\Phi_j,\Phi_h) \\
={} & \sum_{i \in I,\, h \in H} \Big(\sum_{j \in J} f(i,j,h) \Big) \cdot \Dtrtd(\Phi_i, \Phi_h) + \\
&  \sum_{j \in J,\, h \in H} \Big( \sum_{i \in I} f(i,j,h) \Big) \cdot \Dtrtd(\Phi_j,\Phi_h) \\
={} & \sum_{i \in I,\, h \in H} \w_{1,3}(\Phi_i, \Phi_h) \Dtrtd(\Phi_i, \Phi_h) +\\
 & \sum_{j \in J,\, h \in H} \w_{2,3}(\Phi_j, \Phi_h) \Dtrtd(\Phi_j, \Phi_h) & \text{(by construction of $f$)}\\
={} & \Kantorovich(\Dtrtd)(\Psi_1, \Psi_3) + \Kantorovich(\Dtrtd)(\Psi_3, \Psi_2) & \text{(by definition of $\w_{1,3}, \w_{2,3}$)}
\\
={} & \Dtrdd(\Psi_1, \Psi_3) + \Dtrdd(\Psi_3, \Psi_2)& \text{(by definition)}.
\end{array}
\]

To conclude, we need to show that $\Dtrdd$ is $1$-bounded, namely that for each $\Psi_1,\Psi_2 \in \LLd$ we have $\Dtrdd(\Psi_1,\Psi_2) \le 1$.
Assume wlog that $\Psi_1 = \bigoplus_{i \in I} r_i \Phi_i$ and $\Psi_2 = \bigoplus_{j \in J} r_j \in \Phi_j$.
We have
\[
\begin{array}{rlr}
\Dtrdd(\Psi_1,\Psi_2) ={} & \min_{\w \in \W(\Psi_1,\Psi_2)} \sum_{i \in I, j \in J} \w(\Phi_i, \Phi_j) \Dtrtd(\Phi_i, \Phi_j) \\
\le & \sum_{i \in I, j \in J} \w(\Phi_i, \Phi_j) \Dtrtd(\Phi_i, \Phi_j) & \text{(for an arbitrary $\w$)} \\
\le & \sum_{i \in I, j \in J} \w(\Phi_i, \Phi_j) & \text{($\Dtrtd$ is either $1$ or $0$)} \\
={} & 1 & \text{($\w$ is probability distribution)}.
\end{array}
\]
\end{proof}


\begin{example}
\label{ex:Dtrdd}
Consider the formulae $\Psi_1 = 0.6 \diam{a}\diam{b}\top \oplus 0.4 \diam{a}\diam{c}\top$ and $\Psi_2 = 0.7 \diam{a}\diam{c}\top \oplus 0.3 \diam{a}\diam{b}\top$.
We have that 
\[
\begin{array}{rlr}
\Dtrdd(\Psi_1,\Psi_2) ={} & \min_{\w \in \W(\Psi_1,\Psi_2)} \sum_{\Phi \in \support(\Psi_1) \atop \Phi' \in \support(\Psi_2)} \w(\Phi,\Phi') \Dtrtd(\Phi,\Phi') \\
\le & 0.3 \Dtrtd(\diam{a}\diam{b}\top, \diam{a}\diam{b}\top) + 0.4 \Dtrtd(\diam{a}\diam{c}\top, \diam{a}\diam{c}\top) + 0.3 \Dtrtd(\diam{a}\diam{b}\top, \diam{a}\diam{c}\top) \\
={} & 0.3 \cdot 0 + 0.4 \cdot 0 + 0.3 \cdot 1 \\
={} & 0.3
\end{array}
\]
\end{example}

Next result derives from our characterization of trace distribution equivalence of resolutions (Theorem~\ref{thm:equivalent_resolutions}).

\begin{theorem}
\label{thm:kernel_of_Dtrdd}
The kernel of $\Dtrdd$ is trace distribution equivalence of resolutions.
\end{theorem}


\begin{proof}
Let $s,t \in \proc$ and consider $\Z_s \in \res(s)$, with $z_s = \corr{\Z_s}^{-1}(s)$, and $\Z_t \in \res(t)$, with $z_t = \corr{\Z_s}^{-1}(t)$.
By Theorem~\ref{thm:equivalent_resolutions} we have that $z_s \STr z_t$ if{f} $\Psi_{\Z_s} = \Psi_{\Z_t}$.
Since by Proposition~\ref{prop:distance_on_LLd_is_metric} $\Dtrdd$ is a metric on $\LLd$, we have that $\Dtrdd(\Psi_{\Z_s}, \Psi_{\Z_t}) = 0$ if{f} $\Psi_{\Z_s} = \Psi_{\Z_t}$.
Thus we can conclude that
\[
z_s \STr z_t \quad \text{ if{f} } \quad \Psi_{\Z_s} = \Psi_{\Z_t} \quad \text{ if{f} } \quad \Dtrdd(\Psi_{\Z_s}, \Psi_{\Z_t}) = 0.
\]
\end{proof}


We lift the distance over formulae to a distance over processes as the Hausdorff distance between the sets of formulae satisfied by them.

\begin{definition}
\label{def:LL_logical_distance}
The $\LL$-\emph{distance} over processes $\DLLd \colon \proc \times \proc \to [0,1]$ is defined, for all $s,t \in \proc$, by 
\[
\DLLd(s,t) = \Hausdorff(\Dtrdd)(\LL(s), \LL(t)).
\]
\end{definition}

\begin{proposition}
\label{prop:DLLd_is_metric}
The mapping $\DLLd$ is a $1$-bounded pseudometric over $\proc$.
\end{proposition}


\begin{proof}
First we show that $\DLLd$ is a pseudometric over $\proc$,
namely that for each $s,t,u \in \proc$
\begin{flalign}
& \label{eq:identity}
\DLLd(s,s) = 0 \\
& \label{eq:symmetry}
\DLLd(s,t) = \DLLd(t,s) \\
& \label{eq:triang_ineq}
\DLLd(s,t) \le \DLLd(s,u) + \DLLd(u,t)
\end{flalign}
Equation~\eqref{eq:identity} and Equation~\eqref{eq:symmetry} are immediate from the definition of $\DLLd$ (Definition~\ref{def:LL_logical_distance}).

Let us prove Equation~\eqref{eq:triang_ineq}.
Firstly, we notice that from the definition of Hausdorff distance we have
\[
\DLLd(s,t) = \max \{ \sup_{\Psi \in \LL(s)} \, \inf_{\Psi' \in \LL(t)} \, \Dtrdd(\Psi,\Psi'), \; \sup_{\Psi' \in \LL(t)} \, \inf_{\Psi \in \LL(s)} \, \Dtrdd(\Psi,\Psi') \}.
\] 
Thus, for all $s,t,u \in \proc$ we can infer that
\begin{flalign}
& \label{eq:Hausdorff_su}
\sup_{\Psi \in \LL(s)}\, \inf_{\Psi'' \in \LL(u)} \, \Dtrdd(\Psi, \Psi'') \le \DLLd(s,u)\\
& \label{eq:Hausdorff_ut}
\sup_{\Psi'' \in \LL(u)}\, \inf_{\Psi' \in \LL(t)} \, \Dtrdd(\Psi'', \Psi') \le \DLLd(u,t).
\end{flalign}
As a first step, we aim to show that
\begin{equation}
\label{eq:limited_sup_st}
\sup_{\Psi \in \LL(s)}\, \inf_{\Psi' \in \LL(t)} \, \Dtrdd(\Psi, \Psi') \le \DLLd(s,u) + \DLLd(u,t).
\end{equation}
For sake of simplicity, we index formulae in $\LL(s)$ by indexes in the set $J$, formulae in $\LL(t)$ by indexes in set $I$ and formulae in $\LL(u)$ by indexes in $H$.
By definition of infimum we have that for each $\varepsilon_1 > 0$
\begin{equation}
\label{eq:epsilon_su}
\text{for each } \Psi_j \in \LL(s) \text{ there is a } \Psi_{h_{j}} \in \LL(u) \text{ s.t. } \Dtrdd(\Psi_j, \Psi_{h_j}) < \inf_{\Psi_h \in \LL(u)} \,\Dtrdd(\Psi_j, \Psi_h) + \varepsilon_1
\end{equation}  
and analogously for each $\varepsilon_2 > 0$
\begin{equation}
\label{eq:epsilon_ut}
\text{for each } \Psi_h \in \LL(u) \text{ there is a } \Psi_{i_h} \in \LL(t) \text{ s.t. } \Dtrdd(\Psi_h, \Psi_{i_h}) < \inf_{\Psi_i \in \LL(t)} \, \Dtrdd(\Psi_h, \Psi_i) + \varepsilon_2.
\end{equation}  
In particular given $\Psi_j \in \LL(s)$ let $\Psi_{h_j} \in \LL(u)$ be the index realizing Equation~\eqref{eq:epsilon_su}, with respect to $\varepsilon_1$, and let $\Psi_{i_{h_j}} \in \LL(t)$ be the index realizing Equation~\eqref{eq:epsilon_ut} with respect to $\Psi_{h_j}$ and $\varepsilon_2$.
Then we have 
\[
\begin{array}{rlr}
& \Dtrdd(\Psi_j, \Psi_{i_{h_j}}) \\
\le & \Dtrdd(\Psi_j, \Psi_{h_j}) + \Dtrdd(\Psi_{h_j}, \Psi_{i_{h_j}}) \\
< &  \big( \inf_{\Psi_h \in \LL(u)} \, \Dtrdd(\Psi_j, \Psi_h) + \varepsilon_1 \big) + \big( \inf_{\Psi_i \in \LL(t)} \, \Dtrdd(\Psi_{h_j}, \Psi_i) + \varepsilon_2 \big) & \text{(Eq.~\ref{eq:epsilon_su},\ref{eq:epsilon_ut})}\\
\le &  \big( \sup_{\Psi_j \in \LL(s)}\, \inf_{\Psi_h \in \LL(u)} \, \Dtrdd(\Psi_j, \Psi_h) + \varepsilon_1 \big) + \big( \sup_{\Psi_h \in \LL(u)}\, \inf_{\Psi_i \in \LL(t)} \, \Dtrdd(\Psi_h, \Psi_i) + \varepsilon_2 \big)
\end{array}
\]
from which we gather
\[
\inf_{\Psi_i \in \LL(t)} \Dtrdd(\Psi_j, \Psi_i) \le \Dtrdd(\Psi_j, \Psi_{i_{h_j}}) < \sup_{\Psi_j \in \LL(s)}\, \inf_{\Psi_h \in \LL(u)} \, \Dtrdd(\Psi_j, \Psi_h) + \sup_{\Psi_h \in \LL(u)}\, \inf_{\Psi_i \in \LL(t)} \, \Dtrdd(\Psi_h, \Psi_i) + \varepsilon_1 + \varepsilon_2.
\]
Thus, since $j$ was arbitrary, we obtain
\[
\sup_{\Psi_j \in \LL(s)}\, \inf_{\Psi_i \in \LL(t)}\, \Dtrdd(\Psi_j, \Psi_i) \le \sup_{\Psi_j \in \LL(s)}\, \inf_{\Psi_h \in \LL(u)} \, \Dtrdd(\Psi_j, \Psi_h) + \sup_{\Psi_h \in \LL(u)}\, \inf_{\Psi_i \in \LL(t)} \, \Dtrdd(\Psi_h, \Psi_i) + \varepsilon_1 + \varepsilon_2
\]
and since this relation holds for any $\varepsilon_1$ and $\varepsilon_2$ we can conclude that
\[
\sup_{\Psi_j \in \LL(s)}\, \inf_{\Psi_i \in \LL(t)}\, \Dtrdd(\Psi_j, \Psi_i) \le \sup_{\Psi_j \in \LL(s)}\, \inf_{\Psi_h \in \LL(u)} \, \Dtrdd(\Psi_j, \Psi_h) +  \sup_{\Psi_h \in \LL(u)}\, \inf_{\Psi_i \in \LL(t)} \, \Dtrdd(\Psi_h, \Psi_i).
\]
Then, by the inequalities in Equation~\eqref{eq:Hausdorff_su} and Equation~\eqref{eq:Hausdorff_ut} we can conclude that
\[
\sup_{\Psi_j \in \LL(s)}\, \inf_{\Psi_i \in \LL(t)}\, \Dtrdd(\Psi_j,\Psi_i) \le \DLLd(s,u) + \DLLd(u, t)
\]
and thus Equation~\eqref{eq:limited_sup_st} holds.
Switching the roles of $s$ and $t$ in the steps above allows us to infer
\begin{equation}
\label{eq:limited_sup_ts}
\sup_{\Psi_i \in \LL(t)}\, \inf_{\Psi_j \in \LL(s)}\, \Dtrdd(\Psi_j,\Psi_i) \le \DLLd(s, u) + \DLLd(u,t).
\end{equation}
Finally, we have
\[
\begin{array}{rlr}
\DLLd(s, t) ={} & \max\{\sup_{\Psi_j \in \LL(s)}\, \inf_{\Psi_i \in \LL(t)}\, \Dtrdd(\Psi_j, \Psi_i), \sup_{\Psi_i \in \LL(t)}\, \inf_{\Psi_j \in \LL(s)}\, \Dtrdd(\Psi_j, \Psi_i)\} \\
\le{} & \DLLd(s,u) + \DLLd(u,t)
\end{array}
\]
where the last relation follows by Equations~\eqref{eq:limited_sup_st} and \eqref{eq:limited_sup_ts}.

To conclude, we need to show that $\DLLd$ is $1$-bounded.
We recall that by Proposition~\ref{prop:distance_on_LLd_is_metric}, $\Dtrdd$ is $1$-bounded.
We have
\[
\begin{array}{rlr}
\DLLd(s,t) ={} & \Hausdorff(\Dtrdd)(\LL(s), \LL(t)) \\
={} & \max \left\{ \sup_{\Psi_i \in \LL(s)}\, \inf_{\Psi_j \in \LL(t)}\, \Dtrdd(\Psi_i, \Psi_j),\; \sup_{\Psi_j \in \LL(t)}\, \inf_{\Psi_i \in \LL(s)}\, \Dtrdd(\Psi_i, \Psi_j) \right\} \\
\le & \max \{ 1,\,1 \} \\
={} & 1.
\end{array}
\]
\end{proof}


\begin{proposition}
\label{prop:LLs_closure}
Let $s \in \proc$.
The set $\LL(s)$ is a closed subset of $\LL$ wrt.\ the topology induced by $\Dtrdd$.
\end{proposition}


\begin{proof}
As we are working on a metric space, the proof obligation is equivalent to prove that each sequence in $\LL(s)$ that admits a limit converges in $\LL(s)$, namely
\begin{equation}
\label{eq:prop_closure_proof_obligation}
\text{for each } \{\Psi_n\}_{n \in \N} \subseteq \LL(s) \text{ s.t. there is } \Psi \in \LL \text{ with } \lim_{n \to \infty} \Psi_n = \Psi \text{ then } \Psi \in \LL(s).
\end{equation}
From Theorem~\ref{thm:LLs_is_cup_tracing_formula_resolution} we have that
$
\LL(s) = \{\top\} \cup \{\Psi_{\Z} \mid \Z \in \res(s) \}.
$
Since a finite union of closed sets is closed, the proof obligation Equation~\eqref{eq:prop_closure_proof_obligation} is equivalent to prove that 
\begin{flalign}
\label{eq:true}
& \{ \top \} \text{ is closed}\\
\label{eq:diam}
& \{\Psi_{\Z} \mid \Z \in \res(s)\} \text{ is closed}
\end{flalign}

Equation~\eqref{eq:true} is immediate since the only sequence in $\{\top\}$ admitting a limit is the constant sequence $\Psi_n = \top$ for all $n \in \N$.

Let us deal now with Equation~\eqref{eq:diam}.
First of all, we notice that sequences in $\{\Psi_{\Z} \mid \Z \in \res(s)\}$ can be written in the general form
\[
\Psi_n = \bigoplus_{i \in I_n} r_i^{(n)} \Phi_i^{(n)}
\]
with $\{\bigoplus_{i \in I_n} r_i^{(n)}\Phi_i^{(n)}\}_{n \in \N} \subseteq \LL(s) \setminus \{\top\}$.

Assume that there is a trace distribution formula $\Psi \in \LLd$ s.t.\ $\lim_{n \to \infty} \Psi_n = \Psi$.
We aim to show that $\Psi \in \LL(s)$, namely that
\begin{equation}
\label{eq:seq_diam}
\Psi = \Psi_{\Z} \text{ for some } \Z \in \res(s).
\end{equation}
In what follows, we assume wlog that limit trace distribution formula $\Psi$ has the form $\Psi = \bigoplus_{j \in J} r_j \Phi_j$.

From $\{\bigoplus_{i \in I_n} r_i^{(n)} \Phi_i^{(n)}\}_{n \in \N} \subseteq \LL(s) \setminus \{\top\}$ we gather that for each $n \in \N$ there is a resolution $\Z_n \in \res(s)$ s.t.\ $\Psi_{\Z_n} = \bigoplus_{i \in I_n} r_i^{(n)} \Phi_i^{(n)}$.
For each $n \in \N$, let $z_n = \corr{\Z_n}^{-1}(s)$.
Then $\Psi_{\Z_n} = \bigoplus_{i \in I_n} r_i^{(n)} \Phi_i^{(n)}$ implies that $\displaystyle I_n = \tr(\C_{\max}(z_n))$, namely $I_n$ is the set of traces to which the maximal computations of the process $z_n$ are compatible.
Hence, for each $i \in I_n$ we have that $\Phi_i^{(n)}$ is the tracing formula of the trace indexed by $i$ and $\displaystyle r_i^{(n)} = \pr(\C_{\max}(z_n,i))$.

We notice that
\[
\begin{array}{rlr}
& \lim_{n \to \infty} \Psi_n = \Psi \\
\text{if{f} } & \lim_{n \to \infty} \Dtrdd(\Psi_n, \Psi) = 0 \\
\text{if{f} } & \lim_{n \to \infty} \Kantorovich(\Dtrtd)(\Psi_n, \Psi) = 0
\end{array}
\]
that is if{f} the sequence $\{\Psi_n\}_{n \in \N}$ converges to $\Psi$ with respect to the Kantorovich metric.
Since we are considering distributions with finite support, the convergence with respect to the Kantorovich metric is equivalent to the weak convergence of probability distributions (also called convergence in distribution) which states that $\lim_{n \to \infty}\Psi_n(\Phi) = \Psi(\Phi)$ for each continuity point $\Phi \in \LLt$ of $\Psi$.
Since the probability distribution over trace formuale $\Psi$ is discrete and with finite support, its continuity points are the trace formulae which are not in its support.
Hence, we have that $\lim_{n \to \infty}\Psi_n(\Phi) = 0$ for each $\Phi \not\in \{\Phi_j \mid j \in J\}$.
More specifically, we obtain that $\lim_{n \to \infty} I_n = J$ which gives that if there is an index $\tilde{i}$ s.t.\ $\lim_{n \to \infty}\Phi_{\tilde{i}}^{(n)} \not \in \{\Phi_j \mid j \in J\}$, or if $\{\Phi_{\tilde{i}}^{(n)}\}_{n \in N}$ has no limit, then $\lim_{n \to \infty} r_{\tilde{i}}^{(n)} = 0$.
Furthermore, since $\Dtrtd$ is the discrete metric over $\LLt$, we have that a sequence of trace formulae $\{\Phi^{(n)}\}_{n \in \N}$ converges to $\Phi$ if{f} the sequence is definitively constant, namely if{f} there is an $N \in \N$ s.t.\ $\Phi^{(n)} = \Phi$ for all $n \ge N$.
Therefore, from $\lim_{n \to \infty} I_n = J$ we can infer that there is an $N \in \N$ s.t.\ $I_n = J$ for all $n \ge N$.
Consequently, by construction of the sets $I_n$, we obtain that $J = \tr(\C_{\max}(z_N))$ thus giving that, for each $j \in J$, $\Phi_j$ is the tracing formula of the trace indexed by $j$.
Since moreover we are considering image-finite processes, for each $j \in J$ $\pr(\C_{\max}(z_n,j))$ assumes only a finite number of values wrt\ $n \ge N$.
Therefore, we can infer that there is an $M \ge N \in \N$ s.t.\ for each $j \in J$ we have $\pr(\C_{\max}(z_n,j)) = \pr(\C_{\max}(z_m,j))$ for all $n,m \ge M$.
Thus, from Definition~\ref{def:mimicking_formula_for_traces}, we infer that the resolution $\Z_M \in \res(s)$, namely the resolution whose mimicking formula corresponds to the $M$-th trace distribution formula in the sequence $\{\Psi_n\}_{n \in \N}$, is s.t.\ $\Psi = \Psi_{\Z_M}$, thus proving Equation~\eqref{eq:seq_diam} and concluding the proof.
\end{proof}


From our $\LL$-characterization of strong trace equivalence (Theorem~\ref{thm:det_char}) we obtain the following result.

\begin{theorem}
\label{thm:kernel_of_DLLd}
The kernel of $\DLLd$ is trace equivalence.
\end{theorem}


\begin{proof}
($\Rightarrow$)
Assume first that $s \STr t$.
We aim to show that this implies that $\DLLd(s,t) = 0$.
By Theorem~\ref{thm:det_char} we have that $s \STr t$ implies that $\LL(s) = \LL(t)$ from which we gather
\[
\DLLd(s,t) = \Hausdorff(\Dtrdd)(\LL(s), \LL(t)) = 0.
\]
($\Leftarrow$)
Assume now that $\DLLd(s,t) = 0$.
We aim to show that this implies that $s \STr t$.
Since $\LL(s)$ and $\LL(t)$ are closed by Proposition~\ref{prop:LLs_closure} and since $\DLLd$ is a pseudometric by Proposition~\ref{prop:DLLd_is_metric}, from $\DLLd(s,t) = 0$ we can infer that $\LL(s) = \LL(t)$.
By Theorem~\ref{thm:det_char} we can conclude that $s \STr t$. 
\end{proof}


Finally, we obtain the characterization of the strong trace metric.

\begin{theorem}
[Characterization of strong trace metric]
\label{thm:LL_metric_char}
For all $s,t \in \proc$ we have $\TraceMetric(s,t) = \DLLd(s,t)$.
\end{theorem}


\begin{proof}
By definition of trace metric (Definition~\ref{def:trace_metric}) we have that
\[
\TraceMetric(s,t) ={} \max \left\{ \sup_{\Z_s \in \res(s)}\, \inf_{\Z_t \in \res(t)} \, D_T(\Z_s,\Z_t), \; \sup_{\Z_t \in \res(t)} \, \inf_{\Z_s \in \res(s)} D_T (\Z_s, \Z_t) \right\}.
\]
By definition of $\LL$-distance over processes (Definition~\ref{def:LL_logical_distance}) we have that
\[
\begin{array}{rlr}
\DLLd(s,t) ={} & \Hausdorff(\Dtrdd)(\LL(s), \LL(t)) \\
={} & \Hausdorff(\Dtrdd)(\{\top\} \cup \{\Psi_{\Z_s} \mid \Z_s \in \res(s)\}, \{\top\} \cup \{\Psi_{\Z_t} \mid \Z_t \in \res(t)\}) \\
={} & \Hausdorff(\Dtrdd)(\{\Psi_{\Z_s} \mid \Z_s \in \res(s)\}, \{\Psi_{\Z_t} \mid \Z_t \in \res(t)\}) \\
={} & \max \left\{ \sup_{\Z_s \in \res(s)}\, \inf_{\Z_t \in \res(t)} \, \Dtrdd(\Psi_{\Z_s}, \Psi_{\Z_t}), \; \sup_{\Z_t \in \res(t)} \, \inf_{\Z_s \in \res(s)} \Dtrdd (\Psi_{\Z_s}, \Psi_{\Z_t}) \right\}
\end{array}
\]
where the third equality follows from the fact that by Def.~\ref{def:metric_on_LLd} we have $\Dtrdd(\top,\top) = 0$ and $\Dtrdd(\top, \Psi) = 1$ for any $\Psi \neq \top$.
Thus we have that $\top = \argmin_{\Psi \in \{\top\} \cup \{\Psi_{\Z_t} \mid \Z_t \in \res(t)\}} \Dtrdd(\top, \Psi)$ and symmetrically $\top = \argmin_{\Psi \in \{\top\} \cup \{\Psi_{\Z_s} \mid \Z_s \in \res(s)\}} \Dtrdd(\Psi, \top)$.
Moreover, for any $\Psi \neq \top$ we have that $\Dtrdd(\Psi, \Psi') \le \Dtrdd(\Psi, \top)$ for any $\Psi' \in \{\Psi_{\Z_t} \mid \Z_t \in \res(t)\}$ and $\Dtrdd(\Psi'', \Psi) \le \Dtrdd(\top, \Psi)$ for any $\Psi'' \in \{\Psi_{\Z_s} \mid \Z_s \in \res(s)\}$.

Hence, to prove the thesis it is enough to show that 
\begin{equation}
\label{eq:thm_LL_metric_char_proof_obligation}
D_T(\Z_s, \Z_t) = \Dtrdd(\Psi_{\Z_s}, \Psi_{\Z_t})
\text{ for all } \Z_s \in \res(s), \Z_t \in \res(t).
\end{equation}
Let $\Z_s \in \res(s)$, with $z_s = \corr{\Z_s}^{-1}(s)$, and $\Z_t \in \res(t)$, with $z_t = \corr{\Z_t}^{-1}(t)$.
Then by definition of mimicking formula (Definition~\ref{def:mimicking_formula_for_traces}) we have 
\[
\Psi_{\Z_s} = \bigoplus_{\alpha \in \tr(\C_{\max}(z_s))} \pr(\C_{\max}(z_s,\alpha)) \Phi_{\alpha}
\]
where for each $\alpha \in \tr(\C_{\max}(z_s))$ we have that $\Phi_{\alpha}$ is the tracing formula for the trace $\alpha$.
Similarly,
\[
\Psi_{\Z_t} = \bigoplus_{\beta \in \tr(\C_{\max}(z_t))} \pr(\C_{\max}(z_t,\beta)) \Phi_{\beta}
\]
where for each $\beta \in \tr(\C_{\max}(z_t))$ we have that $\Phi_{\beta}$ is the tracing formula for the trace $\beta$.

By definition of trace distance between resolutions (Definition~\ref{def:trace_metric_det_res}) we have that 
\begin{equation}
\label{eq:DTZsZt}
D_T(\Z_s, \Z_t) = \min_{\w \in \W(\TD_{\Z_s}, \TD_{\Z_t})} \sum_{\alpha \in \tr(\C_{\max}(z_s)), \beta \in \tr(\C_{\max}(z_t))} \w(\alpha,\beta) d_T(\alpha,\beta)
\end{equation}
where, by definition of trace distance between traces (Definition~\ref{def:trace_metric_traces}), we have that $d_t (\alpha, \beta) = 0$ if $\alpha = \beta$ and $d_t(\alpha,\beta) = 1$ otherwise.

Hence, by definition of tracing formula (Definition~\ref{def:tracing_formula}), we have that for all $\alpha \in \tr(\C_{\max}(z_s)), \beta \in \tr(\C_{\max}(z_t))$ we have
$
d_T (\alpha, \beta) = \Dtrtd(\Phi_{\alpha}, \Phi_{\beta})
$,
thus giving
\begin{equation}
\label{eq:DTZsZt2}
\eqref{eq:DTZsZt} = \min_{\w \in \W(\TD_{\Z_s}, \TD_{\Z_t})} \sum_{\alpha \in \tr(\C_{\max}(z_s)), \beta \in \tr(\C_{\max}(z_t))} \w(\alpha,\beta) \Dtrtd(\Phi_{\alpha}, \Phi_{\beta}).
\end{equation}
Let $\bar{\w}$ be an optimal matching for $D_T(\Z_s, \Z_t)$, namely
\begin{equation}
\label{eq:DTZsZt3}
\eqref{eq:DTZsZt2} = \sum_{\alpha \in \tr(\C_{\max}(z_s)), \beta \in \tr(\C_{\max}(z_t))} \bar{\w}(\alpha, \beta) \Dtrtd(\Phi_{\alpha}, \Phi_{\beta}).
\end{equation}
Then, by definition of matching and of $\TD_{\_}$ (Definition~\ref{def:trace_distribution}) we have that for any $\alpha \in \tr(\C_{\max}(z_s)), \beta \in \tr(\C_{\max}(z_t))$
\[
\begin{array}{rlr}
\pr(\C_{\max}(z_s,\alpha)) = \TD_{\Z_s}(\alpha) = \sum_{\beta \in \tr(\C_{\max}(z_t))} \bar{\w}(\alpha,\beta) \\
\pr(\C_{\max}(z_t,\beta)) = \TD_{\Z_t}(\beta) = \sum_{\alpha \in \tr(\C_{\max}(z_s))} \bar{\w}(\alpha,\beta).
\end{array}
\]

Therefore we have obtained that $\bar{\w}$ is a matching for $\Psi_{\Z_s}$ and $\Psi_{\Z_t}$.
In particular we notice that $\bar{\w}$ is actually an optimal matching for $\Psi_{\Z_s}, \Psi_{\Z_t}$.
This follows from the optimality of $\bar{\w}$ for $\TD_{\Z_s}, \TD_{\Z_t}$.
In fact each matching for $\Psi_{\Z_s}, \Psi_{\Z_t}$ can be constructed from a matching for $\TD_{\Z_s}, \TD_{\Z_t}$ using the same technique proposed above.
Moreover, given $\w_1 \in \W(\TD_{\Z_s}, \TD_{\Z_t})$ and $\w_2$ being the matching for $\Psi_1,\Psi_2$ built from it, the reasoning above guarantees that
\[
\sum_{\alpha \in \tr(\C_{\max}(z_s)), \beta \in \tr(\C_{\max}(z_t))} \w_1(\alpha, \beta) d_T(\alpha, \beta) 
= 
\sum_{\alpha \in \tr(\C_{\max}(z_s)) \atop \beta \in \tr(\C_{\max}(z_t))} \w_2(\alpha, \beta) \Dtrtd(\Phi_{\alpha}, \Phi_{\beta}).
\]
$\bar{\w}$ being optimal for $D_T$ implies $\tilde{\w}$ being optimal for $\Dtrdd$.
Hence by Definition~\ref{def:metric_on_LLd} we have
\[
\Dtrdd(\Psi_{\Z_s}, \Psi_{\Z_t}) = \sum_{\alpha \in \tr(\C_{\max}(z_s)), \beta \in \tr(\C_{\max}(z_t))} \bar{\w}(\alpha, \beta) \Dtrtd(\Phi_{\alpha}, \Phi_{\beta}).
\]
From Equation~\eqref{eq:DTZsZt3} we infer
$
D_T(\Z_s, \Z_t) = \Dtrdd(\Psi_{\Z_s}, \Psi_{\Z_t})
$
thus proving Equation~\eqref{eq:thm_LL_metric_char_proof_obligation} and concluding the proof.
\end{proof}


\subsection{$\LLw$-characterization of weak trace metric}
\label{sec:char_metric_weak}

The idea behind the definition of a metric on $\LLw$ is pretty much the same to the strong case.
The main difference is that the distance on $\LLw$ is a pseudometric whose kernel is given by $\LLw$-equivalence.

\begin{definition}
[Distance on $\LLwt$]
\label{def:metric_on_LLwt}
The function $\Dtrtdw \colon \LLwt \times \LLwt \to [0,1]$ is defined over $\LLwt$ as follows:
\[
\Dtrtdw( \Phi_1, \Phi_2 ) = 
\begin{cases}
0 & \text{ if } \Phi_1 \eqtrace \Phi_2 \\
1 & \text{ otherwise.}
\end{cases}
\]
\end{definition}

Clearly, $\Dtrtdw$ is a pseudometric on $\LLwt$ whose kernel is given by equivalence of trace formulae and we can lift it to a pseudometric over $\LLwd$ via the Kantorovich lifting functional.

\begin{definition}
[Distance on $\LLwd$]
\label{def:metric_on_LLwd}
The function $\Dtrddw \colon \LLwd \times \LLwd \to [0,1]$ is defined over $\LLwd$ as follows:
\[
\Dtrddw(\Psi_1, \Psi_2) = \Kantorovich(\Dtrtdw)(\Psi_1, \Psi_2).
\]
\end{definition}

\begin{proposition}
\label{prop:distance_on_LLwd_is_metric}
The function $\Dtrddw$ is a $1$-bounded pseudometric over $\LLwd$.
\end{proposition}


\begin{proof}
The same arguments used in the proof of Proposition~\ref{prop:distance_on_LLd_is_metric} apply, where in place of item~\ref{item:identity} we simply need to show that $\Dtrddw(\Psi,\Psi) = 0$, which is immediate from the definition through the Kantorovich pseudometric.
\end{proof}


\begin{theorem}
\label{thm:kernel_of_Dtrddw}
The kernel of $\Dtrddw$ is $\LLw$-equivalence of trace distribution formulae.
\end{theorem}


\begin{proof}
($\Rightarrow$)
Assume first that $\Dtrddw(\Psi_1,\Psi_2) = 0$ for $\Psi_1 = \bigoplus_{i \in I} r_i \Phi_i$ and $\Psi_2 = \bigoplus_{j \in J} r_j \Phi_j$. 
We aim to show that this implies $\Psi_1 \eqtrace^{\dagger} \Psi_2$.
From the assumption, we have
\[
\begin{array}{rlr}
0 ={} & \Dtrddw(\bigoplus_{i \in I} r_i \Phi_i, \bigoplus_{j \in J} r_j \Phi_j)\\
={} & \min_{\w \in \W(\Psi_1, \Psi_2)} \sum_{i \in I,\, j \in J} \w(\Phi_i, \Phi_j) \Dtrtdw(\Phi_i ,\Phi_j) \\
={} & \sum_{i \in I,\, j \in J} \w(\Phi_i,\Phi_j) \Dtrtdw(\Phi_i,\Phi_j) & \text{(for $\w$ optimal matching)}.
\end{array}
\]
Thus, for each $i \in I$ and $j \in J$ we can distinguish two cases:
\begin{itemize}
\item either $\w(\Phi_i, \Phi_j)=0$,
\item or $\w(\Phi_i, \Phi_j) >0$, implying $\Dtrtdw(\Phi_i, \Phi_j)=0$, which is equivalent to say that $\Phi_i \eqtrace \Phi_j$ by Definition~\ref{def:metric_on_LLwt}.
\end{itemize}
For each $i \in I$, let $J_i \subseteq J$ be the set of indexes $j_i$ for which $\w(\Phi_i, \Phi_{j_i})>0$ and, symmetrically, for each $j \in J$ let $I_j \subseteq I$ be the set of indexes $i_j$ for which $\w(\Phi_{i_j}, \Phi_j)>0$.
So we have
\[
\begin{array}{rlr}
\Psi_1 ={} & \bigoplus_{i \in I} r_i \Phi_i \\
={} & \bigoplus_{i \in I} \big(\sum_{j \in J} \w(\Phi_i, \Phi_j)\big) \Phi_i & \text{($\w \in \W(\Psi_1,\Psi_2)$)} \\
\eqtrace^{\dagger} & \bigoplus_{i \in I} \big(\sum_{j_i \in J_i} \w(\Phi_i, \Phi_{j_j}) \big) \Phi_i  & \text{(by construction of each $J_i$)}\\
\eqtrace^{\dagger} & \bigoplus_{i \in I,\, j_i \in J_i} \w(\Phi_i, \Phi_{j_i}) \Phi_{j_i} & \text{($\Phi_i \eqtrace \Phi_{j_i}$ for each $j_i \in J_i$)} \\
\eqtrace^{\dagger} & \bigoplus_{i \in I,\, j_i \in J_i,\, i'_{j_i} \in I_{j_i}} \w(\Phi_{i'_{j_i}}, \Phi_{j_i}) \Phi_{i'_{j_i}} & \text{($\Phi_{i'_{j_i}} \eqtrace \Phi_{j_i}$ for each $i'_{j_i} \in I_{j_i}$)} \\
\eqtrace^{\dagger} & \bigoplus_{i_j \in I_j,\, j \in J} \w(\Phi_{i_j}, \Phi_j) \Phi_{i_j} & \text{(all indexes $j \in J$ are involved)}\\
\eqtrace^{\dagger} & \bigoplus_{j \in J} \big( \sum_{i_j \in I_j} \w(\Phi_{i_j}, \Phi_j) \big) \Phi_j & \text{($\Phi_j \eqtrace \Phi_{i_j}$ for each $i_j \in I_j$)} \\
\eqtrace^{\dagger} & \bigoplus_{j \in J} \big( \sum_{i \in I} \w(\Phi_i, \Phi_j) \big) \Phi_j & \text{(by construction of each $I_j$)}\\
={} & \bigoplus_{j \in J} r_j \Phi_j & \text{($\w \in \W(\Psi_1,\Psi_2)$)} \\
={} & \Psi_2.
\end{array}
\]

$(\Leftarrow)$.
Assume that $\Psi_1 \eqtrace^{\dagger} \Psi_2$. 
We aim to show that $\Dtrddw(\Psi_1, \Psi_2) = 0$.
Assume wlog. that $\Psi_1 = \bigoplus_{i \in I} r_i \Phi_i$.
By definition of $\eqtrace$ (Definition~\ref{def:eqtrace}) and definition of lifting of a relation (Definition~\ref{def:lifting_relation}), from $\Psi_2 \eqtrace^{\dagger} \bigoplus_{i \in I} r_i \Phi_i$ we gather $\Psi_2 = \bigoplus_{i \in I \atop j_i \in J_i} r_{j_i}\Phi_{j_i}$ with $\sum_{j_i \in J_i} r_{j_i} = r_i$ and $\Phi_{j_i} \eqtrace \Phi_i$ for all $j_i \in J_i, i \in I$.
Then
\[
\begin{array}{rlr}
\Dtrddw(\Psi_1,\Psi_2) ={} & \Dtrddw(\bigoplus_{i \in I} r_i \Phi_i, \bigoplus_{i \in I,\, j_i \in J_i} r_{j_i} \Phi_{j_i}) \\
={} & \min_{\w \in \W(\Psi_1, \Psi_2)} \sum_{i \in I,\, j_h \in J_h \atop h \in I} \w(\Phi_i, \Phi_{j_h}) \Dtrtdw(\Phi_i ,\Phi_{j_h}) \\
\le{} & \sum_{i \in I,\, j_h \in J_h \atop h \in I} \tilde{\w}(\Phi_i, \Phi_{j_h}) \Dtrtdw(\Phi_i ,\Phi_{j_h}) \\
={} & \sum_{i \in I,\, j_i \in J_i} r_{j_i} \Dtrtdw(\Phi_i ,\Phi_{j_i}) \\
={} & 0 & \text{($\Phi_i \eqtrace \Phi_{j_i}$ for each $j_i \in J_i$)}
\end{array}
\]
where the inequality follows by observing that function $\tilde{\w}$ defined by $\tilde{\w}(\Phi_i, \Phi_{j_h}) = r_{j_i}$ if $h = i$ and $\tilde{\w}(\Phi_i, \Phi_{j_h}) = 0$ otherwise, is a matching in $\W(\Psi_1, \Psi_2)$.
\end{proof}


\begin{corollary}
\label{cor:kernel_of_Dtrddw}
$\Z_1,\Z_2 \in \res(\proc)$ are weak trace distribution equivalent if{f} $\Dtrddw(\Psi_{\Z_1}, \Psi_{\Z_2}) = 0$.
\end{corollary}


\begin{proof}
($\Rightarrow$)
Assume first that $\Z_1$ and $\Z_2$ are weak trace distribution equivalent.
Then from Theorem~\ref{thm:equivalent_resolutions_weak} we infer that $\Psi_{\Z_1} \eqtrace \Psi_{\Z_2}$.
By Theorem~\ref{thm:kernel_of_Dtrddw} this implies $\Dtrddw(\Psi_{\Z_1}, \Psi_{\Z_2}) = 0$.

($\Leftarrow$)
Assume now that $\Dtrddw(\Psi_{\Z_1},\Psi_{\Z_2}) = 0$.
Then from Theorem~\ref{thm:kernel_of_Dtrddw} we infer that $\Psi_{\Z_1} \eqtrace \Psi_{\Z_2}$.
By Theorem~\ref{thm:equivalent_resolutions_weak} this implies that $\Z_1$ and $\Z_2$ are weak trace distribution equivalent. 
\end{proof}


By the Hausdorff functional we lift the pseudometric $\Dtrddw$ to a pseudometric over processes.

\begin{definition}
\label{def:LLw_logical_distance}
The $\LLw$-\emph{distance} over processes $\DLLdw \colon \proc \times \proc \to [0,1]$ is defined, for all $s,t \in \proc$, by 
\[
\DLLdw(s,t) = \Hausdorff(\Dtrddw)(\LLw(s), \LLw(t)).
\]
\end{definition}

\begin{proposition}
\label{prop:DLLdw_is_metric}
The mapping $\DLLdw$ is a $1$-bounded pseudometric over $\proc$.
\end{proposition}


\begin{proof}
The same arguments used in the proof of Proposition~\ref{prop:DLLd_is_metric} apply.
\end{proof}


\begin{proposition}
\label{prop:LLsw_closure}
Let $s \in \proc$.
The set $\LLw(s)$ is a closed subset of $\LLw$ wrt.\ the topology induced by $\Dtrddw$.
\end{proposition}


\begin{proof}
Since $(\LLwd,\Dtrddw)$ is a pseudometric space (Proposition~\ref{prop:distance_on_LLwd_is_metric} and Theorem~\ref{thm:kernel_of_Dtrddw}), to prove the thesis we need to show that the quotient space $\LLw(s)_{/\eqtrace}$ is a closed subset of $\LLw_{/\eqtrace}$ with respect to the topology induced by $\Dtrddw$ (in fact $(\LLwd_{/\eqtrace}, \Dtrddw)$ is a metric space).
From Remark~\ref{rmk:equivalence_is_equality} we have that $\LLwd_{/\eqtrace} = \LLd$ and $\LLw(s)_{/\eqtrace} = \LL(s)$. 
Moreover, we have that $\Dtrddw\mid_{\LLwd_{/\eqtrace}} = \Dtrdd$.
Hence, the same arguments used in the proof of Proposition~\ref{prop:LLs_closure} allow us to prove that $\LLw(s)_{/\eqtrace}$ is a closed subset of $\LLw_{/\eqtrace}$ wrt. the topology induced by $\Dtrddw$.
This gives the result also for $\LLw(s)$ wrt to $\LLw$ and $\Dtrddw$.
\end{proof}


\begin{theorem}
\label{thm:kernel_of_DLLdw}
The kernel of $\DLLdw$ is weak trace equivalence.
\end{theorem}


\begin{proof}
($\Rightarrow$)
Assume that $s \WTr t$.
We aim to show that $\DLLdw(s,t) = 0$.
By Theorem~\ref{thm:weak_char} we have that $s \WTr t$ implies that $\LLw(s) \eqtrace^{\dagger} \LLw(t)$.
Since the kernel of $\Dtrddw$ is given by $\eqtrace^{\dagger}$ (Theorem~\ref{thm:kernel_of_Dtrddw}), we can infer
\[
\DLLdw(s,t) = \Hausdorff(\Dtrtdw)(\LLw(s), \LLw(t)) = 0.
\]
($\Leftarrow$)
Assume now that $\DLLdw(s,t) = 0$.
We aim to show that this implies that $s \WTr t$.
Since 
\begin{inparaenum}[(i)]
\item $\LLw(s)$ and $\LLw(t)$ are closed by Proposition~\ref{prop:LLsw_closure}, 
\item $\DLLdw$ is a pseudometric by Proposition~\ref{prop:DLLdw_is_metric} and
\item the kernel of $\Dtrddw$ is $\eqtrace^{\dagger}$ by Theorem~\ref{thm:kernel_of_Dtrddw},
\end{inparaenum} 
from $\DLLdw(s,t) = 0$ we can infer $\LLw(s) \eqtrace^{\dagger} \LLw(t)$.
Then, by Theorem~\ref{thm:weak_char} we can conclude $s \WTr t$. 
\end{proof}


Finally, we obtain the characterization of the weak trace metric.

\begin{theorem}
[Characterization of weak trace metric]
\label{thm:LLw_metric_char}
For all $s,t \in \proc$ we have $\wTraceMetric(s,t) = \DLLdw(s,t)$.
\end{theorem}


\begin{proof}
The same arguments used in the proof of Thm~\ref{thm:LL_metric_char} apply.
\end{proof}


\section{From boolean to real semantics}
\label{sec:boom}

In this section we focus on $\LL$ and we exploit the distance between formulae to define a real valued semantics for it, namely given a process $s$ we assign to each formula a value in $[0,1]$ expressing the probability that $s$ satisfies it.
Then we show that our logical characterization of trace metric can be restated in terms of the general schema $\displaystyle\TraceMetric(s,t) = \sup_{\Psi \in \LLd} \mid \val{\Psi}{s} - \val{\Psi}{t} \mid$ 
where $\val{\Psi}{s}$ denotes the value of the formula $\Psi$ at process $s$, accordingly to the new real valued semantics.
We remark that although, due to space restrictions, we present only the result for $\LL$, the technique we propose would lead to the same results when applied to $\LLw$.

First of all, we recall the notion of \emph{distance function}, namely the distance between a point and a set.

\begin{definition}
[Distance function]
\label{def:distance_function}
Let $\LL' \subseteq \LLd$.
Given any $\Psi \in \LLd$ we denote by $\Dtrdd(\Psi, \LL')$ the \emph{distance between $\Psi$ and the set} $\LL'$ defined by
$ \displaystyle
\Dtrdd(\Psi, \LL') = \inf_{\Psi' \in \LL'} \Dtrdd(\Psi, \Psi').
$
\end{definition}

Then we obtain the following characterization of the Hausdorff distance.

\begin{proposition}
\label{prop:char_Hausdorff}
Let $\LL_1, \LL_2 \subseteq \LLd$.
Then it holds that
$ 
\displaystyle \Hausdorff(\Dtrdd)(\LL_1,\LL_2) = \sup_{\Psi \in \LLd} | \Dtrdd(\Psi, \LL_1) - \Dtrdd(\Psi, \LL_2) |.
$
\end{proposition}


\begin{proof}
It is clear that
\begin{equation}
\label{eq:gen_Hausdorff}
\Hausdorff(\Dtrdd)(\LL_1,\LL_2) = \max \left\{ \sup_{\Psi_1 \in \LL_1} \Dtrdd(\Psi_1,\LL_2),\; \sup_{\Psi_2 \in \LL_2} \Dtrdd(\Psi_2, \LL_1) \right\}.
\end{equation}
Firstly we show that 
\begin{equation}
\label{eq:hausdorff_le}
\Hausdorff(\Dtrdd)(\LL_1,\LL_2) \le \sup_{\Psi \in \LLd} | \Dtrdd(\Psi, \LL_1) - \Dtrdd(\Psi, \LL_2) |.
\end{equation}
Without loss of generality, we can assume that $\Hausdorff(\Dtrdd)(\LL_1, \LL_2) = \sup_{\Psi_1 \in \LL_1} \Dtrdd(\Psi_1,\LL_2)$.
Then we have
\[
\begin{array}{rlr}
\sup_{\Psi_1 \in \LL_1} \Dtrdd(\Psi_1,\LL_2) ={} & \sup_{\Psi_1 \in \LL_1} | \Dtrdd(\Psi_1,\LL_2) - \Dtrdd(\Psi_1, \LL_1) | \\
\le & \sup_{\Psi \in \LLd} |\Dtrdd(\Psi, \LL_2) - \Dtrdd(\Psi, \LL_1)|
\end{array}
\]
from which Equation~\eqref{eq:hausdorff_le} holds.

Next, we aim to show the converse inequality, namely
\begin{equation}
\label{eq:hausdorff_ge}
\Hausdorff(\Dtrdd)(\LL_1,\LL_2) \ge \sup_{\Psi \in \LLd} | \Dtrdd(\Psi, \LL_1) - \Dtrdd(\Psi, \LL_2) |.
\end{equation}
To this aim, we show that
\begin{equation}
\label{eq:hausdorff_ge_1}
\text{for each } \Psi \in \LLd \text{ it holds }| \Dtrdd(\Psi, \LL_1) - \Dtrdd(\Psi, \LL_2) | \le \Hausdorff(\Dtrdd)(\LL_1,\LL_2) .
\end{equation}
\begin{itemize}
\item Assume $\Psi \in \LL_1$. 
Then $\Dtrdd(\Psi, \LL_1) = 0$ so that $|\Dtrdd(\Psi, \LL_1) - \Dtrdd(\Psi, \LL_2)| = \Dtrdd(\Psi, \LL_2)$.
Moreover
\[
\Dtrdd(\Psi, \LL_2) \le \sup_{\Psi_1 \in \LL_1} \Dtrdd(\Psi_1, \LL_2) \le \Hausdorff(\Dtrdd)(\LL_1,\LL_2)
\]
and Equation~\eqref{eq:hausdorff_ge_1} follows in this case.

\item The case of $\Psi \in \LL_2$ is analogous and therefore Equation~\eqref{eq:hausdorff_ge_1} follows also in this case.

\item Finally, assume that $\Psi \not\in \LL_1 \cup \LL_2$.
Without loss of generality, we can assume that $\Dtrdd(\Psi, \LL_1) \ge \Dtrdd(\Psi, \LL_2)$.
By definition of \emph{infimum} it holds that for each $\epsilon > 0$ there is a formula $\Psi_{\epsilon} \in \LL_2$ s.t.\ 
\begin{equation}
\label{eq:inf_limit}
\Dtrdd(\Psi, \Psi_{\epsilon}) < \Dtrdd(\Psi, \LL_2) + \epsilon.
\end{equation}
Analogously, for each $\epsilon' > 0$ and for each $\Psi_2 \in \LL_2$ there is a $\Psi_{\epsilon'} \in \LL_1$ s.t.\ 
\begin{equation}
\label{eq:inf_limit_1}
\Dtrdd(\Psi_2, \Psi_{\epsilon'}) < \Dtrdd(\Psi_2, \LL_1) + \epsilon'.
\end{equation}
Let us fix $\epsilon, \epsilon' > 0$.
Then let $\Psi_{\epsilon} \in \LL_2$ be the formula realizing Equation~\eqref{eq:inf_limit}, with respect to $\Psi$, and let $\tilde{\Psi_{\epsilon'}}$ be the formula in $\LL_1$ realizing Equation~\eqref{eq:inf_limit}, with respect to this $\Psi_{\epsilon}$.
Therefore, we have
\[
\begin{array}{rlr}
& |\Dtrdd(\Psi, \LL_1) - \Dtrdd(\Psi, \LL_2)|\\
= & \Dtrdd(\Psi, \LL_1) - \Dtrdd(\Psi, \LL_2)\\
< & \Dtrdd(\Psi, \LL_1) - \Dtrdd(\Psi, \Psi_{\epsilon}) + \epsilon & \text{(by Equation~\eqref{eq:inf_limit})} \\
= & \inf_{\Psi_1 \in \LL_1} \Dtrdd(\Psi, \Psi_1) - \Dtrdd(\Psi, \Psi_{\epsilon}) + \epsilon & \text{(by Definition~\ref{def:distance_function})} \\
< & \Dtrdd(\Psi, \tilde{\Psi_{\epsilon'}}) - \Dtrdd(\Psi, \Psi_{\epsilon}) + \epsilon \\
\le & \Dtrdd(\Psi, \Psi_{\epsilon}) + \Dtrdd(\Psi_{\epsilon}, \tilde{\Psi_{\epsilon'}}) - \Dtrdd(\Psi, \Psi_{\epsilon}) + \epsilon & \text{(by triangle inequality)} \\
= & \Dtrdd(\Psi_{\epsilon}, \tilde{\Psi_{\epsilon'}}) + \epsilon \\
<  & \Dtrdd(\Psi_{\epsilon}, \LL_1) + \epsilon' + \epsilon & \text{(by Equation~\eqref{eq:inf_limit_1})}\\
\le & \sup_{\Psi_2 \in \LL_2} \Dtrdd(\Psi_2, \LL_1) + \epsilon' + \epsilon \\
\le & \Hausdorff(\Dtrdd)(\LL_1,\LL_2) + \epsilon' + \epsilon & \text{(by Equation~\eqref{eq:gen_Hausdorff}).}
\end{array}
\]
Summarizing, we have obtained that
\[
|\Dtrdd(\Psi, \LL_1) - \Dtrdd(\Psi, \LL_2)| < \Hausdorff(\Dtrdd)(\LL_1,\LL_2) + \epsilon' + \epsilon 
\]
and since this inequality holds for each $\epsilon$ and $\epsilon'$, we can conclude that Equation~\eqref{eq:hausdorff_ge_1} holds.
\end{itemize}
Equation~\eqref{eq:hausdorff_le} and Equation~\eqref{eq:hausdorff_ge} taken together prove the thesis.
\end{proof}


To define the real-valued semantics of $\LLd$ we exploit the distance $\Dtrdd$.
Informally, to quantify how much the formula $\Psi$ is satisfied by process $s$ we evaluate first how far $\Psi$ is from being satisfied by $s$.
This corresponds to the minimal distance between $\Psi$ and a formula satisfied by $s$, namely to $\Dtrdd(\Psi, \LL(s))$.
Then we simply notice that, as our distances are all $1$-bounded, being $\Dtrdd(\Psi, \LL(s))$ far from $s$ is equivalent to be $1 - \Dtrdd(\Psi, \LL(s))$ close to it.
Thus we assign to $\Psi$ the real value $1 - \Dtrdd(\Psi, \LL(s))$ in $s$.

\begin{definition}
[Real-valued semantics of $\LLd$]  
\label{def:real_valued_semantics}
We define the \emph{real-valued semantics} of $\LLd$ as the function $\val{\_}{\_} \colon \LLd \times \proc \to [0,1]$ defined for all $\Psi \in \LLd$ and $s \in \proc$ as
$
\val{\Psi}{s} = 1 - \Dtrdd(\Psi, \LL(s)).
$
\end{definition}

We can restate our characterization theorem (Theorem~\ref{thm:det_char}) as a probabilistic $\LLd$-model checking problem.

\begin{theorem}
[Characterization of strong trace metric II]
\label{thm:det_char_2}
For all $s,t \in \proc$ we have 
\[
\TraceMetric(s,t) = \sup_{\Psi \in \LLd} \mid \val{\Psi}{s} - \val{\Psi}{t} \mid.
\]
\end{theorem}


\begin{proof}
From Theorem~\ref{thm:det_char} we have $\TraceMetric(s,t) = \DLLd(s,t)$.
Hence the thesis is equivalent to prove
\[
\DLLd(s,t) = \sup_{\Psi \in \LLd} \mid \val{\Psi}{s} - \val{\Psi}{t} \mid.
\]
We have
\[
\begin{array}{rlr}
\DLLd(s,t) ={} & \Hausdorff(\Dtrdd)(\LL(s),\LL(t)) & \text{(by Definition~\ref{def:LL_logical_distance})} \\
={} & \sup_{\Psi \in \LLd} \mid \Dtrdd(\Psi, \LL(s)) - \Dtrdd(\Psi, \LL(t)) \mid & \text{(by Proposition~\ref{prop:char_Hausdorff})} \\
={} & \sup_{\Psi \in \LLd} \mid \Dtrdd(\Psi, \LL(s)) - \Dtrdd(\Psi, \LL(t)) +1 -1 \mid \\
={} & \sup_{\Psi \in \LLd} \mid 1 - \Dtrdd(\Psi, \LL(t)) - \big( 1 - \Dtrdd(\Psi, \LL(s)) \big) \mid \\
={} & \sup_{\Psi \in \LLd} \mid \val{\Psi}{t} - \val{\Psi}{s} \mid & \text{(by Definition\ref{def:real_valued_semantics})}.
\end{array}
\]
\end{proof}


\section{Concluding remarks}
\label{sec:conclusions_chap7}

We have provided a logical characterization of the strong and weak variants of trace metric on finite processes in the PTS model.
Our results are based on the definition of a \emph{distance} over the two-sorted boolean logics $\LL$ and $\LLw$, which we have proved to characterize resp.\ strong and weak probabilistic trace equivalence by exploiting the notion of \emph{mimicking formula} of a resolution. 

Our distance is a $1$-bounded pseudometric that quantifies the syntactic disparities of the formulae and we have proved that the trace metric corresponds to the distance between the sets of formulae satisfied by the two processes.
This approach, already successfully applied in \cite{CGT16a} to the characterization of the bisimilarity metric, is not standard.
Logical characterizations of the trace metrics have been obtained in terms of the probabilistic $L$-model checking problem, where $L$ is the class of logical properties of interest, \cite{BBLM15,DHKP16,AFS09}.
However we have proved that our approach can be exploited to regain classical one: by means of our distance between formulae we have defined a real-valued semantics for $\LL$, namely a probabilistic model checking of a formula in a process, and then we have proved that the trace metric constitutes the least upper bound to the error that can be observed in the verification of an $\LL$ formula.

Another interesting feature of our approach is its generality, since it can be easily applied to some variants of the trace equivalence and trace metric.
In \cite{S95tr,BdNL14} the authors distinguish between resolutions obtained via deterministic schedulers and the ones obtained via randomized schedulers.
The only difference between the two classes is in the evaluation of the probability weights: in deterministic resolutions, which are the ones we have considered in this paper, each possible resolution of nondeterminism is considered singularly and thus the target probability distributions of their transitions are the same as in the considered process.
In randomized resolutions, internal nondeterminism is solved by assigning a probability weight to each choice and thus the target distributions are obtained from the convex combination of the target distributions of the considered process.
Since the definition of the mimicking formulae depends solely on the values of the probability weights in the resolutions and not on how these weights are evaluated, our characterization can be applied also to the case of trace equivalences and metrics defined in terms of randomized resolutions.

As a first step in the future development of our work, we aim to extend our results to the trace equivalence defined in \cite{BdNL14} which, differently from the equivalence of \cite{S95tr} considered in this paper, is compositional wrt.\ the parallel composition operator.
Roughly speaking, in \cite{BdNL14} for each given trace it is checked whether the resolutions of two processes assign the same probability it, whereas in \cite{S95tr} for a chosen resolution of the first process we check whether there is a resolution for the second process that assigns the same probability to all traces.
Furthermore, no trace metric has been defined yet for the equivalence in \cite{BdNL14}.
Our idea is then firstly to define such a trace metric and secondly to simplify the logic $\LL$ by substituting the trace distribution formulae with a simple test on the execution probability of a trace, with an operator similar to the probabilistic operator in \cite{PS07}.
By applying our approach to the new logic we will obtain the characterization of the trace equivalence and metric.

Then, 
we will study metrics and logical characterizations for the testing equivalences defined in \cite{BdNL14}.

Further, in \cite{BBLM15} a sequence of Kantorovich bisimilarity-like metrics converging to the trace metric on MCs is provided.
Hence we aim to combine our characterization results in \cite{CGT16a} with the ones in this paper in order to see if a similar result of convergence can be obtained also with our technique on PTSs.

Finally, it would be interesting to apply the SOS-based decomposition method proposed in \cite{CGT16b} to $\LL$ (resp.\ $\LLw$) in order to derive congruence formats for the probabilistic strong (resp.\ weak) trace equivalence from its logical characterization.
We also aim to extend this technique in order to derive compositional properties, as \emph{uniform continuity} \cite{GT15}, of strong and weak trace metric.

\bibliographystyle{eptcs}
\bibliography{qapl17_bib}

\end{document}